\documentclass[12pt]{article}

\usepackage{amsmath,amssymb,amsthm,amscd}
\usepackage{graphicx}

\usepackage{hyperref}

\makeatletter

\@addtoreset{equation}{section}
\makeatother

\newcommand{\beq}{\begin{equation}}
\newcommand{\eeq}{\end{equation}}

\def\R{\mathbb{R}}
\def\C{\mathbb{C}}
\def\Z{\mathbb{Z}}

\newcommand{\bC}{\ensuremath{\mathbb{C}}}

\newcommand{\bR}{\ensuremath{\mathbb{R}}}

\newcommand{\bZ}{\ensuremath{\mathbb{Z}}}

\newcommand{\scA}{\ensuremath{\mathcal{A}}}

\newcommand{\scM}{\ensuremath{\mathcal{M}}}
\newcommand{\scN}{\ensuremath{\mathcal{N}}}

\newtheorem{theorem}{Theorem}[section]
\newtheorem{lemma}[theorem]{Lemma}
\newtheorem{proposition}[theorem]{Proposition}
\newtheorem{conjecture}[theorem]{Conjecture}

\theoremstyle{definition}
\newtheorem{remark}[theorem]{Remark}

\newcommand{\lk}{\operatorname{lk}}

\newcommand{\Tr}{\operatorname{Tr}}
\renewcommand{\Re}{\operatorname{Re}}

\def\k{k}
\def\n{n}
\def\Hopf{\textrm{Hopf}}

\begin{document}

\baselineskip=18pt  
\baselineskip 0.7cm

\begin{titlepage}

\setcounter{page}{0}

\renewcommand{\thefootnote}{\fnsymbol{footnote}}

\begin{flushright}
\end{flushright}

\vskip 0.5cm

\begin{center}
{\LARGE \bf
Topological Strings,}
{\center}
{ \LARGE \bf D-Model,}
{\center}
{\LARGE \bf and }
{\center}
{\LARGE \bf Knot Contact Homology}

\vskip 1.5cm

{\large
Mina Aganagic$^{1,2}$, Tobias Ekholm$^{3,4}$, Lenhard Ng$^5$\\ and\\ Cumrun Vafa$^6$
\\
\medskip
}


{\it
$^1$Center for Theoretical Physics, University of California, Berkeley,  USA\\
$^2$Department of Mathematics, University of California, Berkeley,  USA\\
$^3$Department of Mathematics, Uppsala University,  Uppsala, Sweden\\
$^4$Institut Mittag-Leffler, Djursholm, Sweden\\
$^5$Department of Mathematics, Duke University, Durham, USA\\
$^6$Jefferson Physical Laboratory, Harvard University, Cambridge, USA\\
}

\end{center}

\centerline{{\bf Abstract}}
\medskip
\noindent
We study the connection between topological strings and contact homology
recently proposed in the context of knot invariants. In particular, we establish the proposed relation between the Gromov-Witten disk
amplitudes of a Lagrangian associated to a knot and augmentations of its contact homology algebra. This also implies the equality between the $Q$-deformed $A$-polynomial and the augmentation polynomial of knot contact homology (in the irreducible case).
We also generalize this relation to the case of links and to higher rank representations for knots.
The generalization involves a study of the quantum moduli space of special Lagrangian branes with higher Betti numbers probing the Calabi-Yau. This leads to
an extension of SYZ, and a new notion of mirror symmetry, involving higher dimensional mirrors. The mirror theory is a topological string, related to D-modules, which we call the ``D-model.'' In the present setting, the mirror manifold is the augmentation variety of the link. Connecting further to contact geometry, we study intersection properties of branches of the augmentation variety guided by the relation to D-modules. This study leads us to propose concrete geometric constructions of Lagrangian fillings for links.
We also relate the augmentation variety
with the large $N$ limit of the colored HOMFLY, which we conjecture to be related to a $Q$-deformation
of the extension of $A$-polynomials associated with the link complement.

\end{titlepage}
\setcounter{page}{1} 

\section{Introduction}
The subject of topological strings has been important in both physics and mathematics. From the
physical perspective it has been studied from many different viewpoints and plays a prominent role
in the understanding of string dualities. From the mathematical perspective, its underpinnings (Gromov-Witten invariants) are well understood and play a central role in symplectic geometry and related subjects. Moreover, in many cases one can directly compute the amplitudes using either physical methods or mathematical techniques and the two agree.

An interesting application of topological strings involves the study of knot and link invariants.
In particular it is known that HOMFLY polynomials for links $K$ can be reformulated in terms
of topological strings on $T^*S^3$ where one includes $N$ Lagrangian D-branes wrapping
$S^3$ and some Lagrangian brane $L_K$ intersecting $S^3$ along the link and asymptotic to the conormal of $K$ at infinity. In this setup, it is possible to shift $L_K$ off of the $0$-section $S^3\subset T^{\ast}S^{3}$, which suggests that the data of the link is imprinted in how $L_K$ intersects the boundary of $T^*S^3$ at infinity. This ideal boundary is the unit cotangent bundle $U^{\ast}S^{3}$, which is a contact manifold (topologically $S^3\times S^2$) with contact 1-form given by the restriction of the action form $p\,dq$.  The intersection of $L_K$
with $U^{\ast}S^{3}$ is a Legendrian torus $\Lambda_K$.

Large $N$ transition relates this theory to a corresponding theory in the geometry where $S^3$ shrinks to zero size and is blown up with blow-up parameter given by $Ng_{s}$, where $g_{s}$ is the string coupling constant. Under the transition, boundary conditions corresponding to strings ending on branes wrapping $S^{3}$ close up and disappear, but the transition does not affect the Legendrian torus $\Lambda_K$, which is part of the geometry infinitely far away from the tip of the singular cone where $S^3$ shrinks and gets replaced by $S^2$.

Contact homology is a theory that was introduced in a setting corresponding exactly to the above mentioned  geometry at infinity. More precisely, one studies the
geometry of $\R\times\Lambda_{K} \subset \R\times U^{\ast} S^{3}$, which we can naturally view
as the far distance geometry of $L_K$, using holomorphic disks. In this particular case of conormal tori $\Lambda_K$ in the unit cotangent bundle of $S^{3}$, the resulting theory is known as knot contact homology. From the viewpoint of string theory, this theory involves a Hilbert space of physical open string states corresponding to Reeb chords (flow segments of the Reeb vector field with endpoints on $\Lambda_K$) and a BRST-operator deformed by holomorphic disks in $\R\times U^{\ast}S^{3}$ with boundary on $\R\times\Lambda_{K}$ and with punctures mapping to Reeb chords at infinity. In this language, knot contact homology is the corresponding BRST-cohomology. From a mathematical perspective this structure can be assembled into a differential graded algebra (DGA) which is well studied and in particular can be computed from a braid presentation of a link through a concrete combinatorial description of all relevant moduli spaces of holomorphic disks.

In view of the above discussion it is natural to ask how this DGA is related with the open version of Gromov-Witten theory of $L_{K}$. In \cite{AV} a partial answer for this was proposed: it was conjectured that the augmentation variety associated to a knot, which is a curve that parametrizes one-dimensional representations of the DGA associated to the knot, is the same as the moduli space of a single Lagrangian  $L_K$ corrected by disk instantons. This moduli space is in turn, according to a generalization of SYZ conjecture, related to the mirror curve, which is the locus of singular fibers in the bundle of conics that constitutes the mirror, and as such determines the mirror.\footnote{The A-polynomial curves and their deformations were studied in a context of (generalizing) the volume conjecture in \cite{Gukov:2003na}, and its relation to topological string was studied in \cite{2009ForPh, Dijkgraaf:2010ur, Gukov:2011qp, Borot:2012cw, Fuji:2012nx}. Mirror symmetry and torus knots were studied in  \cite{marino} and also in \cite{Jockers:2012pz}.}

One aim of this paper is to indicate a path to a mathematical proof of this conjecture: in Section \ref{sec:qchmaps} we show how to get a local parametrization of the augmentation curve in terms of the open Gromov-Witten potential of a Lagrangian filling of the conormal torus of a knot, see Section \ref{s:GWdef} for the definition, and that expression agrees with the local parametrization of the mirror curve derived using physical arguments, see Section \ref{sec:SYZmirror}. In particular, for knots with irreducible augmentation variety, using the conormal filling we obtain a proof of the conjecture, see Remark \ref{rmk:proofA=A}.

A further aim is to generalize this from one-component knots to many-component links, and the generalization turns out to involve interesting new ingredients. We consider new Lagrangian branes that have the same asymptotics as $L_K$ but have a topologically different filling. Similar new Lagrangian branes in fact appear already in the case of knots, where different $1$-cycles of the Legendrian torus bound and can be shrunk in the different branes.
For links with $n>1$ components, there are further possibilities with Lagrangian branes of various numbers of components. The maximal number of components is $n$, in which case no two components of the conormal tori of the link at infinity are connected through the brane. However, there are Lagrangian fillings and corresponding branes that connect some of the components at infinity. For example, there is a single component Lagrangian brane which connects them all, see Section \ref{ssec:antibrane} for a conjectural geometric construction.

Our study of the quantum moduli space of all these Lagrangians leads
to a new reformulation of mirror symmetry: the moduli space of branes for an $n$-component link is $n$-dimensional and the mirror geometry, instead of being a $1$-dimensional curve, will be an $n$-dimensional variety. It turns out that this variety can be viewed as a Lagrangian variety $V\subset (\C^{*})^{2n}$ in a canonical way and we identify it with the augmentation variety from knot contact homology.

An unusual feature of this structure is that we encounter higher-dimensional geometries as the mirror.
In fact, it may appear that one cannot formulate topological string in this context since the critical dimension for that theory is effectively 1 for non-compact Calabi-Yau 3-folds but for $n$-component links we have effective dimension $n>1$.  Nevertheless we propose a string theory mirror even in these cases, based on what we call the topological ``D-model''. The D-model is the A-model topological string on $({\mathbb C^{*}})^{2n}$ with a Lagrangian D-brane $V$ and a canonical coisotropic brane filling the whole  $({\mathbb C^{*}})^{2n}$.  In this context the topological string is exact at the level of the annulus, with the analog of higher genus A-model corrections being captured not by higher genera of strings, but by the choice of a flux turned on along the coisotropic brane.  The D-model leads to a natural definition of the theory in terms of D-modules (hence the name), while for $n=1$ (in particular for knots), the D-model is already known to be equivalent to the B-model.

The organization of this paper is as follows. In Section \ref{Sec:review} we review the relation between topological strings and Chern-Simons theory, large $N$ transition, and knot invariants. Furthermore, we describe, using a generalized SYZ formulation, how any knot gives a mirror geometry.  In Section \ref{Sec:CSandfillings} we introduce new Lagrangians associated to knots with the usual asymptotics at infinity but with different interior topology. We then generalize the discussion from knots to $n$-component links and show how $n$-dimensional Lagrangian varieties in flat space $(\C^*)^{2n}$ arise in the the study of the moduli space of Lagrangian branes filling the link conormal.  In Section \ref{Sec:knotch}, we introduce basic elements of knot contact homology; furthermore, we
relate augmentation varieties with the (disk instanton corrected) moduli spaces of Lagrangians associated to knots and links, and study intersection properties of branches of the augmentation variety guided by the D-model. In Section \ref{sec:augmentations}, we present a mathematical argument relating knot contact homology for links to disk amplitudes in Gromov-Witten theory, and study geometric constructions of Lagrangian fillings for conormals as well as properties and applications of the resulting Lagrangians.  In Section \ref{sec:ex} we present some examples. In Section \ref{Sec:Dmodel}, we formulate a conjecture
about how to quantize higher dimensional augmentation varieties in terms of the D-model, by relating them to D-modules. Finally, Appendices \ref{app:Hopf} and \ref{app:toruslink} contain calculations of a more technical nature.

\section{Review}\label{Sec:review}

Consider Chern-Simons gauge theory with gauge group $G=SU(N)$ on a closed $3$-manifold $M$ at level $k$, where $k$ is a positive integer. The Chern-Simons partition function is given by the path integral
$$
Z_{CS}(S^{3})= \int {\cal D} {A} \;e^{\,\frac{ik}{4\pi}\,S_{CS}(A)}
$$
over the space of connections $A$ with values in the Lie algebra of the gauge group, where
$$
S_{CS}(A) = \int_{M} {\rm Tr}\Bigl(A\wedge dA  + \tfrac{2}{3} A\wedge A \wedge A\Bigr)
$$
is the Chern-Simons action of the connection $A$. The path integral is independent of metric on $M$ and hence gives a topological invariant of $3$-manifolds. Submanifolds of dimension $0$ in $M$, i.e.~points, carry no topological information but submanifolds of dimension $1$, i.e.~knots and links, do and in Chern-Simons theory there are corresponding topologically invariant observables. More precisely, we associate a Wilson loop observable in representation $R$ to a knot $K$ by inserting the path ordered exponential
$$
\Tr_R U(K),
$$
where $U(K)=P \exp i\oint_{K} A$ is the holonomy of $A$ along $K$, in the path integral.
In fact, to define this we also have to choose a framing of the knot, i.e.~a non-vanishing vector field in the normal bundle of the knot $K\subset M$.
The case of many-component links is similar: a link $K$ is a collection
$$
K=K_1\cup K_2\cup \ldots K_n
$$
of disjoint knots $K_{j}$ in $M$. We specify a framing of each knot component of the link and representations
$$
R_1, \; R_2, \ldots, R_n
$$
coloring $K_1,\dots, K_n$, respectively. The corresponding link invariant is then the expectation value
$$
\left\langle \Tr_{R_1}U(K_1) \ldots \Tr_{R_n}U(K_n)\right\rangle=
Z_{R_1, \ldots, R_n}(M;K_1,\ldots,K_n)/Z(M)
$$
obtained by computing the Chern-Simons path integral with insertion of link observables:
$$
Z_{R_1, \ldots, R_n}(K_1,\ldots,K_n)= \int {\cal D} {A} \;e^{{i k \over 4 \pi}S_{CS}(A)}\;
\Tr_{R_1}U(K_1)\ldots \Tr_{R_n}U(K_n),
$$
and normalizing it with the path integral in the vacuum.

In \cite{Witten:1988hf}, Witten explained how to solve the above theory exactly. Any three dimensional topological theory corresponds to a two dimensional rational conformal field theory (CFT). The Hilbert space of the three dimensional theory and operators acting on it can be constructed from conformal blocks of the CFT and from representations of the corresponding modular group. In the Chern-Simons case, the relevant conformal field theory is the $SU(N)_k$ Wess-Zumino-Witten (WZW) model, and one finds that knowledge of the $S$, $T$ and braiding matrices is all that is needed to solve the theory on any $3$-manifold.

In this way, invariants of knots and links in the $3$-sphere $S^{3}$ that arise from Chern-Simons theory can be explicitly computed. In particular, the polynomial knot invariants considered earlier by Jones correspond to the gauge group $G=SU(2)$ and Wilson lines in the fundamental representation. More generally, for a link $K\subset S^{3}$ with knot components $K=K_{1}\cup\dots\cup K_{n}$, the expectation values
$$
\left\langle \Tr_{R_1} U(K_1)\ldots \Tr_{R_n}U(K_n)\right\rangle= H_{R_1, \ldots, R_n}( K_1,\ldots,K_n)(q,Q)
$$
are polynomials in the variables
$$
q = e^{2\pi i \over k+N},\qquad Q=q^N
$$
with integer coefficients that are independent of both $k$ and $N$. These polynomials are known as HOMFLY polynomials and were constructed from a mathematical point of view in \cite{HOMFLY}.

\subsection{Chern-Simons theory and topological string}\label{ssec:CStopstring}
Chern-Simons theory on $S^3$ with gauge group $SU(N)$ is intimately related to the open topological A-model, or Gromov-Witten theory, on
$$
X= T^*S^3
$$
with $N$ Lagrangian branes on the zero section $S^{3}\subset T^{\ast}S^{3}$ as follows. The topological A-model corresponds to counting holomorphic maps with Lagrangian boundary conditions. In $T^*S^3$, any holomorphic map with boundary on the zero section has vanishing area and is therefore constant. Thus, all maps that contribute to the A-model partition function $Z_{GW}(X)$ are degenerate and it was shown in \cite{Witten:1992fb} that their contributions are exactly captured by the Feynman diagrams of $SU(N)$ Chern-Simons theory on $S^3$. Consequently, the partition functions of the topological A-model on $X$ equals the Chern-Simons partition function on $S^3$:
$$
Z_{GW}(X)=Z_{CS}(S^3),
$$
where, as mentioned above, $Z_{GW}$ localizes on the $0$-dimensional space of holomorphic maps and is thus  given by the (exponentiated) generating function
$$
Z_{GW}(X) = \exp\left(\sum_{g,h\ge 0} F_{g,h}\;N^h g_s^{2g-2+h}\right),
$$
where ${F}_{g,h}$ captures the contribution of maps of connected genus $g$ Riemann surfaces to $X$ with $h$ boundary components mapping to $S^3\subset X$. Here, each boundary component is weighted by a factor $N$ corresponding to the choice of which of the $N$ D-branes wrapping $S^3$ that it lands on, and
the genus counting parameter (or string coupling constant) of the open topological string, $g_s$, equals the effective value of Chern-Simons coupling constant:
$$
g_s  = \frac{2 \pi i}{k+N}.
$$
From the perspective of Chern-Simons perturbation theory, the numbers $F_{g,h}$ arise by organizing the Feynman graphs in the following way: thicken the graphs into ribbon graphs with gauge index labels on the boundary; the number $F_{g,h}$ then keeps track of the contributions from the graphs that give rise to ribbon graphs corresponding to a genus $g$ Riemann surface with $h$  boundary components.
We also point out that the parameter $q$, in terms of which the Chern-Simons knot invariants become polynomial, is given by
$$
q = e^{\,g_s}.
$$

Knots and links can be included in the correspondence between Chern-Simons and the topological A-model in the following way \cite{OV}.
To each knot $K$ in $S^3$, we associate a Lagrangian $L_K$ in $X$, which we take to be its conormal in $T^*S^3$ consisting of all covectors along the knot that annihilate its tangent vector. In particular,
intersecting $L_{K}$ with the zero section, we get the knot itself,
$$
L_K \cap S^3 = K.
$$
For $n$-component links $K=K_{1}\cup\dots \cup K_{n}$ in $S^3$, we consider the Lagrangian $L_{K}$ which is the union of the conormals of its components
$$
L_{K}=L_{K_1} \cup \ldots \cup L_{K_n}.
$$
We will consider D-branes on $L_{K_j}$ and therefore need to include a sector in the theory that corresponds to worldsheets with boundaries both in the branes on $L_{K}$ and in the branes on the zero section $S^3$. We write the partition function of the topological string on $X$ with these branes present as
$$
Z_{GW}(X, L_{K_1},\ldots,L_{K_n})
$$
and note that, in addition to depending on $g_s$ and $N$, it also depends on the moduli of the Lagrangians $L_{K_j}$ which in particular keeps track of the class in $H_1(L_{K_j})$ represented by the boundaries of the worldsheets.

In the case under consideration, $L_{K_j}$ each has the topology of $S^{1}\times\R^2$ and we get one modulus $x_j$ for each Lagrangian $L_{K_j}$. Here, the complex parameter $x_j$ can be written as $x_j=r_j + i \oint_{S^1} A_j$, where the real part $r_j$ can be viewed as coming from the moduli of the deformations of the Lagrangian and the imaginary part is the holonomy of the $U(1)$ gauge field of the brane around the nontrivial cycle $S^1\times\{0\}\subset L_{K_j}$. In later sections, we will study similar Lagrangians for $K$ of first Betti number $b_1=r>1$, in which case the moduli of the Lagrangian is $r$-dimensional. Giving $x_j$ a nonzero value corresponds to lifting the Lagrangian $L_{K_j}$ off of the $S^3$.

From the Chern-Simons perspective, assuming there is a single D-brane on $L_{K}$ for a knot $K$, computing the partition function $Z_{GW}(X;L_{K})$ corresponds to inserting the operator
\beq\label{bif}
\det(1 - e^{-x} \, U(K))^{-1}.
\eeq
This describes the effect of integrating out the bifundamental strings, with one boundary on $S^3$ and the other on $L_K$. To relate this to knot invariants, we formally expand the determinant:
$$
\det(1\otimes 1 - e^{-x} \otimes U(K))^{-1}\ = \sum_{S_{k}} \Tr_{S_k} U(K) \;e^{-kx},
$$
where the sum ranges over all totally symmetric representations $S_k$ of $SU(N)$ of rank $k$.
Thus, computing in $SU(N)$ Chern-Simons theory on $S^3$ the following weighted sum of expectation values,
$$
\Psi_{K_1, \ldots, K_n} (x_1,\ldots,x_n) =\sum_{{k_1},\ldots {k_n} } \left\langle \Tr_{S_{k_1}}U(K_1)\ldots  \Tr_{S_{k_n}}U(K_n)\right\rangle \;e^{-k_1 x_1-\ldots -k_n x_n},
$$
gives the topological string partition function on $X$ with single branes on $L_{K_1},\ldots,L_{K_n}$ and $N$ branes on $S^3$.
In what follows, we will denote the HOMFLY polynomials
$$
\left\langle \Tr_{S_{k_1}}U(K_1)\ldots  \Tr_{S_{k_n}}U(K_n)\right\rangle =
H_{S_{k_1}\ldots S_{k_n}}(K_1, \ldots , K_n)
$$
simply by
$$
H_{S_{k_1}\ldots S_{k_n}}(K_1, \ldots , K_n)=
H_{{k_1}\ldots {k_n}}(K_1, \ldots , K_n).
$$
Since we isolated the part of the topological string amplitude on $X$ with some boundary component on the Lagrangian branes on $L_{K_1} \cup \ldots \cup L_{K_n}$, we get the following equation explicitly relating HOMFLY to topological string partition functions:
$$
\Psi_{K_1 \ldots K_n} (x_1,\ldots,x_n)=Z_{GW}(X, L_{K_1},\ldots,L_{K_n})/Z_{GW}(X).
$$

\subsection{Higher representations and multiple branes}
\label{ssec:multiplebranes}
In Section \ref{ssec:CStopstring}, we considered knots and links with a single brane on the conormal of each component. It is natural to ask what happens if we instead insert several branes on the conormals.
As it turns out, this reduces to a special case of single brane insertions. We explain this in the case of a single component knot $K$; the case of many component links is then an immediate generalization.

Consider  $X=T^*S^3$ with $\n>1$ branes on the Lagrangian conormal $L_K$ of a knot  $K$. Let $V$
be the $\n\times \n$ matrix of holonomies on the branes, with eigenvalues $(e^{x_1}, \ldots, e^{x_\n})$. The topological string partition function of the $\n$ branes
$$
\Psi_K^{[\n]}(x_1,\ldots,x_\n).
$$
has a contribution from worldsheets with at least one boundary component on one of the branes wrapping $S^{3}$ that can be computed by inserting
$$
 \det(1\otimes 1 -V^{-1} \otimes U(K))^{-1}
$$
in the Chern-Simons path integral, and which describes the effect of integrating out the bifundamental strings between $L_K$ and the $S^3$, generalizing \eqref{bif} to the case of more than one brane on $L_K$.

Consider instead $\n$ distinct Lagrangians that are copies of $L_{K}$ separated from it by moduli corresponding to ${\rm Re}(x_i)$. More precisely, we take the conormal of the $\n$-component link
${\tilde K} =K_{\epsilon}\cup \ldots \cup K_{\n\epsilon}$ where $K_{j\epsilon}$ is the knot obtained by shifting $K$ a distance $j\epsilon$, where $\epsilon>0$ is very small, along the framing vector field of $K$ used to define the quantum invariants (i.e.~the expectation values $\left\langle \Tr_{R} U(K)\right\rangle$). Note that the topological string partition function for a single brane on each $L_{K_{j\epsilon}}$ has exactly the same contribution, as follows e.g.~from the simple mathematical fact that
$$
\det(1\otimes 1 -V^{-1} \otimes U)^{-1}=\prod_i \det(1 -e^{-x_i} U)^{-1},
$$
which then holds inside the expectation values as well. In fact the system with $\n$ branes on $L_{K}$ is physically indistinguishable from the system with single branes on all components of the conormal $L_{\tilde K}$. The expression on the left is more naturally associated to the former, whereas the one on the right is more naturally associated to the latter.

Treating $\tilde K$ as a general link as considered in Section \ref{ssec:CStopstring}, disregarding the effects of the branes being very close, we get a corresponding string partition function
\[
\Psi_{\tilde K}(x_1,\dots,x_\n)
\]
given by the insertion in the Chern-Simons path integral described above. However, when computing the string partition for $\n$ coincident branes on $L_{K}$, it is natural to take these effects into account: contributions that are not captured by Chern-Simons theory correspond to worldsheets with no boundary component on $S^{3}$. For parallel branes these come from short strings connecting different branes; such strings only contribute nontrivially to the annulus diagram of the topological string with amplitude $\sum_{k>0}\frac{1}{n} e^{kx_i}e^{-kx_j}$ corresponding to strings connecting $L_{K_i}$ to $L_{K_j}$, $i<j$. Exponentiating this contribution, we then find the following relation between the partition function for branes on $L_K$ and the partition function of single branes on $L_{\tilde K}$, treated as a general link:
$$
\Psi_{K}^{[\n]}(x_1,\ldots,x_\n)=\Psi_{\tilde K}(x_1,\ldots,x_\n) \Delta(x_1,\ldots,x_\n),
$$
where
$$
\Delta(x_1,\ldots,x_\n)= \prod_{i<j}(e^{x_i}-e^{x_j}).
$$

Note that $\Psi^{[\n]}_{K}$ encodes the HOMFLY of the knot $K$ colored by representations with $\n$ rows:
$$
\Psi_K(x_1,\ldots,x_\n)= \left\langle\det\left(1\otimes 1 -V^{-1} \otimes U(K)\right)^{(-1)} \right\rangle  \Delta(x),
$$
using the expansion
$$
\det\left(1\otimes 1 -V^{-1} \otimes U\right)^{-1}
= \sum_{R}  \Tr_{R} U \;  \Tr_{R} V^{-1},
$$
where the sum ranges over all $SU(N)$ representations $R$ with at most $\n$ rows.

\subsection{Large $N$ duality}\label{ssec:largeN}
It was conjectured in \cite{GV} that $SU(N)$ Chern-Simons theory on $S^3$, or equivalently the topological A-model string on $X=T^*S^3$ with $N$ D-branes on $S^3$, has a dual description in terms of the topological A-model on the resolved conifold $Y$ which is the total space of the bundle
$$
{\cal O}(-1) \oplus {\cal O}(-1)\rightarrow {\C P}^1.
$$
Note that $X$ and $Y$ both approach the conical symplectic manifold $\R\times U^{\ast}S^{3}$, which is topologically $\R\times S^2\times S^3$, at infinity. At the apex of the cone sits an $S^{3}$ in $X$, while in $Y$, there sits a ${\C P}^{1}\approx S^{2}$. The duality arises from the geometric transition from $X$ to $Y$ that shrinks $S^3$ and replaces it by ${\C P}^1$ without altering the geometry at infinity. Furthermore, the $N$ branes on $S^3$ in $X$ disappear under the  transition, but their number is related to the size (symplectic area) $t$ of ${\C P}^1$ in $Y$ as follows: $t=Ng_s$, where $g_{s}$ is the string coupling constant. We write
\begin{equation}\label{eq:Q=expNg_s}
Q=\exp(-t) = q^N.
\end{equation}
The partition function of the closed topological string on $Y$ counts holomorphic maps into $Y$. All such maps arise from perturbation of branched covers of the central ${\C P}^1\subset Y$, and the variable $Q$ keeps track of their degree. In \cite{GV}, the partition function of the closed topological A-model string on $Y$ was shown to agree with the partition function of $SU(N)$ Chern-Simons theory on $S^3$, and consequently with the partition function of A-model topological string on $X$ in the background of $N$ D-branes on $S^3$, provided that the Chern-Simons parameters, $k$ and $N$, and the string coupling constant $g_s$ (for topological strings in both $X$ and $Y$) are related as $g_s = \frac{2 \pi i}{k+N}$, and that \eqref{eq:Q=expNg_s} holds. In other words, for parameters related as described:
$$
Z_{CS}(S^3;k,N) = Z_{GW}(X;g_s,N) = Z_{GW}(Y;g_s,Q).
$$

We next describe how to include knots and links in this picture, see \cite{OV}. Let $K=K_1\cup\dots\cup K_n\subset S^{3}$ be a link. As described in Section \ref{ssec:CStopstring}, adding branes along the Lagrangian conormal $L_{K}$, we relate the Chern-Simons path integral with insertions corresponding to the link with open topological string on $X$ with boundaries on either the $N$ branes on $S^{3}$ or on the branes on $L_K$.

Recall that $L_{K}\subset X$ can be pushed off of the zero section $S^{3}$, corresponding to turning on $x_{j}\ne 0$ along $L_{K_j}$. We thus assume that $L_{K}\subset X$ is disjoint from $S^{3}\subset X$ and consider the effect of the geometric transition from $X$ to $Y$. Since the transition affects only a small neighborhood of the tip of the cone, corresponding to small neighborhoods of $S^{3}\subset X$ and ${\C P}^{1}\subset Y$, the Lagrangian $L_{K}$ is canonically pushed through the transition as a Lagrangian in $Y$. Furthermore, in analogy with the closed string case discussed above, boundary conditions corresponding to worldsheet boundaries on the $N$ branes on $S^{3}$ close up and disappear, while worldsheet boundaries on the branes on $L_K$ remain unchanged. Thus, the partition function of branes on the components of $L_K$ in $Y$
$$
\Psi_{K_1, \ldots, K_n}(x_1, \ldots, x_n) =Z_{GW}(Y, L_{K_1}, \ldots, L_{K_n})/Z_{GW}(Y)
$$
also gives the partition function of branes in $X$, wrapping the Lagrangian in $X$ corresponding to $L_K\subset Y$ under transition, and $N$ branes on $S^{3}\subset X$, provided \eqref{eq:Q=expNg_s} holds.

\subsection{SYZ mirror symmetry for knots}\label{sec:SYZmirror}
Consider the A-model topological string on the resolved conifold $Y$ with a D-brane wrapping the Lagrangian conormal $L_K\subset Y$ of a knot $K\subset S^{3}$, see Section \ref{ssec:largeN}. There is then a contribution from short strings beginning and ending on $L_{K}$. Noting that a small neighborhood of $L_K$ is symplectomorphic to a neighborhood of the zero section in $T^{\ast}L_K$ and applying the construction of \cite{Witten:1992fb}, we find that this contribution is given by the partition function of $U(1)$ Chern-Simons theory on $L_{K}\approx S^{1}\times \R^{2}$. At infinity, $L_K$ looks like $\R\times\Lambda_K$, with ideal boundary $\Lambda_K\approx T^2$, which means we should study $U(1)$ (or more precisely, $GL(1)$) Chern-Simons theory on the manifold $S^{1}\times D^{2}$ with $T^2$ boundary. Let $\lambda$ be the longitudinal cycle of the $T^2$, which determines the parallel of the knot $K$ that links it trivially, and let $\mu$ be the meridional cycle. We denote by $x$ the holonomy along $\lambda$, and by $p$ the holonomy along $\mu$:
$$
 \oint_\lambda A = x, \qquad \oint_\mu A  = p,
$$
where $A$ is the connection $1$-form.
In Chern-Simons theory on $L_K$, the holonomies of $\lambda$ and $\mu$ are canonically conjugate:
$$
[p,x] = g_s,
$$
where $g_s$ is the string coupling constant. This means in particular that for a D-brane on $L_K$ in $Y$, $p$ and $x$ are not independent; rather, if $\Psi_K(x)$ denotes the partition function of $L_K$, we have
$$
p \, \Psi_K(x) = g_s {\partial \over \partial x}\Psi_K(x).
$$
In particular, we should view $\Psi_K(x)$ as a wave function with asymptotics
$$
\Psi_K(x)\sim \exp\left(\frac{1}{g_s}\int p \,dx +\dots\right).
$$
From the perspective of topological strings, in the classical limit $g_s\to 0$, maps of the disks dominate the perturbation expansion, with maps of more complicated Riemann surfaces giving contributions of higher order. This means that
$$
\Psi_K(x)\sim \exp\left(\frac{1}{g_s} W_K(x,Q)+\dots\right),
$$
where
$$
W_K(x,Q) = \sum_{n,k\geq 0} a_{n, k} \, Q^k\, e^{-n x}
$$
is the generating function of Gromov-Witten invariants $a_{n,k}$ corresponding to counting holomorphic disks in $Y$ with boundary on $L_K$. Thus, at the level of the disk,
\begin{equation}\label{eq:p=dW/dx}
p = p(x)= {\partial W_K\over \partial x}(x,Q).
\end{equation}

We  point out that \eqref{eq:p=dW/dx} is consistent with the fact that, classically in $L_K$, the cycle $\mu$ bounds. Here $\mu$ bounding in $L_K$ means that the holonomy around $\mu$ vanishes if we disregard the disk instanton corrections: ${\partial W_K\over \partial x}$ vanishes up to contributions of instantons.
This leads to the interpretation of the moduli space of the brane as a Lagrangian curve $V_K\subset T^*T^2$, with coordinates $x,p$ and symplectic form $dx\wedge dp$, given by the equation
\begin{equation}\label{AQ}
A_K(e^{x}, e^{p},Q)=0
\end{equation}
for a polynomial $A_K(e^{x}, e^{p},Q)$ (where we view $Q$ as a parameter).
In general, for large $x$, \eqref{AQ} may have more than one solution $p=p(x)$ for $p$ in terms of $x$. This corresponds to the fact that the theory may have more than critical point in this phase; to get the solution corresponding to $L_K$, we need to pick the one where $p\sim 0 +{\cal O}(e^{-x})$.

The curve \eqref{AQ} sums up disk instanton corrections to the moduli space of a D-brane on $L_K$ probing $Y$. It was argued in \cite{AV} that this gives rise to a mirror Calabi-Yau manifold, $X_K$, given by the hypersurface 
\begin{equation}\label{CYAAQ}
A_K(e^{x}, e^{p},Q) =uv.
\end{equation}
Mirror symmetry exchanges A-branes, i.e.~Lagrangian submanifolds, of $Y$ with B-branes, i.e.~holomorphic submanifolds, of the mirror $X_{K}$. Moreover, quantum mirror symmetry implies that for every (special) Lagrangian brane in $Y$ there is a B-brane in $X_K$ such that the quantum corrected moduli space of $L_K$ in $Y$ is the same as the classical moduli space of the mirror B-brane in $X_K$. Here we see that the quantum dual B-branes are given by the line $\{u=0, v \text{ arbitrary}\}$ over a point $q\in V_K$, i.e.~if $q=(e^{x_0},e^{p_0})$ then \eqref{AQ} holds for $(e^{x},e^{p})=(e^{x_0},e^{p_0})$.

In the special case when $K$ is the unknot, we obtain the same mirror of $Y$ as in \cite{HV}, but for more general knots, we get new mirrors. In fact, for every knot $K$ in $S^3$ and its associated Lagrangian $L_K$ in $Y$, we get a canonical mirror geometry \eqref{CYAAQ}, where the mirror of
$L_K$ is determined by a point on the mirror Riemann surface $V_K$.
Of course, we expect each of the mirror Calabi-Yau manifolds \eqref{CYAAQ} to contain not just the mirror of the Lagrangian $L_K$ corresponding to the knot $K$, but also, by mirror symmetry, mirrors of all other Lagrangian branes $L_{K'}$ for knots $K'\neq K$. In general, however, these mirrors have more complicated descriptions.

\section{Chern-Simons theory and Lagrangian fillings}\label{Sec:CSandfillings}
Via large $N$ duality, the Chern-Simons path integral on $S^3$ contains non-perturbative information about the A-model topological string in the resolved conifold $Y$. Before the transition, in $X=T^{\ast} S^{3}$, the data that go into specifying the theory are just the knot $K$, the number of branes on it, and the holonomy at infinity.  We will set the number of branes on $K$ to be $1$ in this section, and consider the higher rank generalization in the next section. The data are imprinted at infinity of the ambient Calabi-Yau manifold, and are thus visible both before and after the transition. After the transition, the data at infinity may be compatible with more than one filling in the interior. In Section \ref{sec:SYZmirror}, we focused on the filling that gives the Lagrangian conormal $L_K$, but, as we shall see, Chern-Simons theory encodes information about other Lagrangian fillings as well; these correspond to different classical solutions of the topological string.

At infinity, the Calabi-Yau looks like $\R\times U^{\ast}S^{3}\approx \R\times S^3\times S^2$ and the Lagrangian $L_{K}$ approaches $\R\times \Lambda_K$, where $\Lambda_K$, the conormal of the knot, is naturally identified with the boundary of a tubular neighborhood of $K$ in $S^3$ and is thus topologically a torus
\footnote{As will be further discussed in Section \ref{Sec:knotch}, $U^{\ast}S^{3}\approx S^2\times S^3$ is naturally a contact manifold, and $\Lambda_K$ a Legendrian submanifold. These observations are the starting  point for the knot contact homology approach to topological string in this background.}.
There are two natural basic $1$-cycles in $\Lambda_K$: the longitude $\lambda$, which is a parallel copy of the knot, and the meridian $\mu$, which is the boundary of a fiber disk in the tubular neighborhood.

A vacuum of the topological string with these boundary conditions is determined once we specify the holonomy of the $GL(1)$ connection around a $1$-cycle of the torus $\Lambda_K$ at infinity and find a Lagrangian brane that fills in $\R\times\Lambda_K$ in $Y$ with these moduli. Note that the holonomy of the $GL(1)$ connection encodes both the position of the brane and the $U(1)$-holonomy on it, see Section \ref{sec:SYZmirror}.\footnote{We cannot expect to specify holonomies around both cycles of $\Lambda_K$ simultaneously, as the theory on the Lagrangian A-brane in a Calabi-Yau three-fold is Chern-Simons theory, and in Chern-Simons theory on a manifold with a $T^2$ boundary holonomies of the 1-cycles generating $H_1(T^2)$ are canonically conjugate. } In the discussion so far we used the filling $L_K$ and fixed the
holonomy around the longitude cycle $x$. With a fixed connection at infinity, there may be more than one corresponding filling. Moreover, having found a vacuum with one choice of the filling, there is an $SL(2,\mathbb{Z})$ family of choices of flat connections, related by canonical transformations, that gives rise to the same filling. This corresponds to a choice of the framing of the Lagrangian, which we will discuss in more detail in Section \ref{sec:framingoflag}.

By taking the holonomy $x$ around the longitude to be very large, we always get a choice of filling corresponding to the conormal Lagrangian $L_K$. The mirror Riemann surface
\begin{equation}
0= A_K(e^{x}, e^{p},Q),
\end{equation}
which encodes the geometry of the mirror,  also has information about fillings of $\R\times\Lambda_K$.
As we explained, the Lagrangian $L_K$ corresponds to a branch of the Riemann surface where
$$
p= \frac{\partial W_K}{\partial x}(x) \sim 0.
$$
Since $p$ is the holonomy around the meridian of the knot, the fact that it vanishes classically means that the meridian cycle gets filled in in $L_K$, as is indeed the case topologically. The Lagrangian $L_K$ has a one-dimensional moduli space, and the branch $p\sim 0$ is the branch where $x\rightarrow \infty$ and the worldsheet instanton corrections are maximally suppressed.
As we vary $x$, we can probe all of the Riemann surface and go to a region where this is no longer the case.

There can be other ways of filling in. The region of $p$ and $x$ where the mirror geometry is well approximated by
\begin{equation}\label{asymp}
m p+n x \sim \text{const}
\end{equation}
should correspond to filling in which the $m\mu+n\lambda$ cycle in $\Lambda_K\approx T^2$ bounds. In all cases, ``$\sim$'' denotes ``up to instanton corrections'': if the approximation in \eqref{asymp} can be made arbitrarily good in the mirror, on the original Calabi-Yau $Y$, then the instanton correction can be made arbitrarily small, and the corresponding classical geometry exists in $Y$. The relation between the different phases should be akin to flop transitions in closed Gromov-Witten theory, in the sense that phase transitions change the topology of the cycles in the manifold: in this case, of the Lagrangian.

For example, filling in the longitude cycle $\lambda$ gives
$$
x\sim 0.
$$
The Lagrangian $M_K$ one obtains in this way is related to the knot complement $S^3- K$. In the special case when the knot $K$ is fibered, the complement can actually be constructed as a Lagrangian submanifold of $Y$ asymptotic to $\R\times\Lambda_K$, see Section \ref{ssec:constrlag}. In the general case we will give a conjectural construction of a Lagrangian filling with the same classical asymptotics, see Section \ref{ssec:antibrane}. We will discuss this with further details in later sections. Here we simply let $M_K$ denote a Lagrangian filling of $\Lambda_K$ in $Y$ that has the classical asymptotics $x\sim 0$. In this setting, the holonomy $p$ around the meridian cycle is the natural parameter. In fact, the topological string partition function of $Y$ with a single brane on $M_K$ turns out to be
\begin{equation}\label{kcp}
\tilde{\Psi}_{K}(p)  = H_{p/g_s}(K),
\end{equation}
computed
by the HOMFLY polynomial $H_m(K)$ of the knot colored by the totally symmetric representation with $m$ boxes. Here $m =p/g_s$ where $p$ is fixed and $g_s$ is small. This follows from existence of a geometric transition that relates $M_K$ and $L_K$.  Despite the fact that the transition changes the topology of the cycle that gets filled in, it is quantum mechanically completely smooth due to instanton corrections, in the case of knots. In later sections, we will study the case of links, where again there are different choices of filling in the cycle at infinity. However, there the phase structure of the theory becomes far more intricate.

\subsection{Geometric transition between $L_K$ and $M_K$}
Let us clarify the relation between $L_K$ and $M_K$ and  the nature of the transition between them. We start with $N$ D-branes on the $S^3$ and a single brane on the conormal bundle $L_K$, intersecting along the knot $K$
$$
L_K \cap S^3 = K.
$$
The fact that $L_K$ and $S^3$ intersect makes the configuration slightly singular, but one can remedy this by using the fact that $L_K$ moves in a one-parameter family, parametrized by $x$. For any $\Re(x)\neq 0$, we can then move $L_K$ off of $S^3$ so that they no longer intersect.

However, there is another way to smooth out the singularity by smoothing out the intersection between the $S^3$ and $L_K$. This requires breaking the gauge symmetry on $S^3$ by picking one of the $N$ D-branes.
Cut out a neighborhood of the knot $K$ from both the D-brane on $S^3$ and on $L_K$. This gives $S^3-K$ and
$L_K-K$ both with a torus boundary. The two manifolds can then be glued together along their boundaries. Topologically, this results in $S^{3}-K$, since gluing in $T^2 \times {\R}$ does not change the topology. (This way of obtaining $S^{3}-K$ discussed, in a related context, in \cite{Dimofte2}.)
As mentioned above, it is not always possible to move the Lagrangian version of $S^{3}-K$ off of the zero section, and we have a somewhat more involved construction with similar features to deal with this case. As in the previous section, we will use $M_K$ to denote a Lagrangian that classically gives $x\sim 0$.

Before the transition, the Lagrangian $S^{3}- K$ and $L_K$ project with different degrees to $S^3$ and no transition between them is possible. Also before the transition, $N$ is finite, and hence $Q=e^{N g_s}=1$ classically. In this limit, the curve factorizes and contains
\begin{equation}
0= A_K(e^{x}, e^{p},Q=1) = (1-e^p)A_{K \text{ class}}(e^x,e^p)\times \ldots
\end{equation}
where $A_{K\text{ class}}$ is the classical $A$-polynomial of the knot, describing the moduli space of flat connections on the knot complement. The $e^p-1=0$ branch, corresponding to $L_K$, disconnects from the curve describing $S^{3}-K$.

After the transition, the $S^3$ has shrunk, $Q$ gets an expectation value $Q\neq 1$,
the curve generally becomes irreducible
\begin{equation}
0= A_K(e^{x}, e^{p},Q),
\end{equation}
and the distinction between $L_K$ and $M_K$ disappears. Since both $p\sim 0$ and $x\sim 0$ branches lie on the same Riemann surface,  $M_K$ and $L_K$ are  smoothly connected once we include disk instanton corrections, with no phase transition between them.

To make this point even clearer, it is useful to consider one very simple example of this when
$K$ is the unknot.  We will study this both before and after the transition.
In the mirror geometry $L_K$ is given by choosing a point $\lambda$ in the curve
\begin{equation}\label{eq:V_O}
Q-\lambda-\mu+\lambda \mu=0
\end{equation}
where $\lambda=e^x$ and $\mu=e^p$.
When $|\lambda|\gg 1$ this gives the Lagrangian mirror of $L_K$ for the unknot.  Note that
this is consistent with $\mu=1$, i.e.,  $p=0$.  The Lagrangian mirror to $M_K$ is given by
$|\mu|\gg 1$, which requires $\lambda =1$, i.e.,  $x=0$ as expected.  Note that going from $L_K$ to $M_K$
can be done via a smooth path in the mirror curve, even though in the original
geometry they are classically distinct.\footnote{Note also that before the transition, to go from $L_K$
to $M_K$, one must add one copy of $S^3$, as is clear from the path in the mirror geometry
going from one asymptotic region to the other, taking into account that the periods
get modified by one unit around the cycle mirror to the blown-up ${\C P}^1$.}

The transition between $L_K$ and $M_K$ described above has a simple interpretation in the topological string. At the intersection of the D-branes wrapping $L_K$ and $S^3$ lives a pair of complex scalars $A$, $B$ transforming in the bifundamental representation $(N,-1)$, $(\bar N,1)$,  corresponding to strings with one boundary on the zero section and the other on $L_K$. As long as $L_K$ and $S^3$ intersect, the bifundamental is massless (the real mass vanishes).  In one phase, we move the Lagrangian $L_K$ off of the zero section. In this phase, the bifundamental hypermultiplet gets a mass $r=\Re(x)$ corresponding to the modulus of moving off. From this perspective, the partition function $\Psi_K(x)$ arises as follows. Integrating out the massive bifundamental of mass $x$ generates the determinant \cite{OV}
$\det(1\otimes 1 - e^{-x} \otimes U)^{-1}$,
where, as usual, $U = P e^{i \oint_K A}$ is the holonomy along the knot and $x$ is identified with holonomy on $L_K$.
In Chern-Simons theory, we integrate over $A$ and thereby compute the expectation value of the determinant
$\Psi_K(x) =\langle \det(1\otimes 1 - e^{-x} \otimes U)^{-1}\rangle_{S^3}$, which is in turn computed by the colored HOMFLY polynomial
$$
\Psi_K(x) = \sum_{n=0}^{\infty} H_n(K)\; e^{-n x },
$$
as explained earlier.
This is in fact the derivation of the partition function given in \cite{OV}.

The theory has a second phase, where the bifundamental hypermultiplet
remains massless and gets an expectation value $p$. This requires
the gauge group on the $S^3$ to be broken from $U(N)$ to $U(N-1)\times
U(1)$, and then the $U(1)$ factors on the $S^3$ and on $L_K$ to be
identified.  Giving the expectation value to the bifundamental gives rise to the smoothing of the intersection between $L_K$ and $S^3$, which we described above.
In the next subsection we will show that the partition function in this phase is given by \eqref{kcp}.

\subsection{Framing of the Lagrangians}\label{sec:framingoflag}
To write down the partition function of the theory, in addition to choosing the vacuum by picking a point on the Riemann surface, we have to choose a framing of the Lagrangian: a flat connection at infinity. Since the holonomies around different cycles of the $\Lambda_K\approx T^2$ do not commute, different choices are related by wave-function transforms.
In $L_K$, the natural variable is the holonomy along the longitude of the knot $x$, and the wave function is given by
$$
  \Psi_{K}(x) = \sum_{S_{n}} H_n(K) \;e^{-nx}= \exp\left(\frac{1}{g_s} W_K(x) +\ldots\right).
$$

In $M_K$, the natural variable is $p$, the holonomy around the meridian. Since $x$ and $p$ are canonically conjugate in the theory on the brane,
$$
[p,x]=g_s,
$$
and moreover since there is no distinction between $M_K$ and the $L_K$ phase, we can identify the Fourier transform of $\Psi_K(x)$ with the partition function of $M_K$:
$$
\tilde{\Psi}_K(p) = \int_{-\infty}^{\infty} e^{\frac{1}{g_s} px}  \Psi_K(x).
$$
We can simply deform the contour, at least in perturbation theory, to find \eqref{kcp}:
$$
\tilde{\Psi}_K(p) = H_{p/g_s}(K) = \exp\left(\frac{1}{g_s}U_K(p) +\ldots\right).
$$
In particular, the Gromov-Witten ``potentials'' $W_K(x)$ and $U_K(p)$ that we previously associated to $L_K$ and $M_K$ are simply dual to each other,
$$
U_K(p) \;=_{\text{crit}}\;  px + W_K(x),
$$
related by Legendre transformation. In Section \ref{sec:legtransf}, we will provide a different path to this result by directly comparing contributions of disk instantons to Gromov-Witten potentials of $M_K$ and $L_K$.

\section{Large $N$ duality, mirror symmetry, and HOMFLY invariant for links}
\label{sec:largeN}
Consider Chern-Simons theory on $S^3$ with an $\n$-component link $K$,
$$
K = K_1 \cup \ldots \cup K_\n,
$$
where $K_j$ are the knot components of $K$. The conormal $L_K$ of $K$
in $X=T^*S^3$ is the union of the $\n$ conormals of the knot components,
$$
L_K = L_{K_1} \cup \dots \cup L_{K_\n},
$$
and $L_K\cap S^{3}=K$. Under large $N$ transition as in Section \ref{ssec:largeN}, $X$ is replaced by the total space $Y$ of the bundle ${\cal O}(-1)\oplus{\cal O}(-1)\rightarrow {\C P}^1$. As we explained in Section \ref{sec:framingoflag}, it is natural to view the boundary data as the only fixed data, in which case we should consider the geometry of $L_K$ far from the apex of the cone.
Here the Lagrangian $L_K$ approaches $\R \times \Lambda_K$, where
$\Lambda_K$ is the union of the $\n$ conormal tori $\Lambda_{K_j}$ of the components of $K$ that describe the imprint of the corresponding branes at the  $S^2\times S^3$ at infinity of the Calabi-Yau.
As in the knot case, there are different ways to fill in $\Lambda_K$ in the interior of $Y$. As we shall see, the link case is more involved than the case of a single knot: there is a larger number of ways to fill in $\R\times\Lambda_K$, connecting different numbers of components of $\Lambda_K$ in the interior.

As in Section \ref{sec:SYZmirror}, we are led to consider holonomies of a $GL(1)$ gauge field on the brane around the $1$-cycles of $\Lambda_K$. This gives a phase space of the system, which is the cotangent bundle of a $2\n$-dimensional torus:
$$
{\mathcal M}_\n=(\C^{\ast})^{2\n} = T^*(T^{2\n}).
$$
The torus is the torus of holonomies of the $U(1)$ gauge field on the brane around the $2\n$ $1$-cycles
of $\Lambda_K$. The cotangent direction arises from the moduli of the Lagrangians. Equivalently, the phase space ${\mathcal M}_\n$ arises from holonomies of the $GL(1)$ gauge field on $\R\times\Lambda_K$ around the cycles of $\n$ copies of  $T^{2}$ that comprise the infinity $\Lambda_K$.  Thus on ${\mathcal M}_\n$ we have $(\C^\ast)^{2n}$-coordinates
$$
e^{x_i}, \quad e^{p_j}, \qquad i,j = 1,\ldots, n
$$
associated to the longitudes and the meridians of knots $K_i$, and a symplectic form
$$
\omega = \sum_i dx_i \wedge dp_i.
$$
Note that $\mathcal{M}_\n$ is in fact hyper-K{\"a}hler, a fact which we will make use of later in this paper.

Inside $\mathcal{M}_\n$ there is a Lagrangian subvariety $V_K$ associated to the link (i.e.~$V_K$ is half-dimensional and $\omega|_{V_K} = 0$). As we shall see, $V_K$ is in general reducible:
$$
V_K = \bigcup_P V_K(P),
$$
where the subvarieties correspond to different fillings of $\Lambda_K$ labeled by partitions $P$ of $\{1,\ldots,n\}$. We want to identify the fillings that can be obtained by smoothly varying the holonomies, since these give rise to the same Lagrangian submanifold $V_K(P)$.  We will initiate the study of the varieties in this section, using Chern-Simons theory and large $N$ duality. In Sections \ref{Sec:knotch} and \ref{sec:augmentations} we will use another approach based on knot contact homology, where $V_K$ will be identified with the augmentation variety of the knot. In Section \ref{sec:Dmirror}, we will study quantum mirror symmetry, where we explain how to quantize this variety.

\subsection{Conormal bundle $L_K$}
In the simplest case, $\R\times\Lambda_K$ is filled by the conormal bundle of the link
$$
L_K=L_{K_1} \cup\ldots \cup L_{K_\n}.
$$
This is simply a union of $\n$ disconnected Lagrangians $L_{K_j}$, which we move off of $S^3$ independently and then transition to $Y$.  This was studied previously in \cite{LMV}.

On $Y$, in the presence of $L_K$, there are no holomorphic disks with boundary on more than one component of $L_{K}$, as the conormals $L_{K_j}$ are disjoint from each other. So the disk amplitude is simply a sum of $k$ contributions, coming from disks with boundaries on one of the $L_{K_j}$:
\beq\label{Von}
W_{K,1^\n}(x_1,\dots,x_\n) = W_{K_1}(x_1)+\dots +W_{K_\n}(x_\n).
\eeq
From the perspective of large $N$ dualities, this arises as follows. We start with $X=T^*S^3$, with the $\n$ Lagrangians $L_{K_j}$ intersecting the zero section along the link $K$. Pushing $L_{K_j}$ off of $S^3$ by $x_j$ for each $j$,
we give a mass $x_j$ to the bifundamental hypermultiplet at the intersection of $K_j$ and $S^3$.
Integrating these $k$ massive hypermultiplets out, we end up computing the expectation value of
\begin{align*}
&\prod_{i=1}^\n \det(1\otimes 1 - e^{-x_i} \otimes U(K_i))^{-1}  \\
&= \sum_{m_1, \ldots, m_\n}  {\Tr}_{m_1}U(K_1)\ldots {\Tr}_{m_\n}U(K_\n) \; e^{-m_1x_1-\cdots-m_\n x_\n}
\end{align*}
where $m_i$ denotes the $m_i$-th symmetric representation of the fundamental representation, in the topological string or Chern-Simons theory.
This can be written as
\beq\label{exacpsilink}
 \Psi_{K, 1^\n}(x_1,\dots,x_\n, g_s, N)= \sum_{m_1, \ldots, m_\n} H_{S_{m_1},...,S_{m_\n}}(K) \; e^{-m_1x_1-\cdots-m_\n x_\n},
\eeq
where
\beq\label{representation}
H_{m_1,...,m_\n}(K)=\langle {\Tr}_{m_1}U(K_1)\ldots {\Tr}_{m_\n}U(K_\n) \rangle.
\eeq
As was shown in \cite{LMV}, the classical limit of this is
$$
\Psi_{K, 1^\n}(x_1,\dots,x_\n, g_s, Q)= \exp\left(\frac{1}{g_s}W_{K,1^\n}(x_1,\dots,x_\n)+\ldots\right),
$$
where $W_{K,1^\n}$ is as in \eqref{Von}. In this case, the link indeed gives rise to a $\n$-dimensional variety, but one which is simply a direct product of $\n$ $1$-dimensional curves,
$$
V_K(1^\n):\qquad p_i = \frac{\partial W_{K_i}}{\partial x_i}(x_i).
$$
We will denote this variety by $V_{K}( 1^\n)$, to indicate that it factors as a product of $\n$ pieces, each of dimension $1$.

\subsection{Link-complement-like fillings $M_K$}\label{sec:linkcompl}
Just like in the knot case, there are other fillings of $\R\times{\Lambda_K}$ in $Y$. Here we focus on fillings that in the classical limit looks like the link complement. We denote a general such filling $M_K$, see Section \ref{ssec:antibrane} for geometric constructions of such fillings.

Physically, $M_K$ is obtained by giving expectation values $p_i = m_i g_s$ to the $\n$ hypermultiplets at the intersections of $L_{K_i}$ and $S^3$ in $X=T^*S^3$. This smooths the intersections of the D-branes on $L_{K_i}$ with one of the D-branes on the $S^3$, to give a single Lagrangian $M_K$.  This Lagrangian has first Betti number $b_1(M_K) = \n$ so we still have $\n$ moduli that allow us to move $M_K$ off of the zero section. These $\n$ moduli are the $p_i$'s, which are also the holonomies around the meridians. After that, the geometric transition relates the Lagrangian $M_K$ to a copy of itself on $Y$.

The disk amplitude depends on $\n$ meridian holonomies $p_i$, and is obtained from the colored HOMFLY polynomial of the link,
\beq\label{representationa}
H_{m_1,\ldots,m_\n}(K)=\langle {\Tr}_{m_1}U(K_1)\ldots {\Tr}_{m_\n}U(K_\n) \rangle,
\eeq
by rewriting it in terms of $p_i = m_i g_s$:
$$
\begin{aligned}
\tilde{\Psi}_{K, \n}(p_1, \ldots , p_\n, g_s, Q) &= H_{p_1/g_s,\ldots, p_\n/g_s}(K, g_s, Q) \cr
&= \exp\left(\frac{1}{g_s}U_K(p_1,\ldots p_\n, Q)+\ldots\right).
\end{aligned}
$$
Equivalently, we can write this as
$$
U_K(p_1,\ldots p_\n, Q) = \lim_{g_s \rightarrow 0} [g_s \log(H_{p_1/g_s,\ldots, p_\n/g_s}(K, g_s, Q))].
$$

The quantum corrected moduli space of the brane on $M_K$ in $Y$ is an $\n$-dimensional variety $V_{K}(\n)$ given by
\begin{equation}\label{eq:x_i=dU/dp_i}
V_K(\n):\qquad x_i = \frac{\partial U_K}{\partial p_i}(p_1,\ldots,p_\n, Q).
\end{equation}
By the above construction,
$V_K(\n)$ can be
viewed as a Lagrangian subspace in the $2\n$-dimensional space of $(x_i,p_i)\in (\C)^{2\n}$ relative to the symplectic form
$\sum_i dx_i\wedge dp_i$.

Note that if the link $K$ is a completely split link (i.e., one can find disjoint solid balls $B_1,\ldots,B_\n \subset S^3$ such that $K_i \subset B_i$ for each $i$), 
then in fact $V_K(\n)$ and $V_K(1^\n)$ coincide. This is because the HOMFLY polynomial of the link factor in this case,
\beq
H_{m_1,...,m_\n}(K)=_{K\text{ split}}H_{m_1}(K_1) \ldots H_{m_\n}(K_\n),
\eeq
and consequently the brane partition function, is an exact product of $V_{K_i}(1)$. This suggests that in this case, one can go smoothly from the phase where we have a single Lagrangian $M_K$, to the conormal Lagrangian $L_K$ consisting of $\n$ disconnected Lagrangians.
Since $V_K(\n)$ would have no new information in this case, we will simply identify it with $V_K(1^\n)$ and reserve the notation of $V_K(\n)$ for a non-split link $K$ of $\n$ components: the fact that the link is non-split translates into the fact that $V_K(\n)$ is not a direct product of some lower-dimensional varieties.

Just as in the case of the knots it is natural to expect $V_K(\n)$ to correspond to moduli
space of $SL(2,\C)$ flat connections on $S^3- K$.  More precisely, we expect this to be a $Q$-deformed
version of it, where $SL(2,\C)$ is embedded in the canonical way in $SL(N,\C)$.  This is natural
from the viewpoint of the Higgsing construction discussed before.

\subsection{General fillings}
We will now describe the set of physically distinct ways of filling in the torus $\Lambda_K$ in the interior of $Y.$ As we will explain, the physically distinct fillings are labeled by the ways to partition an $\n$-component link $K$ into sublinks.
The different ways to do this are labeled by integer partitions $P$ of $\{1,\ldots,\n\}$, where we view the parts $P_a$ of the partition $P$ as giving corresponding sublinks $K_a$ of $K$. An alternative way to label the fillings is by a link $K(P)$ which is a split union of sublinks $K_a$:
$$K(P) =\bigcup_{P_a \subset P} K_a.$$
(That is, move each of the sublinks away from the others, so that the result is split.) We conjecture that the general distinct fillings are labeled by the \textit{primitive} partitions $P$, namely those for which each sublink $K_a$ is non-split, or, equivalently, those that have no refinement $P'$ such that $K(P) = K(P')$. 
The set of primitive partitions of an $\n$-component link $K$ is a subset of the set of all partitions of $\{1,\ldots,\n\}$.  For example, for a fully split link, there is a single primitive partition $P=\{\{1\},\{2\},\ldots, \{\n\}\}$, and a single distinct partition of $K$ into sublinks: $K=K(P)$.

Let us now explain why this is reasonable. The Lagrangian fillings of $\Lambda_K$, taken crudely, describe how components of $\Lambda_K$ come together in the interior of $Y$. What one misses in this description is any information about which cycles get filled in; however, as we have seen, this fact is not invariant (one should recall the knot case, where phases with different cycles filled in get smoothly connected, due to instanton corrections).  We want to identify those fillings which, while they may result in distinct classical geometries of the Lagrangian, are indistinguishable once disk instanton corrections are taken into account. Since one can interpolate smoothly between such Lagrangians, they lead to the same mirror variety.   If this were the only consideration, we would simply label the fillings by all partitions $P$.  However, we also have to take linking information into account.  If some sublinks are split, since the amplitudes
should not depend on where the corresponding Lagrangians are, we can separate them infinitely far from one
another, which suggests there could not have existed a single Lagrangian
with those asymptotics.  Even if such Lagrangians existed, one could smoothly interpolate between the would-be connected ones and the disconnected phases of the Lagrangian by BRST trivial deformation as we separate them infinitely apart. Thus, partitions $P$ that are not primitive contain only redundant information.
Thus, for an $\n$-component link $K$, 
the physically inequivalent fillings are labeled by primitive partitions $P$ of $\{1,\ldots,n\}$, the ones that result in sublinks $K_a$ that do not split further.

Given such a primitive partition $P$ and the associated collection of sublinks $K(P)$, we will denote the corresponding Lagrangian filling by
$$
F_{K(P)} \qquad \subset \qquad Y.
$$
The filling $F_{K(P)}$ gives rise to an $\n$-dimensional Lagrangian submanifold
$$
V_K(P) = \prod_{P_a\subset P} V_{K_a} \subset {\cal M}_\n,
$$
which is a product of the varieties $V_{K_a}$ associated to the sublinks $K_a$ of $K(P)$. Each sublink $K_a$ is by assumption non-split; suppose that $K_a$ consists of $\n_a$ knot components, where $\n = \sum_a \n_a$. We now explain the meaning of $V_K(P)$.

Consider the disk amplitude
$$
U_{K(P)}(p_1, \ldots, p_\n)
$$
obtained by studying the topological A-model string on $Y$, together with the Lagrangian brane ${F}_{K(P)}$.
We define the ``D-mirror" variety $V_K(P)$ as
$$
V_K(P):\qquad x_i = \frac{\partial U_{K(P)}}{\partial p_i}(p_1,\dots, p_k).
$$
We have written $U_{K(P)}$ in terms of the meridian variables, which is a convenient choice when the meridian cycles get filled in.
Since there cannot be any disk instantons ending on disconnected components of a Lagrangian filling, the disk potential is a sum
$$
U_{K(P)} = \sum_{a} U_{K_a}.
$$
Each disconnected component $M_{K_a}$ of the filling is homologically like the complement of the link $K_a$ in the classical limit, and $U_{K_a}$ is the corresponding potential.
This together with the fact that the potentials $U_{K_a}$ give rise to varieties $V_{K_a}(\n_a)$ gives the result we claimed.
This is consistent with what we expect from large $N$ duality: since the link $K(P)$ is a split link, with $K_a$ its non-split sublinks, the corresponding potential should be captured by the classical limit of the HOMFLY polynomial of $K(P)$:
$$
U_{K(P)}(p_1,\ldots p_{\n_a}) = \lim_{g_s \rightarrow 0} g_s \log(H_{p_1/g_s,\ldots, p_{\n_a}/g_s}(K(P))).
$$
It is important to note that the HOMFLY polynomial $H_{K(P)}$ of the link $K(P)$ is {\it not} the full topological string amplitude for this filling;
$P$ and $K(P)$ only label the filling. The branes in the filling labeled by $P$ do remember the full data of the link $K$ they came from. The partitions $P$ for which we get nontrivial fillings correspond to saddle (critical) points of the HOMFLY polynomial $H_K$ of the link, and the full 
partition function of topological string with the Lagrangian branes on the filling $F_{K(P)}$ is encoded by the expansion of $H_K$, the HOMFLY polynomial of the original link $K$, around its saddle point $P$. We will discuss this in more detail below.

To summarize, for each different splitting of $K$ into non-split sublinks $K_{a}$,  labeled by the primitive partition $P$, we get a Lagrangian filling $F_{K(P)}$ of the Legendrian tori $\Lambda_K$ at infinity. This filling is a union of link-complement-like Lagrangians $M_{K_a}$, one for each sublink of $K(P)$.  The mirror variety corresponding to this filling is $V_{K}{(P)}$, which is a product of the corresponding varieties $V_{K_a}(\n_a)$.

\subsection{D-mirror variety}\label{sec:Dmirror}
Given the link $K$, it is natural to define the mirror variety so that it depends on the data at infinity of the Calabi-Yau alone, and not the specific filling. With this in mind, we propose to identify
$$
V_K = \bigcup_P V_K(P)
$$
as the mirror variety to the link $K$. This is a natural proposal as it contains information both about knots that comprise the link, and the way they are linked together. Moreover, this is exactly the same data that goes into defining the dual Chern-Simons partition function on $S^3$ with the link $K$.
In Section~\ref{Sec:Dmodel}, we propose a way to quantize this variety.

Information about all the fillings $P$ is contained in the HOMFLY polynomial of the link. In particular, quantum mechanically, $V_K(P)$ are the Lagrangian submanifolds of ${\cal M}_\n$ that are associated to different classical saddle points of  the wave function\footnote{Non-perturbative aspects of topological strings have been studied in \cite{Eynard:2008he,Marino:2009dp,Beem:2012mb}.}
$$
\Psi_K(x_1, \ldots, x_\n) = \sum_{m_1, \ldots, m_\n} H_{m_1, \ldots, m_\n}(K) e^{-m_1 x_1 - \ldots - m_\n x_\n}.
$$
In general, different saddle points contribute to $\Psi_K$:
$$
\Psi_K(x_1, \ldots, x_\n) = \sum_P c_P \Psi^{P}_{K}(x_1, \ldots, x_\n),
$$
where $ \Psi^{P}_{K}$ is a wave function canonically associated to the corresponding saddle point $P$, and the coefficients $c_P$ are integers.
Which saddle point dominates the Chern-Simons path integral depends on the values of the parameters $x_i$, $g_s$, and $N$, and the classical action at the saddle point
$$
\Psi^P_K(x_1,\ldots,x_\n) = \exp\left(\frac{1}{g_s}W_K(P)(x_1, \ldots, x_\n) +\ldots\right),
$$
where $W_K(P)$ is the potential defining the variety $V_K(P)$ via
$$
V_K(P): \qquad p_i = \frac{W_K(P)}{\partial {x_i}}.
$$
In the regime where one of the saddle points dominates the others, are exponentially suppressed.

The above way of writing $\Psi_K$ is a little bit crude because it neglects the fact that in general, a single $P$ may give rise to more than one wave function, as $V_K(P)$ may give rise to several vacua. We will disregard this fact for two reasons: notational simplicity, and the fact that it is $V_K(P)$ that plays the crucial role in quantizing the theory. (While the subtlety affects the possible $W_K(P)$, it does not affect $V_K(P)$.)

The existence of a single function $\Psi_K(x_1, \ldots, x_\n)$ that has all these different classical limits imposes constraints on the structure of $V_K$, viewed as a reducible variety. As we will discuss in Section~\ref{Sec:Dmodel}, quantization of the variety $V_K$ gives rise to a D-module on ${\cal M}_\n$. The wave function $\Psi_K(x_1, \ldots, x_\n)$ is a section of this D-module. The fact that the D-module that arises in this way is irreducible (it corresponds to a single irreducible quantum system, rather than being a direct sum of super-selection sectors) implies that different components $V_K(P)$ of $V_K$ must fit together in a specific way, encoded by a graph that we now describe.

To every link $K$ we can associate the graph $\Gamma_K$ as follows.
The vertices of the graph correspond to the primitive partitions $P$ of $\{1,\ldots,\n\}$, corresponding to the ways of partitioning $K$ into non-split sublinks. The set of all such partitions is a subset of partitions of the set
$\{1,\ldots,\n\}$. 
Two vertices $P$ and $P'$ of the graph are connected by an edge if $P'$ is a refinement of
$P$ and there is no other graph vertex $P''$ between the two (where $P''$
is a refinement of $P$ and $P'$ is a refinement of $P''$).

We claim that the graph $\Gamma_K$ obtained in this way captures the geometry of $V_K$ as follows: for any pair of vertices of the graph that are connected by an edge, we conjecture that
$$
\textrm{codim}(V_K(P) \cap V_K(P')) =1,
$$
where codimension is counted inside either of the $\n$-dimensional varieties $V_K(P),V_K(P')$. Note that, based on counting dimensions, a generic intersection would be over points (codimension $\n$). The presence of an edge means that the intersection of $V_K(P)$ and $V_K(P')$ is as non-generic as possible without them coinciding. In a generic situation thus the graph would consist of vertices alone. The existence of a single function $\Psi_K$ implies that this is never the case for multi-component links.
Furthermore, it is natural to generalize our conjecture as follows: for any two primitive partitions $P$ and $P'$, the codimension of the intersection of $V_K(P)$ and $V_K(P')$ is at most the distance between the partitions, defined as the minimum number of edges in $\Gamma_K$ needed to pass from $P$ to $P'$:
$$
{\rm codim}(V_P\cap V_{P'}) \leq d(P, P').
$$

If all partitions are primitive and thus give vertices of the graph, then we can connect the coarsest partition\footnote{For convenience, we will sometimes abbreviate $P_c=(12\cdots \n)$ by $\n$ and $P_f = (1)\cdots(\n)$ by $1^\n$ when writing $V_K(P_c)$ or $V_K(P_f)$.}
$P_c=(12\cdots \n)$ with the finest partition $P_f = (1)\cdots(\n)$ by a sequence of
partitions $P_c=P_1$,$P_2,\ldots,P_\n=P_f$, where $P_i = (1)\cdots(i-1)(i\cdots
\n)$. 
Passing successively
from $P_1$ to $P_\n$ involves removing the knot components of $K$ one
by one, starting with the first component, until we have removed all
of them. We expect that for each $i$,
$V_K(P_i)$ and $V_K(P_{i+1})$ intersect in codimension $1$, and this
allows us to move between components of $V_K$ in such a way that each
move is between components whose intersection has codimension $1$. See
Section~\ref{ssec:augvarcomps} for further discussion for general links and a reformulation of this discussion in terms of the
augmentation variety.

Let us now explain why the  irreducibility of the D-module implies the
codimension condition.
A way to characterize the fact that $\Psi_K$ corresponds to an irreducible D-module is that, first of all, it satisfies a set of differential (or, more precisely, difference) equations generated by a finite set of operators
\beq\label{Wo}
{\cal A}_{\alpha} \Psi_K = 0,
\eeq
where each of ${\cal A}_{\alpha}$ is of the form
$$
{\cal A}_\alpha= \sum_{\k_1,\ldots,\k_\n}a_{\k_1, \ldots, \k_\n}(e^{x_1},\ldots , e^{x_\n}) e^{g_s \sum_i \k_i \partial_{x_i}}.
$$
where $a_{\k_1, \ldots, \k_\n}$ are polynomials in $e^{x_i}$. The fact that the equation \eqref{Wo} is linear implies that not only $\Psi_K$ but also each of $\Psi_K^P$ satisfy the equation.
It was proven in \cite{Garoufalidis}  that the colored HOMFLY polynomial of any link satisfies such a finite set of difference equations, and that the set is holonomic (or more precisely, $q$-holonomic). What this means is that, in the classical limit $g_s \rightarrow 0$,
$$
 {\cal A}_{\alpha}(e^{x_i}, e^{g_s \partial_j}) \qquad 
 \longrightarrow \qquad { A}_{\alpha}(e^{x_i}, e^{p_j}),
 $$
when $x_i$, $p_j = g_s \partial_{x_j}$ become numbers, the set of equations \eqref{Wo} parametrize a Lagrangian submanifold of ${\cal M}_\n$. This Lagrangian submanifold is just what we called $V_K$,
$$
V_K : \qquad \qquad {A}_{\alpha}(e^{x_i}, e^{p_j})=0.
$$
The fact that irreducible components of $V_K = \bigcup_P V_K(P)$ intersect over the varieties of codimension $1$ is a consequence of a general theory of systems of differential equations of this type \cite{kashiwara, miwa, kashiwarad}.\footnote{The system at hand is $q$-holonomic, rather than holonomic. The theorems we need are best developed in the ordinary holonomic case. However, we can view the $q$-holonomic system as holonomic, by working locally, so presumably the distinction is immaterial.}. We will discuss this point further in Section \ref{Sec:Dmodel}, where we discuss quantization of the system from the physical perspective.

\subsection{An example: Hopf link}
\label{ssec:Hopf}

 The simplest nontrivial 2-component link is the Hopf link $K=\Hopf=K_1\cup K_2$, with $K_1,K_2 = \bigcirc$. In this case, we expect two distinct fillings, corresponding to partitions $P=(12)$ (which we will abbreviate as $2$) and $P=(1)(2)$ (which we will abbreviate as $1^2$).
In the first case $P=(12)$,  Lagrangian filling of $\Lambda_\Hopf$ is the complement of the entire Hopf link in $S^3$,
$$
M_{\Hopf} = S^3- \Hopf,
$$
and this gives rise to $V_{\Hopf}(2)$. The filling is easy to see from the toric diagram, before and after the transition. The Hopf link is realized by taking two Lagrangian branes on opposite toric legs before the transition. They can each glue to one copy of the brane on $S^3$ to move off the $S^3$. The topology of the resulting Lagrangian is $T^2\times {\mathbb R}$. This has $b_1 = 2$, so the moduli space is $2$-dimensional. Moreover, there are no disk instanton corrections to it\footnote{This is the case because any holomorphic map with an $S^1$ boundary on the $T^2$ comes in an $S^1$ family. The fact that the Euler characteristic of the $S^1$ vanishes leads to vanishing of the corrections.}, so the potential is determined classically. It is easy to see that
\begin{equation}\label{eq:Hopfpotential}
U_{\Hopf}(p_1, p_2) = p_1 p_2,
\end{equation}
since the critical points associated to this simply state that the meridian of one knot in the Hopf link gets identified with the longitude of the other, once we glue $M_K$ from pieces, as we described. Namely, in terms of
$$
 e^{p_i} = \mu_i , \qquad  e^{x_i} = e^{\partial_{p_i} U_\Hopf}=\lambda_i,
$$
we have
$$
V_{\Hopf}(2): \qquad  \lambda_1 =  \mu_2, \qquad  \lambda_2= \mu_1.
$$

The other filling corresponding to $P=(1)(2)$ has two branes $L_{K_i}$, each simply a conormal to the unknot, probing the conifold geometry after the transition.
At the level of the disk potential, they cannot talk to each other, so
$$
U_{\Hopf}(1^2)(p_1,p_2) = U_{\bigcirc}(p_1)+U_{\bigcirc}(p_2),
$$
leading to a direct product of two copies of Riemann surfaces mirror to the unknot:
$$
V_{\Hopf}(1^2): \qquad Q- \lambda_1 - \mu_1 + \lambda_1 \mu_1=0, \qquad Q- \lambda_2 - \mu_2 + \lambda_1 \mu_1=0.
$$
The variety
$$
V_\Hopf = V_{\Hopf}(1^2) \cup V_\Hopf(2)
$$
coincides with the augmentation variety of the Hopf link, as we will see in the next section.

The same result can be obtained in a number of different ways, two of which we mention here. First, the result can be obtained by calculating the HOMFLY polynomials of the Hopf link, colored by totally symmetric representations, and taking the classical limit. We will show this in Appendix \ref{app:Hopf} in detail.
Second, as we will brush on in Section \ref{Sec:Dmodel}, one can obtain the same result by considering a pair of branes on the Riemann surface mirror to the unknot.

These two approaches allow one not only to recover the classical variety $V_\Hopf$, but also give a prediction for its quantization. Namely,
$$
\Psi_\Hopf(x_1, x_2) = \sum_{m_1, m_2} H_{m_1, m_2} e^{m_1 x_1+m_2 x_2}
$$
is given in terms of the $SU(N)_\n$ WZW S matrix
$$
H_{m_1,m_2} = S_{m_1, m_2}/S_{00}
$$
evaluated at  $q=e^{2\pi i/(\n+N)}$and $Q=q^N$. Here $S_{m_1m_2}$ is the matrix element corresponding to the totally symmetric representations with $m_1,m_2$ boxes. Using the well-known explicit expressions (see Appendix \ref{app:Hopf}), one can show that $\Psi_\Hopf$ satisfies a set of difference equations
$$
{\cal A}_{\alpha} \Psi_\Hopf =0, \qquad \alpha=1,2,3,
$$
where
\begin{align*}
{\cal A}_1 &=- e^{x_1}+ e^{x_2} - (1-Q e^{x_1})e^{g_s\partial_1} +
 (1-Qe^{x_2})e^{g_s\partial_2}, \\
{\cal A}_2&=  (1- q^{-1} e^{x_2} - (1- Q e^{x_2} )  e^{g_s\partial_2})( e^{g_s\partial_1} - e^{x_2}), \\
{\cal A}_3&=    (1- q^{-1} e^{x_1} - (1- Q e^{x_1} )  e^{g_s\partial_1})( e^{g_s\partial_2} - e^{x_1}).
\end{align*}
In the classical limit, these reduce to $A_{\alpha}=0$, where
\begin{align*}
{A}_1 &=- \lambda_1+\lambda_2- (1-Q \lambda_1)\mu_1+ (1-Q \lambda_2)\mu_2,\\
{A}_2&=  (1- \lambda_2-\mu_2   +
Q  \lambda_2\,\mu_2 )(\mu_1- \lambda_2), \\
{A}_3&=  (\mu_2- \lambda_1)(1- \lambda_1-\mu_1   +
Q  \lambda_1\,\mu_1 )(\mu_2- \lambda_1).
\end{align*}
The Lagrangian solutions to this are precisely the variety
$V_\Hopf = V_\Hopf(1^2)\cup V_\Hopf(2).$
(There are other solutions that do not lead to Lagrangians. These are not of interest, as they do not describe saddle points.)
Finally, note that the intersection
$V_\Hopf(1^2)\cap V_\Hopf(2)$
is indeed codimension $1$: it is a curve, given by
$$
\mu_2 = \lambda_1, \qquad \lambda_2 = \mu_1,
$$
where $\mu_1, \lambda_1$ lie on the unknot curve
$$
1- \lambda_1- \mu_1 + Q \lambda_1 \mu_1=0.
$$

The equations are simple enough that we can solve them exactly. Two linearly independent solutions, with different asymptotics, can be obtained by considering two different contours of integration ${\cal C}_{(2)}$, ${\cal C}_{(1^2)}$ in
$$
\Psi_\Hopf^P(x_1,x_2) = \int_{{\cal C}_P} dp_2\; {1\over 1-e^{-p_2} e^{x_1}} \Psi_{\bigcirc}(x_1) \Psi_{\bigcirc}^{-1}(p_2) \;e^{p_2 x_2/g_s}\\
$$
where $\Psi_{\bigcirc}(x)$ is the partition function of the unknot.\footnote{We have simplified things slightly by shifts of variables. For details see Appendix~\ref{app:Hopf}.}
The wave function corresponding to $V_\Hopf(2)$ is obtained by taking a contour ${\cal C}_{(2)}$ which is a small circle $p_2=x_1$; this gives
$$
\Psi_\Hopf^{(2)}(x_1,x_2) = e^{x_1 x_2/g_s}.
$$
The wave function corresponding to $V_{\Hopf}(1^2)$ can be obtained by taking the contour ${\cal C}_{(1^2)}$ to run along the real axis. Using the fact that Fourier transform $\Psi_{\bigcirc}^{-1}(p)$ is $\Psi_{\bigcirc}(x)$, this  gives
$$
\Psi_\Hopf^{(1^2)}(x_1, x_2) = \exp(W_{\bigcirc}(x_1)/g_s + W_{\bigcirc}(x_2)/g_s +\ldots).
$$

\subsection{Knot parallels and higher rank representations}
\label{ssec:knotparallels}
Consider a link ${\tilde K}$ obtained by taking $\n$ parallels of a single knot $K$. This is closely related to studying $\n$ branes on a single Lagrangian associated to the knot $K$. As explained in Section \ref{ssec:multiplebranes}, the partition functions are related by
$$
Z_{{\tilde K}}(x_1, \ldots, x_\n) =  Z^{(\n)}_{{ K}}(x_1, \ldots, x_\n)/\Delta(x),
$$
where
$$\Delta(x) = \prod_{1\leq i<j\leq \n}(e^{(x_i-x_j)/2}-e^{(x_j-x_i)/2})
$$
arises from integrating out short strings, with boundary on $L_K$ alone and with no boundaries on the $S^3$. This is a sum over annuli and hence has no $g_s$ dependence.

This relation between $Z_{{\tilde K}}$ and $Z^{(\n)}_{{ K}}$ implies that for every saddle point of $Z_{{\tilde K}}$ we get a saddle point of  $Z^{(\n)}_{{ K}}$
as well, at least as long $\Delta(x)$ has no zeroes there.
Thus we expect the saddle points of $Z_K$ to be a subset of saddle points of $Z_{{\tilde K}}$. The ones that may be missing are those where short bifundamental strings get expectation values. These are missing in  $ Z^{(\n)}_{{ K}}(x_1,\ldots,x_\n)$, where by assumption these strings are massive and we have integrated them out to get $\Delta(x)$. We will give a rigorous mathematical proof of this statement in Section \ref{Sec:knotch}, from the perspective of knot contact homology. It is satisfying that the same picture emerges from Chern-Simons theory and large $N$ duality.

\section{Knot contact homology}\label{Sec:knotch}
In this section we discuss the connection between knot contact homology and topological
strings in the context of knot and link invariants.  We first give a brief discussion of this connection
in the next subsection, before giving a more detailed discussion in the following subsections
and in Section \ref{sec:augmentations}.

\subsection{Brief summary of relation between topological strings and knot contact homology}
Knot contact homology uses the Lagrangian brane $L_K$ in $X=T^*S^3$, but
with the zero section $S^3\subset T^{\ast} S^{3}$ deleted. In other words, one considers the
geometry far from the apex of the cone where the geometry is $\R\times U^{\ast} S^{3}\approx \R\times S^3\times S^2$; this shares the basic feature of the geometry of $Y$
after the transition, where $S^2$ does not bound.

For simplicity of notation, we temporarily assume that $K$ is a single-component knot. Then
$L_K$ is $\R\times \Lambda_K\approx \R\times T^2$ and it is natural to view the $\R$-factor as the ``time'' direction. Here the torus $\Lambda_K$ can be viewed as the geometry of the Lagrangian
brane at infinity.  The knot leaves its imprint on how the Legendrian submanifold $\Lambda_K$ is embedded in the contact manifold $U^{\ast}S^{3}\approx S^3\times S^2$.

Consider the physical open string states ending on $\Lambda_K$. Among these curves with endpoints on $\Lambda_K$, the paths that are stationary for the action $\int p\,dq$ are ``Reeb chords'' $a_i$, which are string trajectories that are flow segments of the Reeb vector field. (Reeb chords also have the property that holomorphic maps can end on them.)  In the physical setup we would say that the $a_i$ are classically annihilated by the BRST symmetry $Q_B$:
$$
Q_B\cdot a_i=0.
$$
However, disk instantons modify the operation of $Q_B$ in a way similar to Witten's formulation
of Morse theory \cite{Witten:1982im}, where the critical points of the Morse function
correspond to vacua, but instanton corrections, which correspond
to gradient flows (rigid up to translation), modify the supersymmetry algebra.  In the case at hand the role
of the gradient flows are played by disk instantons, which are disks $D$ (rigid up to $\R$-translation) that at $t\rightarrow +\infty$ start with the Reeb chord $a_i$ and at $t\rightarrow -\infty$ approach
the Reeb chords $a_{j_1},\dots, a_{j_{k_D}}$ as ordered multi-pronged strips with $k_D+1$ punctures. The boundary of the disk maps to $\R\times\Lambda_K$ and hence (upon choosing some fixed paths capping the ends of the Reeb chords) represents an element in $H_1(\Lambda_K)$, and thus picks up the holonomy factor $\exp(l_D x+m_D p)$ from the Wilson lines on the Lagrangian brane.  Let $n_D$ denote
the intersection of this disk (again suitably capped off) with the 4-cycle dual to $S^2$, i.e.~$\R\times S^3$.  Then there is a deformed $Q_B$ operator (called $\partial$ below), which we interpret as the quantum corrected $Q_B$, given by
$$
Q_B\cdot a_i=\sum_{\text{rigid disks }D} Q^{n_D} e^{l_Dx+m_Dp} a_{j_1}\cdots a_{j_{k_D}},
$$
where $x,p$ correspond to the holonomies of the probe brane around the longitudinal/meridional directions of the knot. Thus the open string states correspond
to nontrivial elements of the deformed $Q_B$ cohomology.  The existence of
a $1$-dimensional representation of this cohomology leads to an algebraic constraint $A(e^x,e^p,Q)=0$.
Such $1$-dimensional representations can be interpreted as the disk corrected moduli space of one such Lagrangian brane after the large $N$ transition; see Section \ref{ssec:nonexact}, where it is explained how a Lagrangian filling with associated moduli space of disks gives an augmentation. Using the large $N$ duality, this should be
 the same as the $Q$-deformed $A$-polynomial we defined earlier, which is the quantum corrected moduli space of the single brane $L_K$.  The fact that we find agreement between these two
 setups can be interpreted as further evidence for large $N$ duality.

 The same setup applies
 to the case of links.  The main novelty there is that we have $\n$ Legendrian tori and thus
 we have $\n^2$ sectors of open strings associated to the $\n^2$ choices for the endpoint components of the Reeb chords. As in the knot case, we can use disk instantons to find the $Q_B$ corrected
 cohomology.  The representation theory of the resulting algebra, where there are
 nontrivial open string states between all pairs, leads to an $\n$-dimensional variety $V(\n)$,
 which we identify with the corresponding moduli space $V_K(\n)$ discussed earlier; see Section \ref{ssec:antibrane} for possible  geometric constructions of corresponding Lagrangian fillings.
 Moreover, we can consider other representations where part of the open strings between
 some pairs of Legendrian tori are mapped to zero.  This gives other varieties $V_K(P)$ as discussed
 in Section \ref{sec:largeN}.

We now turn our attention to a more mathematical description of knot contact homology. We begin with a
brief general introduction to contact homology, an object associated to
contact manifolds and to their Legendrian submanifolds. Next we
describe knot contact homology, which is the contact homology
associated to the Legendrian torus $\Lambda_K \subset U^*\R^3$ (which can be viewed as a patch of $U^*S^3=S^2\times S^3$), first
geometrically and then algebraically. We then present augmentations
from an algebraic perspective and use these to define a polynomial
knot invariant, the augmentation polynomial, or a variety in the case
of a multi-component link. The augmentation variety is conjectured to
agree with the mirror-symmetry variety $V_K$ from previous
sections. Besides agreeing
with $V_K$ in a number of examples, we will see that the augmentation
variety shares many properties of $V_K$. This will be discussed
further in Section~\ref{sec:augmentations}.

\subsection{Mathematical overview of knot contact homology}
\label{ssec:mathintro}

Here we provide a summary of the mathematics behind knot contact
homology and augmentations. This goes into more detail than the
previous subsection and
also places the construction in the context of contact and symplectic
geometry. The full story is rather long, and we will omit many
technical details
in the interest of readability. As a result, some of the discussion
below is rather imprecise, but we will provide references to more
detailed treatments in the
literature for the interested reader. A slightly less brief overview
can be found in \cite{Ngsurvey}.

Knot contact homology is a special case of \textit{Legendrian contact
  homology}, which is itself a small part of the more elaborate
Symplectic Field Theory package in symplectic topology introduced in
\cite{EGH}. The general setup begins with a contact manifold, which
for our purposes is a
$(2m-1)$-dimensional manifold $V$ equipped with a $1$-form $\alpha$
such that $\alpha \wedge (d\alpha)^{m-1}$ is a volume form on
$V$. The $1$-form $\alpha$ determines the Reeb vector field $R_\alpha$
on $V$ given by $d\alpha(R_\alpha,\cdot) = 0$ and $\alpha(R_\alpha) = 1$.
Associated to the contact manifold $V$ is its symplectization, the
$2m$-dimensional manifold $\R \times V$ equipped with the symplectic
form $\omega = d(e^t\alpha)$, where $t$ is the coordinate in the
additional $\R$-factor; on the symplectization we may choose an
$\R$-invariant almost complex structure $J$ pairing $\partial/\partial
t$ and the Reeb vector field $R_\alpha$.

Legendrian contact homology is associated to the contact manifold $V$
along with a \textit{Legendrian submanifold} $\Lambda \subset V$,
which is a submanifold along which $\alpha$ is identically $0$, of
maximal dimension; the contact condition on $\alpha$ forces this
maximal dimension to be $m-1$. In this case, $\R\times\Lambda$ is a
Lagrangian submanifold of $\R\times V$.

Define a \textit{Reeb chord}
of $\Lambda$ to be a flowline of $R_\alpha$ that begins and ends on $\Lambda$.
We define the algebra $\mathcal{A}=\mathcal{A}(\Lambda)$ to be the free (tensor)
algebra over the ring $\Z[H_2(V,\Lambda)]$ generated by Reeb chords of
$\Lambda$. That is, if $a_1,\ldots,a_r$ are the Reeb chords of
$\Lambda$, then an element of $\mathcal{A}$ is a linear combination of
monomials of the form
\[
e^{\gamma} a_{i_1}\cdots a_{i_k}
\]
where $1\leq i_1,\ldots,i_k\leq r$ and $\gamma \in
H_2(V,\Lambda)$. (Here the $a_i$'s do not commute with each other,
although for many purposes one can abelianize and replace
$\mathcal{A}$ by the polynomial ring generated by $a_1,\ldots,a_r$.)
The algebra $\mathcal{A}$ has a grading (by Conley--Zehnder indices)
that we do not describe here.

\begin{figure}
\centerline{
\includegraphics[width=0.75\linewidth]{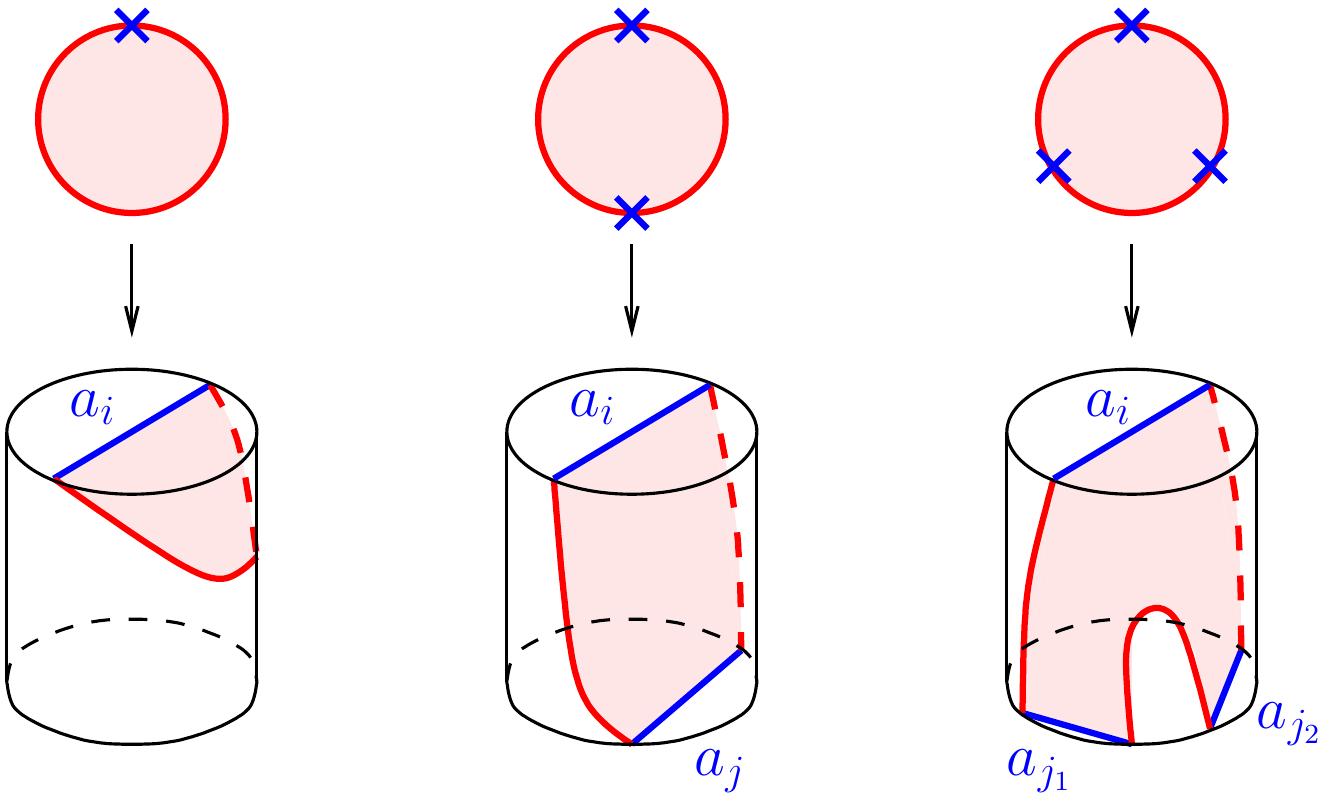}
}
\caption{
Holomorphic disks contributing terms $1$, $a_j$, $a_{j_1j_2}$,
respectively, to $\partial(a_i)$. The cylinders represent
$\R\times\Lambda \subset \R\times V$, with top and bottom
corresponding to $+\infty$ and $-\infty$ in $\mathbb{R}$, respectively.
}
\label{fig:holdisks}
\end{figure}

We can construct a differential map $\partial
:\thinspace \mathcal{A} \to \mathcal{A}$ by counting certain
$J$-holomorphic curves in $\R\times V$ with boundary on
$\R\times\Lambda$. More precisely, for $1\leq i\leq r$, define
\[
\partial(a_i) = \sum_{k\geq 0} \sum_{\Delta} (\operatorname{sgn} \Delta) e^{[\Delta]} a_{j_1} \cdots
a_{j_k},
\]
where the second sum is over all rigid (i.e., $0$-dimensional moduli space)
$J$-holomorphic maps $\Delta$ from a disk
with $k+1$ punctures along its boundary to $\R\times V$, such that the
boundary of the disk is mapped to $\R\times\Lambda$ and the punctures
are mapped in order to a neighborhood of $a_i$ near $t=\infty$, and
neighborhoods of $a_{j_1},\ldots,a_{j_k}$ near $t=-\infty$; $(\operatorname{sgn}
\Delta)$ is a sign associated to $\Delta$; and $[\Delta]$ is the
homology class of $\Delta$ in $H_2(V,\Lambda)$. See Figure~\ref{fig:holdisks}.

In appropriate circumstances, $\partial^2=0$,
$\partial$ lowers degree by $1$, and
the homology of
the differential graded algebra $(\mathcal{A},\partial)$ is an invariant of
the Legendrian submanifold $\Lambda$, called the \textit{Legendrian
  contact homology} of $\Lambda$. (See e.g.~\cite{EES, Ekholm_rsft}. From a physical perspective, the main necessary condition is that the closed string sector, Reeb orbits, decouples from the open string sector, Reeb chords.)

  As mentioned in previous sections,
there is a strong parallel to open
topological strings, where $\R\times\Lambda$ is a brane, $\mathcal{A}$
keeps track of open topological strings on $\R\times\Lambda$, and
$\partial$ is the BRST operator $Q_B$.

We now apply this general setup to one particular case of interest. Let
$V$ be the bundle $U^*S^3$ consisting of unit-length cotangent
vectors in $S^3$; this is a five-dimensional contact manifold with
$\alpha = p\,dq$. If $K \subset S^3$ is a link, then define the unit
conormal bundle $\Lambda_K
\subset U^*S^3$ to be the subset consisting of all unit covectors in
$T^*S^3$ lying over $K$ and annihilating $TK$. It is an exercise to
check that $\Lambda_K$ is Legendrian in $U^*S^3$. We can then define
the \textit{knot contact homology} of $K$ to be the homology of the
differential graded algebra $(\mathcal{A},\partial)$ defined above
with $\Lambda=\Lambda_K$.

If $K$ is an $n$-component link $K_1 \cup \cdots \cup K_n$, then
$\Lambda_K = \Lambda_{K_1} \cup \cdots \cup \Lambda_{K_n}$ is a
disjoint union of $n$ two-dimensional tori, and we
can write
\begin{align*}
H_2(V,\Lambda) &= H_2(S^3\times S^2,\Lambda_K) \cong H_2(S^2) \oplus
H_1(\Lambda_K)\\
& \cong \Z \oplus \Z^{2n} \cong \Z\langle t,x_1,p_1,\ldots,x_n,p_n \rangle.
\end{align*}
Here $t$ is the generator of $H_2(S^2)$ and $x_i,p_i$ are generators
of $H_1(\Lambda_{K_i}) = H_1(T^2)$ such that $p_i$ is a meridian and
$x_i$ is a longitude for $K_i$ with respect to the standard Seifert
framing (so that $x_i$ and $K_i$ have linking number $0$).
Correspondingly, we can write the coefficient ring for $\mathcal{A}$
as
\[
\Z[H_2(V,\Lambda)] \cong \Z[Q^{\pm 1},\lambda_1^{\pm 1},\mu_1^{\pm 1},\ldots,
\lambda_n^{\pm 1},\mu_n^{\pm 1}]
\]
where $Q = e^t$, $\lambda_i = e^{x_i}$, and $\mu_i = e^{p_i}$.

Given a braid with $m$ strands such that $K$ is the closure of the braid (i.e., the result of gluing together corresponding ends of the braid in $S^3$), there is a combinatorial formula for the differential graded algebra $(\mathcal{A},\partial)$ associated to $K$. Outside of computations, we will not need the precise formula, and we omit it here; please see the Appendix of \cite{Ngsurvey} for the full definition of the version of the invariant that we use in this paper, noting that our $Q$ is denoted by $U$ in that paper. (The algebra given in \cite{Ngsurvey} is in turn based on work that originally appeared in the series of papers \cite{Ngframed,EENStransverse,Ngtransverse,EENS}.) For our purposes, it suffices to note that $\mathcal{A}$ is finitely generated as an algebra over the above coefficient ring $\Z[Q^{\pm 1},\lambda_1^{\pm 1},\mu_1^{\pm 1},\ldots,
\lambda_n^{\pm 1},\mu_n^{\pm 1}]$, and the generators of $\mathcal{A}$ are all of degree $\geq 0$. These generators include $m(m-1)$ generators of degree $0$ that we write as $a_{ij}$ for $1\leq i,j\leq m$ and $i\neq j$.

As an example, consider the Hopf link, which is represented by the
$2$-strand braid $\sigma_1^2$. For this link, $\mathcal{A}$ is
generated by the following generators: $a_{12},a_{21}$ of degree $0$;
$b_{12},b_{21},c_{11},c_{12},c_{21},c_{22}$ of degree $1$; and a
number of other generators of degree $1$ and $2$ that are irrelevant
for our purposes. From the formula from \cite{Ngsurvey}, the
differential $\partial$ is given by:
\begin{align*}
\partial(a_{12}) &= \partial(a_{21}) = 0 \\
\partial(b_{12}) &= \lambda_1\mu_1\lambda_2^{-1}\mu_2^{-1}a_{12}-a_{12} \\
\partial(b_{21}) &= a_{21} - \lambda_1^{-1}\mu_1^{-1}\lambda_2\mu_2 a_{21} \\
\partial(c_{11}) &= Q-\lambda_1-\mu_1+\lambda_1\mu_1+\lambda_1\mu_1\mu_2^{-1}
a_{12}a_{21} \\
\partial(c_{12}) &= Q a_{12} + \lambda_1 \mu_2^{-1} a_{12}(\mu_1\mu_2-\mu_1-\mu_2+\mu_1a_{21}a_{12}) \\
\partial(c_{21}) &= \lambda_2 a_{21}-\mu_1 a_{21} \\
\partial(c_{22}) &= Q-\lambda_2-\mu_2+\lambda_2\mu_2+\lambda_2 a_{12}a_{21}.
\end{align*}
We will use this computation in the next subsection.

\subsection{Augmentations, the augmentation variety, and the
  $A$-polynomial}
\label{ssec:augmentations}

Having defined the differential graded algebra for knot contact
homology, we now turn to augmentations. In general, an \textit{augmentation} of a differential graded algebra $(\mathcal{A},\partial)$ is an algebra map $\epsilon$ (preserving multiplication) from $\mathcal{A}$ to a ring $S$, such that $\epsilon(1)=1$, $\epsilon\circ\partial = 0$, and $\epsilon$ vanishes on elements of nonzero degree. Here we will take our augmentations to be graded, where we think of the ring $S$ as a DGA with trivial differential concentrated in degree zero. Thus, our augmentations $\epsilon$ are nonzero only on generators of degree zero.

Augmentations arise naturally in symplectic topology in the context of symplectic fillings. Suppose that $V$ is a contact manifold and $W$ is a symplectic filling of $V$; this is a symplectic manifold whose boundary is $V$, satisfying a certain compatibility condition ($V$ should be a convex end of $W$). We further assume that $W$ is an \textit{exact} symplectic filling if the contact $1$-form $\alpha$ on $V$ extends to a $1$-form $\alpha$ on $W$ such that $\omega = d\alpha$, where $\omega$ is the symplectic form on $W$.
Next let $\Lambda \subset V$ be Legendrian, and suppose that $L \subset W$ is a Lagrangian filling of $\Lambda$; this is an embedded Lagrangian submanifold whose boundary is $\Lambda$ (and which looks like $\R\times\Lambda$ near $V$).

A special case of this construction is when $L$ is an \textit{exact}
Lagrangian filling of $\Lambda$, i.e., when $\alpha|_L$ has a
primitive on $L$. In this case, we can associate to $L$ a canonical
augmentation $\epsilon$ of the DGA for $\Lambda$,
\[
\epsilon :\thinspace \mathcal{A} \to \Z[H_2(W,L)],
\]
obtained by counting holomorphic disks in $W$ with
boundary on $L$ and asymptotic to a Reeb chord for $\Lambda$.
More generally, if $L$ is non-exact, we can use $L$ to deduce certain
structures on the space of abstract augmentations of $\Lambda$. This
perspective relating augmentations to Lagrangian fillings is central
to the mathematical side of this paper, and we return to it in detail
in Section~\ref{sec:augmentations}.

For now, we examine the space of augmentations in our particular
case. When $K$ is an $n$-component link and $(\mathcal{A},\partial)$
is its differential graded algebra with coefficient ring
$\mathbb{Z}[\lambda_1^{\pm 1},\mu_1^{\pm 1},\ldots,
\lambda_n^{\pm 1},\mu_n^{\pm 1},Q^{\pm 1}]$, an augmentation to $\C$
is a (graded) algebra map $\epsilon :\thinspace \mathcal{A} \to \C$ such that
$\epsilon(\partial(a)) = 0$ for all Reeb chords $a$.
Note that an augmentation restricts to a ring homomorphism from the
coefficient ring of $\mathcal{A}$ to $\C$.

With this in mind, define the \textit{augmentation variety} of the
$n$-component link $K$ to be the subset of $(\C^*)^{2n+1} =
(\C-\{0\})^{2n+1}$ given by
\[
V_K = \{(\epsilon(\mu_1),\ldots,\epsilon(\mu_n),\epsilon(\lambda_1),\ldots,\epsilon(\lambda_n),\epsilon(Q))\,|\,
\epsilon\colon\mathcal{A}\to\bC \text{ augmentation}\}.
\]
(More precisely, the augmentation variety is the closure of the
maximal-dimensional piece of this set.) This variety is
the set of ways to assign
nonzero complex numbers to $\mu_1,\ldots,\mu_n,\lambda_1,\ldots,\lambda_n,Q$
such that a collection of polynomials in $m(m-1)$ variables
$a_{ij}$ (with coefficients involving
$\mu_1,\ldots,\mu_n,\lambda_1,\ldots,\lambda_n,Q$), namely the images under $\partial$ of the generators of degree one, has a common root.
Note that the augmentation variety is algebraic (it is cut out in
$(\C^*)^{2n+1}$ by the zero locus of some collection of polynomials),
and that it is a link invariant since the DGA for knot contact homology is invariant up to homotopy.
Also, in line with previous sections, one can view the
augmentation variety as a one-parameter
family of varieties in $(\bC^*)^{2n}$ with parameter given by $Q$.

As an example, for the Hopf link, the expression for the differential
from Section~\ref{ssec:mathintro} implies that any augmentation
$\epsilon$ must satisfy either $\epsilon(a_{12})=\epsilon(a_{21})=0$,
in which case
\[
\epsilon(Q)-\epsilon(\lambda_1)-\epsilon(\mu_1)+\epsilon(\lambda_1\mu_1)
=
\epsilon(Q)-\epsilon(\lambda_2)-\epsilon(\mu_2)+\epsilon(\lambda_2\mu_2)
=0,
\]
or $\epsilon(a_{12}),\epsilon(a_{21}) \neq 0$, in which case
$\epsilon(\lambda_2)=\epsilon(\mu_1)$ and
$\epsilon(\lambda_1)=\epsilon(\mu_2)$. It follows that
\[
V_\Hopf = \{Q-\lambda_1-\mu_1+\lambda_1\mu_1 =
Q-\lambda_2-\mu_2+\lambda_2\mu_2 = 0\} \cup
\{\lambda_2-\mu_1=\lambda_1-\mu_2=0\}.
\]
Further examples of augmentation varieties are given in
Section~\ref{sec:ex}.

As observed in \cite{Ngframed}, there is a close connection between
the augmentation variety and the $A$-polynomial, which we now
discuss. First assume that $K$ is a single-component knot, and let
$m,l\in \pi_1(S^3- K)$ denote the meridian and longitude of $K$.
For a
representation $\rho :\thinspace \pi_1(S^3- K) \to SL(2,\bC)$,
simultaneously diagonalize
$\rho(m)$ and $\rho(l)$, and let $(\mu,\lambda)$ and
$(\mu^{-1},\lambda^{-1})$ be the corresponding eigenvalues. Then the
closure of the (highest-dimensional part of) the set
$\{(\mu,\lambda)\}$ over all $SL(2,\bC)$ representations is a curve
in $(\bC^*)^2$ whose defining equation is the $A$-polynomial
$A_K(\lambda,\mu)$.

When $K$ is a knot, the augmentation variety $V_K \subset (\bC^*)^3$
is also the
vanishing locus of a polynomial, the \textit{augmentation polynomial}
$\operatorname{Aug}_K(\lambda,\mu,Q)$. Similarly, the slice $V_K \cap
\{Q=1\} \subset (\bC^*)^2$ is the vanishing locus of the \textit{two-variable
augmentation polynomial} $\operatorname{Aug}_K(\lambda,\mu)$. It is
conjectured that $\operatorname{Aug}_K(\lambda,\mu)$ is equal to
$\operatorname{Aug}_K(\lambda,\mu,Q=1)$ up to repeated factors.

The relation between the augmentation variety and the $A$-polynomial
is then as follows \cite[Prop.~5.9]{Ngframed}:
\[
(\mu-1) A_K(\lambda,\mu^{1/2}) \, | \, \operatorname{Aug}_K(\lambda,\mu).
\]
Here ``$|$'' denotes ``divides''. In particular, $(\mu-1)$ and
$(\lambda-1)$ are both factors of $\operatorname{Aug}_K(\lambda,\mu)$
for any knot $K$, where the latter comes from reducible $SL(2,\bC)$
representations. These two factors have a nice interpretation in terms
of Lagrangian fillings: $(\mu-1)$ comes from the filling of
$\Lambda_K$ by the conormal Lagrangian $L_K$, while $(\lambda-1)$
comes from the filling by the knot complement Lagrangian $M_K$.

We now generalize to the case where $K$ is an $\n$-component link.
There is a well-known generalization of the $A$-polynomial to links,
as follows. As before, let $\rho :\thinspace
\pi_1(S^3- K) \to
SL(2,\bC)$ be a representation, and let $m_i,l_i$ denote the meridian
and longitude for component $i$. For any single $i=1,\ldots,\n$, $\rho(m_i)$ and $\rho(l_i)$
can be simultaneously diagonalized, with eigenvalues $(\mu_i,\lambda_i)$ and
$(\mu_i^{-1},\lambda_i^{-1})$. Then
(the closure of the highest-dimensional part of)
\[
\{(\mu_1,\ldots,\mu_\n,\lambda_1,\ldots,\lambda_\n) \,|\,
\rho \text{ any } SL(2,\bC) \text{ representation}\}
\]
is a variety in $(\bC^*)^{2\n}$, which we denote by $V_K^A$ and call
the \textit{$A$-polynomial variety} of $K$.\footnote{This is not new; see, e.g., \cite{Tillmann}.}

Note that for any fixed $i=1,\ldots,\n$, one can abelianize
$\pi_1(S^3- K)$ and set all meridians but $m_i$ to $0$, to
obtain $\bZ$. It follows that for any $\mu_i\in\bC^*$, the curve
\[
C_i = \{(\mu_1,\ldots,\mu_\n,\lambda_1,\ldots,\lambda_\n,1) \,|\,\mu_j= 1
\text{ and }
\lambda_j=\mu_i^{a_{ij}} \text{ for
} j\neq i, \text{ and } \lambda_i=1 \}
\]
is contained in $V_K$. This is the analogue of the fact that
$(\lambda-1)$ divides the $A$-polynomial for knots.

Arguing along the lines of Proposition 5.9 from \cite{Ngframed} then yields the
following result: if
$(\mu_1,\ldots,\mu_\n,\lambda_1,\ldots,\lambda_\n) \in V_K^A$, then
\[
(\mu_1^2,\ldots,\mu_\n^2,\tilde{\lambda}_1,\ldots,\tilde{\lambda}_\n,1)
\in V_K.
\]
Here we define $\tilde{\lambda}_i = \lambda_i \mu_1^{\k_{i1}} \cdots
\mu_\n^{\k_{i\n}}$, where the $\k_{ij}$ are integers determined by
$[l_i] = \sum_{j=1}^\n \k_{ij}[m_j] \in H_1(S^3- K)$.
Thus the $A$-polynomial variety of $K$ can be viewed as a subset at
$Q=1$ of the augmentation variety, for links as well as for knots.

We conclude this section with a remark about the nature of the
augmentation variety. For all examples where the variety has been
computed, the following holds: for fixed $Q$, $V_K \subset (\bC^*)^{2\n}$
is Lagrangian with respect to the symplectic form
\[
\sum_{i=1}^n \frac{1}{\lambda_i\mu_i} \,d\lambda_i \wedge
d\mu_i.
\]
In $x,p$ coordinates, this says that the augmentation variety is
Lagrangian with respect to the symplectic form $\sum_{i=1}^n dx_i
\wedge dp_i$, in accordance with the notion that the variety is
locally given by a potential. (In particular, $V_K$ is $\n$-dimensional.) The underlying geometric reason for this observation is related to Lagrangian fillings and will be discussed further in Section \ref{sec:augmentations}
from the viewpoint of knot contact homology; it agrees well with corresponding
predictions from large $N$
duality and the relation with Chern-Simons theory. See
Section~\ref{sec:largeN}.

\subsection{Components of the augmentation variety}
\label{ssec:augvarcomps}

Here we discuss how partitions give different components of the
augmentation variety of a link, in parallel to the physics discussion from
Section~\ref{sec:Dmirror}. As discussed there, for an $\n$-component link $K$, any partition $P$ of $\{1,\ldots,\n\}$ divides $K$ into a collection of sublinks; if we move these sublinks far from each other, we obtain a split link $K(P)$, and a partition is primitive if each of these sublinks is non-split. (Recall that a link is split if it is the union of two sublinks that lie in two disjoint solid balls.)

\begin{conjecture}\label{cnj:auglink}
The augmentation variety $V_K$  of a link $K$ is a union of components
$$
V_K = \bigcup_P V_K(P)
$$
labeled by the primitive partitions $P$ of $\{1,\ldots,\n\}$. Furthermore, $V_K(P)$ coincides with the augmentation variety of the link $K(P)$.
\end{conjecture}
In the rest of the subsection we will provide evidence for the conjecture.
First, we will prove that every partition $P$ of $\{1, \ldots, \n\}$ leads to a (possibly empty) component of the augmentation variety of $K$, which we label $V_K(P)$ in accordance with Section~\ref{sec:Dmirror}.  Next, we will provide evidence that only the primitive partitions $P$ lead to nonempty components of augmentation variety, and that moreover these are all components of $V_K$.

If we view a 
link $K$ as a union of two sublinks $K_1
\cup K_2$, then
the fiber product
\[
V_{K_1} \times_Q V_{K_2} =
\{(\vec{\mu}_1,\vec{\mu_2},\vec{\lambda}_1,\vec{\lambda}_2,Q) \,|\,
(\vec{\mu}_1,\vec{\lambda}_1,Q) \in V_{K_1} \text{ and }
(\vec{\mu}_2,\vec{\lambda}_2,Q) \in V_{K_2}\}
\]
is contained in $V_K$.
Here $\vec{\mu}_1 = (\mu_1,\ldots,\mu_{\n_1})$, $\vec{\mu}_2 =
(\mu_{\n_1+1},\ldots,\mu_\n)$, and similarly for $\vec{\lambda}_1$ and
$\vec{\lambda}_2$, where the first $\n_1$ components of $K$ belong
to $K_1$ and the remainder to $K_2$.
This follows from considering
augmentations of the DGA for $K$ that send all ``mixed Reeb chords''
$a_{ij}$, i.e., Reeb chords with one endpoint in $\Lambda_{K_1}$ and the other in
$\Lambda_{K_2}$, to $0$. The idea of separating the DGA for a
Legendrian link into the DGAs of sublinks by sending mixed chords to
$0$ is well-established in contact geometry; see e.g. \cite{Mishachev}.

By iterating this process, one sees that if $P$ is a partition of
$\{1,\ldots,\n\}$ consisting of $\ell$ subsets, then augmentations of
(the DGAs for)
the corresponding $\ell$ sublinks of $K$ produce an augmentation of
$K$, where mixed Reeb chords connecting different sublinks are sent to
$0$. Say that an augmentation of $K$ obtained in this way is
\textit{associated to} the partition $P$. Note that
if $P'$ is a
refinement of $P$, then any augmentation associated to $P'$ is also
associated to $P$, and in particular that all augmentations are
associated to the partition $(12\cdots \n)$.

Say that an augmentation
associated to a partition $P$ is
\textit{irreducible} if it is not associated to any refinement of
$P$. Then we can define $V_K(P)$ to be the subset of the augmentation
variety $V_K$ corresponding to irreducible augmentations associated to
$P$. If $P$ consists of $\ell$ subsets, then $V_K(P)$ is a fiber
product of the $\ell$ augmentation varieties associated to the
relevant sublinks. As in Section~\ref{sec:Dmirror}, we sometimes abbreviate $V_K((1\cdots \n))$ to
$V_K(\n)$ and $V_K((1)\cdots(\n))$ to $V_K(1^\n)$.

For example, for the Hopf link, from the computation in
Section~\ref{ssec:augmentations}, we have $V_\Hopf = V_\Hopf(2) \cup V_\Hopf(1^2)$, where
\begin{eqnarray*}
V_\Hopf(2) &=& \{\lambda_1-\mu_2 = \lambda_2-\mu_1 = 0\} \\
V_\Hopf(1^2) &=& \{Q-\lambda_1-\mu_1+\lambda_1\mu_1 =
Q-\lambda_2-\mu_2+\lambda_2\mu_2 = 0\}.
\end{eqnarray*}

Under certain circumstances, $V_K(P)$ may be forced to be
empty:

\begin{proposition}
Suppose that $K = K_1 \cup K_2$ is a split link;
\label{prop:split}
that is, there is
an embedded $S^2$ in $S^3$ that separates $K_1$ from $K_2$. Then any
augmentation of the DGA of $K$ splits into augmentations of the DGAs
of $K_1$ and $K_2$, and thus $V_K = V_{K_1} \times_Q V_{K_2}$.
\end{proposition}

\begin{proof}
If $K$ is split, then it can be given as the closure of a braid $B$
that is also split. In this case, the matrices $\Phi^L,\Phi^R$ from
the definition of the DGA for knot contact homology are
block-diagonal, and the result follows from the formula for the
differential.

A more geometric proof is as follows: the split link is a connected sum $(S^{3}, K_1)\#(S^{3},K_2)$ and the unit disk cotangent bundle is then obtained by joining the two unit disk cotangent bundles with the corresponding Weinstein $1$-handle. This shows that there exists a contact form on $U^{\ast}S^{3}$ for which there are no Reeb chords connecting $\Lambda_{K_1}$ to $\Lambda_{K_2}$, and the result follows immediately.
\end{proof}

By Proposition~\ref{prop:split}, if $K = K_1 \cup K_2$ is split and
$P$ is the corresponding two-set partition of $\{1,\ldots,\n\}$, then
$V_K(P) = \emptyset$. More generally, if $K$ is a general link and $P$
is an $\ell$-set partition for which one of the $\ell$ sublinks is
split, then $V_K(P) = \emptyset$. We are not currently aware of any
other circumstance in which $V_K(P)$ is empty; see also the discussion
in Section~\ref{sec:Dmirror}.

\subsection{The link graph $\Gamma_K$ and the augmentation variety}
\label{ssec:augvarcomps2}
As in Section~\ref{sec:Dmirror}, to an $\n$-component link $K$ we can associate a
graph $\Gamma_K$ whose vertices are the primitive partitions of the link $K$, and where $P$ and $P'$ are connected if one
is a refinement of the other, and there is no intermediate refinement
$P''$ between the two such that $P''$ is also primitive. Based on
computations of the augmentation variety for a number of links, as
well as the D-module argument from Section~\ref{sec:Dmirror}, we make
the following conjecture.

\begin{conjecture}\label{cnj:codim1}
The graph associated to a link is connected, and if $P$ and $P'$ are
connected by an edge, then \label{conj:intersection}
\[
\textrm{codim}(V_K(P) \cap V_K(P')) = 1.
\]
More generally, if $P_1,\ldots,P_r$ are vertices such that $P_{i+1}$
is a refinement of $P_i$ and $P_i,P_{i+1}$ are joined by an edge for
all $1\leq i\leq r-1$, then $\textrm{codim}(V_K(P_1) \cap V_K(P_r))
\leq r$.
\end{conjecture}

There is one important special case where we can recast this
conjecture in terms of ``Lagrangian reduction''. (In Section \ref{ssec:antibrane} we discuss Conjecture \ref{cnj:codim1} and Lagrangian reduction from a geometric perspective that physically corresponds to adding an anti-brane to a filling brane.)
Suppose that $P =
(1\cdots \n)$ and $P' =
(1\cdots (\n-1))(\n)$. Write $K = (K-K_\n) \cup K_\n$ where $K_\n$ is the
$\n$-th component of $K$ and $K-K_\n$ is the rest. Suppose that both
$V_K(\n) = V_K(P)$ and $V_K(P') = V_{K-K_\n}(k-1) \times_Q V_{K_\n}(1)$
are nonempty. According to the remark at the end of
Section~\ref{ssec:augmentations}, for fixed $Q$, $V_K(\n)$ is a
Lagrangian submanifold of $(\bC^*)^{2(\n-1)}$ with the appropriate
symplectic structure, as is $V_{K_\n}(1) \subset (\bC^*)^2$ (trivially).

Suppose more generally that we have a product symplectic manifold $W =
W_1 \times W_2$, where $W_1$
and $W_2$ are symplectic, and Lagrangian submanifolds $L \subset W$
and $L_2 \subset W_2$. If $W_1 \times L_2$ is transverse to $L$ in
$W$, then the projection of $L \cap (W_1\times L_2) \subset W$ in
$W_1$ is also a Lagrangian submanifold that we call the reduction of
$L$ along $L_2$.

In our case (with $W_1 = (\bC^*)^{2\n-2}$, $W_2 = (\bC^*)^2$, $W = (\bC^*)^{2\n}$), we have the following conjecture.

\begin{conjecture} \label{conj:reduction}
$V_{K-K_\n}(\n-1)$ is a subset of the reduction of $V_{K}(\n)$ along $V_{K_\n}(1)$.
Concretely, we have:
\begin{align*}
V_{K-K_\n}(k-1) &\subset \{(\vec{\mu}_1,\vec{\lambda}_1,Q) \in (\bC^*)^{\n-1} \times (\bC^*)^{\n-1} \times \bC^* \,|\, \\
&\qquad \text{there exist } \mu_2,\lambda_2\in\bC^* \text{ with } (\mu_2,\lambda_2,Q) \in V_{K_\n}(1) \text{ and } \\
&\qquad (\vec{\mu}_1,\mu_2,\vec{\lambda}_1,\lambda_2,Q) \in
V_K(\n) \}.
\end{align*}
\end{conjecture}

Conjecture~\ref{conj:reduction} has been verified in all computed examples.
Note that it is not necessarily the case that $V_{K-K_\n}(\n-1)$ is precisely equal to, rather than just a subset of, the reduction of $V_K(\n)$ along $V_{K_\n}(1)$: see Section~\ref{ssec:toruslinks}.

For our particular choice of $P = (1\cdots \n)$ and $P' = (1\cdots (\n-1))(\n)$, it would follow immediately from Conjecture~\ref{conj:reduction} that $V_K(P)$ and $V_K(P')$ have codimension-1 intersection, since $V_{K-K_\n}(\n-1)$ is $(\n-1)$-dimensional. This is a key special case of Conjecture~\ref{conj:intersection}.

\subsection{Higher rank representations of knot contact homology}
\label{ssec:higherrank}

Here we trace out the mathematical story that parallels the discussion of multiple branes in Section~\ref{ssec:multiplebranes}.
Augmentations of a differential graded algebra $(\mathcal{A},\partial)$, such as the DGA for knot contact homology, are maps $\epsilon :\thinspace \mathcal{A} \to \C$, and can thus be viewed as one-dimensional representations of the DGA. One can also consider a generalization to arbitrary dimension: define a \textit{rank $\n$ representation} of a DGA $(\mathcal{A},\partial)$ to be an algebra map
\[
\rho :\thinspace \mathcal{A} \to \operatorname{End} \C^\n
\]
such that $\rho\circ\partial = 0$, $\rho$ sends any element of nonzero degree to $0$, and $\rho(Q)$ is a scalar multiple of the identity. (Here $\operatorname{End} \C^\n$ is the space of endomorphisms of $\C^\n$, or $\n\times \n$ matrices.)

As a technical note, it is important for rank $\n>1$ representations that we replace the DGA $(\mathcal{A},\partial)$ described in previous subsections, in which $\lambda,\mu$ commute with all Reeb chords, with a slight variant, the \textit{fully noncommutative DGA}, in which $\lambda,\mu$ commute with each other (and with $Q$) but not with Reeb chords. This allows $\rho(\lambda)$ and $\rho(\mu)$ to be something besides scalar multiples of the identity, even in an irreducible representation. See \cite{Ngsurvey} for a description of the fully noncommutative DGA associated to a knot.

Since $\lambda$ and $\mu$ commute with each other, the matrices
$\rho(\lambda)$ and $\rho(\mu)$ can be simultaneously
diagonalized. Say that a rank $\n$ representation of
$(\mathcal{A},\partial)$ is \textit{diagonal} if $\rho(\lambda)$ and
$\rho(\mu)$ are both diagonal matrices.

Our goal here is to relate diagonal higher rank representations of the knot
contact homology of a single-component knot $K$ to augmentations of
a link consisting of multiple parallel copies of $K$. This follows
from a more general result about higher rank representations of
Legendrian submanifolds, which we now describe.

Let $\Lambda$ be a Legendrian submanifold in a contact manifold
$V$. The Legendrian neighborhood theorem states that a tubular
neighborhood of $\Lambda$ is contactomorphic to (i.e., diffeomorphic
to with the same contact structure as) the $1$-jet space $J^1(\Lambda)
= T^*\Lambda \times \R$. Here $J^1(\Lambda)$ is equipped with the
contact $1$-form $dz-p\,dq$, where $z$ is the $\R$ coordinate and
$p\,dq$ is the Liouville $1$-form on $T^*\Lambda$. Define a
\textit{parallel} to $\Lambda$ to be a Legendrian $\Lambda' \subset V$
such that $\Lambda'$ lies in a small neighborhood of $\Lambda$ and
$\Lambda'$ corresponds to the graph of a function in $J^1(\Lambda)$:
there is some $f :\thinspace \Lambda\to\R$ such that
\[
\Lambda' = \{(q,p,z) \in J^1(\Lambda) \,:\,
z=f(q),~p=df_q\}.
\]
Finally, let \textit{$\n$ parallel copies} of $\Lambda$ denote an $\n$-component Legendrian submanifold of $V$ given by the union of
$\Lambda$ and $\n-1$ parallels of $\Lambda$, where these $\n$ components
are pairwise disjoint. (This is a generalization of the notion of an
$N$-copy from \cite{Mishachev}. Note that our notion of parallel
copies is not unique.)

Let $\widetilde{\Lambda} = \Lambda_1 \cup \Lambda_2 \cup \cdots \cup
\Lambda_\n$ be an $\n$-component Legendrian submanifold of $V$ given by
$\n$ parallel copies of $\Lambda = \Lambda_1$. The Reeb chords of
$\widetilde{\Lambda}$ can be divided into two types: \textit{short}
chords, which lie within the tubular neighborhood of $\Lambda$, and
\textit{long} chords, which do not. By shrinking the tubular
neighborhood if necessary, we can ensure that there is an $\n^2$-to-$1$
correspondence between long Reeb chords of $\widetilde{\Lambda}$ and
all Reeb chords of $\Lambda$: for a Reeb chord $a_i$ of $\Lambda$ and
any $j_1,j_2\in\{1,\ldots,\n\}$, there is a unique long chord
$a_i^{j_1,j_2}$ of $\widetilde{\Lambda}$ that begins on
$\Lambda_{j_1}$, ends on $\Lambda_{j_2}$, and limits to $a_i$ as the
neighborhood shrinks to zero.

In the case where $K \subset S^3$ is a knot and $\Lambda_K$ is its
conormal bundle, one can consider the conormal bundle
$\Lambda_{\widetilde{K}}$ to $\n$ parallel copies $\widetilde{K}$ of $K$. (These
parallel copies of $K$ can be given by pushoffs of $K$ with respect to
any framing, or more generally by any collection of $\n$ disjoint closed curves
in a tubular neighborhood of $K$ that are each the graph of a section
of the normal bundle to $K$, viewed as the tubular neighborhood.) It
is an exercise in local coordinates (see \cite{EENS}) to show that
$\Lambda_{\widetilde{K}}$ comprises $\n$ parallel copies of
$\Lambda_K$: write $\Lambda_{\widetilde{K}} =
\widetilde{\Lambda}_K$. Then the knot contact homology of the link
$\widetilde{K}$ is the Legendrian contact homology of
$\widetilde{\Lambda}_K$. Our main result in this subsection is now
as follows.

\begin{theorem}
Let $K$ be a knot. There is a one-to-one correspondence between
diagonal rank $\n$ representations of the knot contact homology of $K$,
and augmentations of the knot contact homology of the $\n$-component
link $\widetilde{K}$ that send all short chords to $0$.
\end{theorem}

\begin{proof}
Consider any holomorphic disk contributing to the differential in the
Legendrian contact homology of $\widetilde{\Lambda}_K$ and not
involving any short chords. Each piece of
the boundary of this disk lies on some component of
$\widetilde{\Lambda}_K$. Conversely, given a holomorphic disk
contributing to the differential of $\Lambda_K$, one can label each
piece of its boundary by any integer from $1$ to $\n$, and there is a
unique holomorphic disk with boundary on $\widetilde{\Lambda}_K$ whose
boundary pieces lie on the corresponding components of
$\widetilde{\Lambda}_K$.

Now assemble the long Reeb chords of $\widetilde{\Lambda}_K$ into
$\n\times \n$ matrices $A_i = (a_i^{j_1,j_2})$, and suppose that
in the DGA for $\Lambda_K$, $\partial(a_i) = \sum \gamma_0 a_{j_1}
\gamma_1 \cdots a_{j_\n} \gamma_\n$, where $\gamma_0,\ldots,\gamma_\n$
are words in $Q,\lambda,\mu$. Then by
the above argument, the differential in the DGA for
$\widetilde{\Lambda}_K$, omitting short chords, is given by
\[
\partial(A_i) = \sum \Gamma_0 A_{j_1} \Gamma_1 \cdots A_{j_\n} \Gamma_\n,
\]
where $\partial(A_i)$ is the matrix $(\partial(a_i^{j_1,j_2}))$,
multiplication on the right is matrix multiplication, and $\Gamma_j$
is the result of replacing $\lambda,\mu$ in $\gamma_j$ by diagonal
matrices $\operatorname{diag}(\lambda_1,\ldots,\lambda_\n)$,
$\operatorname{diag}(\mu_1,\ldots,\mu_\n)$. The result follows.
\end{proof}

Note that from a diagonal rank $\n$ representation of the knot contact
homology of a knot $K$, we can obtain a variety in $(\C^*)^{2\n+1}$
analogous to the augmentation variety: if $\rho$ is such a
representation and $\rho(\lambda) =
\operatorname{diag}(\lambda_1,\ldots,\lambda_\n)$, $\rho(\mu) =
\operatorname{diag}(\mu_1,\ldots,\mu_\n)$, then we obtain a point
$(\mu_1,\ldots,\mu_\n,\lambda_1,\ldots,\lambda_\n,\rho(Q))$ in
$(\C^*)^{2\n+1}$. From the above theorem, the resulting variety in
$(\C^*)^{2\n+1}$ where $\rho$ ranges over all diagonal rank $\n$
representations is a subset of the augmentation variety of the link
$\widetilde{K}$ consisting of $\n$ parallel copies of $K$.

We conclude with two examples. First consider the case where $K$ is
the unknot. In this case, a diagonal rank $\n$ representation of the
knot contact homology of $K$ is a map $\rho$ with values in $\n\times
\n$ matrices such that
$\rho(\lambda) = \operatorname{diag}(\lambda_1,\ldots,\lambda_\n)$, $\rho(\mu) =
\operatorname{diag}(\mu_1,\ldots,\mu_\n)$, and
\[
Q \, \text{Id} - \rho(\lambda) - \rho(\mu) + \rho(\lambda)\rho(\mu) =
0.
\]
Any such representation is reducible and a product of $\n$
$1$-dimensional representations. The resulting variety in
$(\C^*)^{2\n+1}$ is the same as the augmentation variety $V_{\widetilde{K}}(1^\n)$ for
any link $\widetilde{K}$ consisting of $\n$ parallel copies of the
unknot. Note that $\widetilde{K}$ could well have augmentations besides the ones in $V_{\widetilde{K}}(1^\n)$ (e.g. if $\widetilde{K}$ is a $(2,2m)$ torus link for $m\neq 0$, where $V_{\widetilde{K}}(2)$ is distinct from $V_{\widetilde{K}}(1^2)$), but these extra augmentations do not send short Reeb chords to $0$.

Next, suppose that $K$ is the right-handed trefoil. In this case, by a computation of knot contact homology (omitted here), a rank $\n$ representation of the knot contact homology of $K$ is a way to assign $\n\times \n$ matrices to $\lambda,\mu,a$ such that the following relations hold:
\begin{align*}
\lambda\mu &= \mu\lambda \\
\lambda\mu^6 a &= a \lambda\mu^6 \\
\lambda\mu^6-\lambda\mu^5+Q \mu^2 a \mu-Q\mu a \mu a &= 0 \\
Q-\mu-\mu a \mu^{-1}-Q \mu^2 a \lambda^{-1}\mu^{-4} a &= 0.
\end{align*}
A diagonal rank $\n$ representation further sends $\lambda,\mu$ to diagonal matrices as above.

When $\n=2$, diagonal rank $2$ representations for the trefoil fall into two classes. The reducible representations satisfy
\begin{align*}
Q^3-Q^3\lambda_1-Q^2\mu_1+Q^2\lambda_1\mu_1-2Q\lambda_1\mu_1^2+2Q^2\lambda_1\mu_1^2 \qquad & \\
+Q \lambda_1\mu_1^3-\lambda_1^2\mu_1^3-Q\lambda_1\mu_1^4+\lambda_1^2\mu_1^4 &= 0
\end{align*}
and the corresponding equation with $\lambda_1,\mu_1$ replaced by $\lambda_2,\mu_2$. (The left hand side of this equation is the augmentation polynomial for the right-handed trefoil.) This corresponds to $V_{\tilde{K}}(1^2)$ where $\tilde{K}$ is any link given by two parallel copies of $K$. The irreducible representations can be shown to satisfy the simultaneous equations
\begin{align*}
\lambda_1\mu_1^6 &= \lambda_2\mu_2^6 \\
-\lambda_1^2\mu_1^7+Q^2\lambda_1\mu_1^3\mu_2^2-Q\lambda_1\mu_1^4\mu_2^3+Q^2\mu_2^5 &= 0.
\end{align*}
This corresponds to at least a portion of $V_{\tilde{K}}(2)$, namely the subset corresponding to augmentations that send all short Reeb chords to $0$. We do not currently know if this must give all of $V_{\tilde{K}}(2)$ for any $\tilde{K}$, due to the difficulty of computing the augmentation variety for parallel copies of the trefoil.

When $\n=3$, diagonal rank $3$ representations for the trefoil are either reducible, in which case they come from a direct sum of a rank $1$ representation (i.e., an augmentation of $K$) and a rank $2$ representation, or irreducible. The latter can be shown to satisfy
\begin{align*}
\lambda_1\mu_1^6 &= \lambda_2\mu_2^6 \\
\lambda_1\mu_1^6 &= \lambda_3\mu_3^6 \\
-\lambda_1^3\mu_1^{13}+Q\lambda_1^2\mu_1^9\mu_2^2\mu_3^2+Q^3\lambda_1^4\mu_1^4\mu_2^5\mu_3^3-Q^3\mu_2^5\mu_3^5 &= 0,
\end{align*}
and correspond to a portion of $V_{\tilde{K}}(3)$ where $\tilde{K}$ consists of three parallel copies of $K$.

\section{Augmentations and Lagrangian fillings}\label{sec:augmentations}

In this section we will consider augmentations of knot contact homology, induced by exact and non-exact Lagrangian fillings in various ways. We start however by discussing the underlying general framework. Recall the setup in Section \ref{ssec:mathintro}: Legendrian contact homology associates a DGA, $\mathcal{A}(\Lambda)$, to a Legendrian submanifold $\Lambda$ in a contact manifold $V$. The algebra $\mathcal{A}(\Lambda)$ is generated by the Reeb chords of $\Lambda$ and has differential which counts holomorphic disks with one positive puncture and with Lagrangian boundary condition $\R\times\Lambda$ in the symplectization $(\R\times V, d(e^{t}\alpha))$ of $V$. In the case of knot contact homology, the ambient contact manifold is the unit conormal bundle $U^{\ast}S^{3}$, and the Legendrian submanifold is the conormal tori $\Lambda_K$ of a link $K\subset S^{3}$. Here the symplectization  $\R\times U^{\ast} S^{3}$ can be identified with the complement of the $0$-section in $T^{\ast}S^{3}$, and from the point of view of the conifold, this symplectization together with the Lagrangian $\R\times\Lambda_K$ represents the geometry at infinity of Lagrangians associated to the link $K$. The knot contact homology, or the DGA $\mathcal{A}(\Lambda_K)$, is then entirely determined by this geometry at infinity through Reeb chords of $\Lambda_K$ and holomorphic disks in $\R\times U^{\ast}S^{3}$ with boundary on $\R\times\Lambda_K$.

Returning to  the general setting of Legendrian submanifolds, our study in this section starts from properties of Legendrian contact homology that are similar to TQFT, which behaves functorially under cobordism, where the role of cobordisms is played by exact Lagrangian cobordisms $(W,L)$ defined as follows. The manifold $W$ is a non-compact symplectic manifold with exact symplectic form $\omega=d\beta$, which outside a compact subset consists of a positive end $[0,\infty)\times V_+$ where $\beta=e^{t}\alpha_+$ for a contact form $\alpha_+$ on a contact manifold $V_+$, and a negative end $(-\infty,0]\times V_-$ where $\beta=e^{t}\alpha_-$ for a contact form $\alpha_-$ on a contact manifold $V_-$. The Lagrangian submanifold $L\subset W$ is required to be exact, meaning that the closed form $\beta|_{L}$ is exact, i.e.~$\beta|_{L}=df$ for some function $f$ on $L$, and also to agree with $[0,\infty)\times\Lambda_+$ in $[0,\infty)\times V_+$ and with $(-\infty,0]\times \Lambda_-$ in $(-\infty,0]\times V_-$ for some Legendrian submanifolds $\Lambda_{\pm}\subset V_{\pm}$, see Figure \ref{fig:lagcob}.
\begin{figure}[htp]
\centering
\includegraphics[width=.5\linewidth, angle=0]{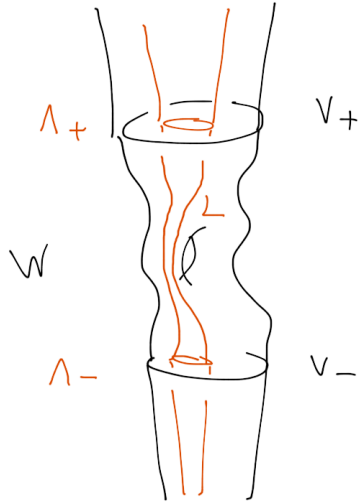}
\caption{An exact Lagrangian cobordism.}
\label{fig:lagcob}
\end{figure}
Now, if $\mathcal{A}_{\pm}$ denotes the DGA of $\Lambda_{\pm}$, then $(W,L)$ induces a DGA-map $\phi\colon \mathcal{A}_+\to\mathcal{A}_-$: that is, $\phi$ is an algebra map that commutes with the differentials, $\phi\circ\partial_+=\partial_-\circ\phi$.

The definition of $\phi$ is close to the definition of the DGA differential: $\phi$ counts rigid disks in $W$ with one positive puncture and Lagrangian boundary condition $L$. The proof of the chain map equation is analogous to the proof that the square of the differential is $0$; both derive from the identification of the boundary of 1-dimensional moduli spaces with two-level rigid curves, compare Theorem \ref{thm:exactchmap} below. We point out that the exactness assumption is crucial here: if the Lagrangian cobordism is exact then an elementary argument using Stokes' Theorem shows that no nonconstant closed holomorphic disks can form on $L$, while in the non-exact case such disks are present and typically give new boundary components of the moduli spaces that destroy the chain map property, compare Section \ref{sec:qchmaps} and Figure \ref{fig:boundarybubble} below.

Our discussion of exact Lagrangians $L$ below is a special case of the above TQFT-like functoriality in the case when the Legendrian submanifold $\Lambda_-$ at the negative end of $L$ is empty. In this case we say that $L$ is an exact Lagrangian filling of $\Lambda_+$. Since the empty Legendrian has no Reeb chords, its DGA simply equals the ground ring, which in this case is $\C[H_2(W,L)]$, with the trivial differential. Thus in this case the DGA-map $\phi\colon \mathcal{A}_{+}\to\C[H_2(W,L)]$ induces augmentations $\mathcal{A}_+\to\C$ upon choosing values in $\C^{\ast}$ of the generators of $\C[H_2(W,L)]$.

In the case of knot contact homology, we look at exact Lagrangian fillings of $\Lambda_K$ in $T^{\ast} S^{3}$, see Section \ref{ssec:exactaug}, which give augmentations of $\mathcal{A}(\Lambda_{K})$ with $Q=1$. Such exact fillings give only standard branches of the augmentation variety that contain very little information about the underlying link. (However, in Section \ref{ssec:exactcord}, we describe how a more involved count of holomorphic disks on an exact filling with the topology of the link complement, keeping track of the homotopy classes of the disk boundaries, relates to the so-called cord algebra \cite{Ngframed}, which in turn determines the augmentation variety for $Q=1$.)

In order to connect augmentations in knot contact homology with $Q\ne 1$ to geometry, it is thus not sufficient to consider exact Lagrangian fillings of $\Lambda_{K}$. This lead us to consider certain non-exact Lagrangian fillings, and our generalization of the discussion above from the exact to the non-exact case gave a new and striking relation between the DGA of $\Lambda_K$, determined by Reeb chords and holomorphic disks at infinity, and the Gromov-Witten disk potential of Lagrangian fillings, determined by counting certain configurations of disks in the compact part, see Section \ref{s:GWdef}. It is clear that similar considerations can be applied, restoring functoriality for more general non-exact Lagrangian cobordisms, but we restrict our discussion to the case most relevant to knot contact homology.

Consider a non-exact Lagrangian filling $L$ of $\Lambda_K$ lying in the resolved conifold $Y$. In Section \ref{ssec:nonexact} we specify the class of non-exact Lagrangian fillings $L$ and ambient symplectic manifolds more precisely. Here we just mention two key properties: holomorphic disks at infinity with boundary on $L$ correspond naturally to holomorphic disks in $\R\times U^{\ast} S^{3}$ with boundary on $\R\times\Lambda_K$, just like in the exact case, and closed holomorphic disks with boundary $L$ and bounded area lie in a compact subset of $Y$.

Following ideas of \cite{FO3}, we introduce obstruction chains in $L$, which are $2$-chains that connect boundaries of closed holomorphic disks on $L$ to basic $1$-cycles in $\Lambda_{K}$. Using these chains, we replace the disk count of an exact Lagrangian with a quantum corrected disk count for a non-exact Lagrangian that allows us to restore functoriality, provided the generators of the coefficient ring lie in the Lagrangian subvariety determined by the Gromov-Witten potential $W$ of the filling. For example, for the conormal filling of a knot, the constraint is $p=\frac{\partial W}{\partial x}$, and consequently $(e^{x},e^{\frac{\partial W}{\partial x}},Q)$ lies in the augmentation variety.

Analyzing this a little bit further we note that from the explicit form of the DGA $\mathcal{A}(\Lambda_K)$ one can compute the augmentation variety of a link and in particular show that it is an algebraic variety. Since any non-exact Lagrangian filling $L$ admits deformations by shifts along closed forms corresponding to $H^{1}(L)$, we find that such a filling, through the constraint given by its Gromov-Witten disk potential, parametrizes a small piece of the augmentation variety. However, since the variety is algebraic, the filling then in fact determines a whole irreducible component of the augmentation variety. For example, for the conormal filling of a knot, shifting the Lagrangian and changing the monodromy of its flat $U(1)$-connection corresponds to changing the real and imaginary parts of $x$ and thus gives a local parametrization $p=\frac{\partial W}{\partial x}$ (as $x$ varies) of the augmentation variety, which then determines a whole irreducible component of the variety (and in particular the whole variety if it is irreducible). Summarizing the discussion so far: the Gromov-Witten disk potential for non-exact fillings of $\Lambda_K$ near a given filling determines an irreducible component of the augmentation variety. Whether there exists a suitable Lagrangian filling for each irreducible component is not known at present, see Section \ref{ssec:antibrane} for partial results.

We can also view this correspondence between the augmentation variety and potentials of non-exact fillings from another perspective as follows. If $L$ is a non-exact Lagrangian filling of $\Lambda_K$, then it determines a Gromov-Witten disk potential as explained above. It is difficult to say how this potential changes under deformations (e.g.~of $L$, the almost complex structure, etc.) and, a priori, we have very little control of this change. However, the fact that the potentials of nearby fillings parametrize the augmentation variety means that possible Gromov-Witten potentials of Lagrangian fillings, which are defined by counting disks in the interior, are in fact (explicitly) restricted by the augmentation variety, which only counts disks at infinity. For example, in the case of the conormal filling of a knot, the fact that the equation $p=\frac{\partial W}{\partial x}$ parametrizes a branch of the augmentation variety determines $W(x)$ in terms of the augmentation variety (up to finite ambiguity since there are in general many branches of the augmentation variety at $x$).

Connecting to the physical perspective, this relation between knot contact homology and the Gromov-Witten potential shows that near augmentations that are geometrically induced by a Lagrangian filling, the augmentation variety of a link as defined through knot contact homology agrees with the physically defined mirror variety $V_K$; in particular, in the case of knots, the $Q$-deformed $A$-polynomial and the augmentation polynomial have zero sets that agree locally. As above, ``locally'' is in the algebro-geometric sense, so the subsets where the two agree are rather large. In fact, for many knots the augmentation variety is known to be irreducible; in these cases, since there is always the conormal filling, we know that the two polynomials must agree.

Here is an outline of the remainder of this section. In Section \ref{ssec:exactaug}, we discuss the most well-studied case of augmentations of contact homology DGAs that are induced by exact Lagrangian fillings. This gives a geometric explanation of the trivial branches  of the augmentation variety for any link. We then generalize this to certain non-exact Lagrangian fillings, which are the central objects for a geometric understanding the augmentation variety: in Section \ref{ssec:nonexact} we describe the class of non-exact filling that we use; in Sections~\ref{sec:qchmaps} and~\ref{s:GWdef}, following \cite{FO3}, we introduce obstruction chains that allow us to generalize the augmentation maps in the exact case to more involved chain maps induced by our class of non-exact Lagrangian fillings.

We then turn to geometric constructions of Lagrangian fillings. In particular, we aim at constructing  connected Lagrangian fillings of the conormal of a $\n$-component link $K$, that leads augmentations in the component $V_K(\n)$ corresponding to the trivial partition, and that would as for knots above explain why the physical and the mathematical 
varieties agree. In Sections \ref{ssec:constrlag} and~\ref{ssec:immersed}, we discuss these geometric constructions of Lagrangian fillings 
in connection with the conjecture on codimension-$1$ intersection between augmentation varieties of non-split links from Section \ref{sec:Dmirror}. In Section \ref{ssec:antibrane} we discuss a gluing operation for Lagrangian fillings that from the physical perspective corresponds to adding an anti-brane along a Lagrangian and that conjecturally leads to unobstructed connected Lagrangian fillings of the conormal of any link. Such Lagrangians would, as mentioned above, serve as actual geometric sources of the branch $V_K(\n)$ of the augmentation variety corresponding to the trivial partition of a $\n$-component link $K$. In Section \ref{ssec:alggeomex} we discuss geometric constructions and gluing explicitly in a simple example and in Section \ref{sec:legtransf} we look at Legendre duality of potentials from a geometric viewpoint.
Finally, in Sections \ref{ssec:exactcord} and \ref{ssec:nonexactcord} we study other ways of obtaining augmentations from Lagrangian fillings by coupling to flat connections on the fillings in the exact and the non-exact case, respectively.

We would like to point out that many of the arguments given here are at best outlines of possible proofs, and indeed all of them need more details. When we state a result as a theorem we believe that it can be proved with existing technology. Rather than providing proofs, the main purpose of this section is to give some idea of the geometry behind the behavior of augmentation varieties.

\subsection{Exact Lagrangian fillings and chain maps}\label{ssec:exactaug}
As always, let $K$ be a knot in $S^{3}$ and let $\Lambda_{K}$ denote its conormal torus. An \textit{exact Lagrangian filling} of $\Lambda_K$ is an exact Lagrangian submanifold $F_K\subset T^{\ast}S^{3}$ that agrees with the Lagrangian conormal $L_K$ of $K$ outside some finite-radius disk subbundle of $T^{\ast}S^{3}$. For grading reasons we also assume that the Maslov class of $F_K$ vanishes.
If $a$ is a Reeb chord of $\Lambda_{K}$, let $\mathcal{M}_{\tau}(a)$ denote the moduli space of holomorphic disks $u\colon D\to T^{\ast}S^{3}$ with boundary on $F_K$ and with one positive puncture where the disk is asymptotic to the Reeb chord $a$, and which (after capping) represents the homology class $\tau\in H_{2}(T^{\ast}S^{3},F_{K})$. Then the dimension of the moduli space is given by the grading of $a$, $\dim(\mathcal{M}_{\tau}(a))=|a|$. Define the map $\epsilon_{F_K}\colon \mathcal{A}\to\C[H_{2}(T^{\ast}S^{3},F_K)]$ by counting holomorphic disks with one positive puncture:
\begin{equation}\label{eq:exactaug}
\epsilon_{F_{K}}(a)=\sum_{\tau}|\mathcal{M}_{\tau}(a)|\,e^{\tau},
\end{equation}
where $|a|=0$ and $|\mathcal{M}_{\tau}(a)|$ denotes the algebraic number of disks in $\mathcal{M}_{\tau}(a)$.

\begin{theorem}\label{thm:exactchmap}
The map $\epsilon_{F_{K}}$ is an augmentation of $\mathcal{A}$, which we call ``the augmentation induced by $F_{K}$''.
\end{theorem}

\begin{proof}
We need to check the chain map condition $\epsilon_{F_{K}}\circ \partial=0$. By SFT compactness \cite{BEHWZ}, the two level broken curves that contributes to $\epsilon_{F_{K}}\circ \partial$ are in one to one correspondence with the boundary of the compact oriented $1$-manifold
\[
\bigcup_{|b|=1,\tau\in H_{2}(T^{\ast}S^{3},F_{K})} \mathcal{M}_{\tau}(b),
\]
see Figure \ref{fig:sftcompactness}.
\end{proof}

\begin{figure}[htp]
\centering
\includegraphics[width=.45\linewidth, angle=-90]{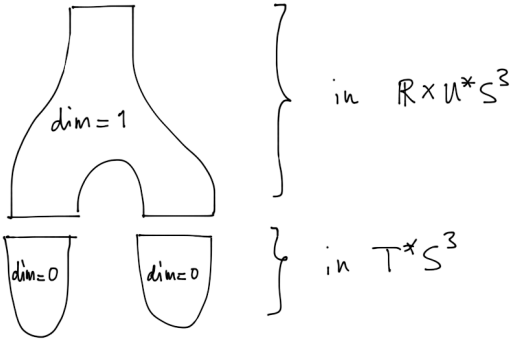}
\caption{A two level holomorphic building in the boundary of a 1-dimensional moduli space of disks with one positive puncture.}
\label{fig:sftcompactness}
\end{figure}

There are two natural exact Lagrangian fillings of $\Lambda_{K}$ in $T^{\ast} S^{3}$: its conormal $L_{K}$ and its complement $M_{K}\approx S^{3}-K$. Here $L_{K}$ is an exact Lagrangian filling of $\Lambda_{K}$ by definition. In order to see $M_{K}$ as a Lagrangian we include a brief general description of fronts.

Let $F\subset T^{\ast}M$ be an exact Lagrangian submanifold in the cotangent bundle of a smooth manifold $M$ and let $z\colon F\to \bR$ be a primitive of the pullback of the action form. Then
\[
\{(q,z(q))\in T^{\ast}M\times \R\colon q\in F\}
\]
is a Legendrian submanifold in $T^{\ast} M\times\R$ with the contact form $dz-p\,dq$, and its projection into $M\times\R$ is called the front of $F$. For generic $F$ the front is stratified. The top-dimensional smooth stratum consists of graphs of functions $g\colon U\to\R$ for $U\subset M$ and the corresponding part of $F$ is given by the graph $\Gamma_{dg}$ of the differential of $g$. Here we recover the fiber coordinates $y_j$ uniquely by
\[
y_j=\frac{\partial z}{\partial x_j},
\]
where $(x_1,\dots,x_m)$ are coordinates on $U$. The front has certain restricted singularities along lower dimensional strata that allows us to patch the solutions for $y_j$ corresponding to different sheets of the front. See Figure \ref{fig:frontex1} for an example. We will sometimes make use of a specific non-generic front that we call the Lagrangian cone, see Figure \ref{fig:lagcone}.

\begin{figure}[htp]
\centering
\includegraphics[width=.5\linewidth, angle=-90]{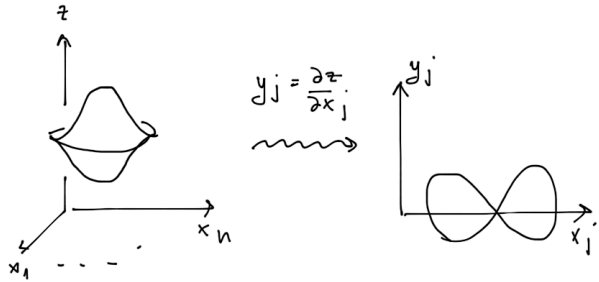}
\caption{The front of the standard immersed Lagrangian $n$-sphere: two smooth sheets come together over a cusp edge that locally is a product of a semi-cubical cusp $x^{3}=z^{2}$ and $\R^{n-1}$. The intersection of the Lagrangian immersion with a symplectic coordinate plane is shown on the right.}
\label{fig:frontex1}
\end{figure}

\begin{figure}[htp]
\centering
\includegraphics[width=.5\linewidth, angle=-90]{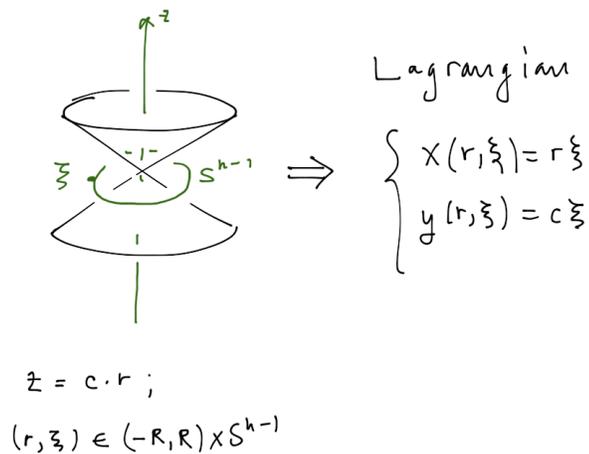}
\caption{A Lagrangian cone.}
\label{fig:lagcone}
\end{figure}

Using fronts we can represent $L_K$ and $M_K$ rather concretely. A front of the conormal $L_K$ is shown in Figure \ref{fig:conormalfront}.
To represent $M_K$, pick a Bott-Morse function on $S^{3}$ with a Bott maximum along $K$. Remove the Bott maximum to obtain an exact Lagrangian isotopic version by instead gluing in the conormal along the Bott minimum of the negative of the function. It is clear that $M_K$ agrees with $L_K$ outside a compact set, see Figure \ref{fig:complementfront}.

\begin{figure}[htp]
\centering
\includegraphics[width=.5\linewidth, angle=90]{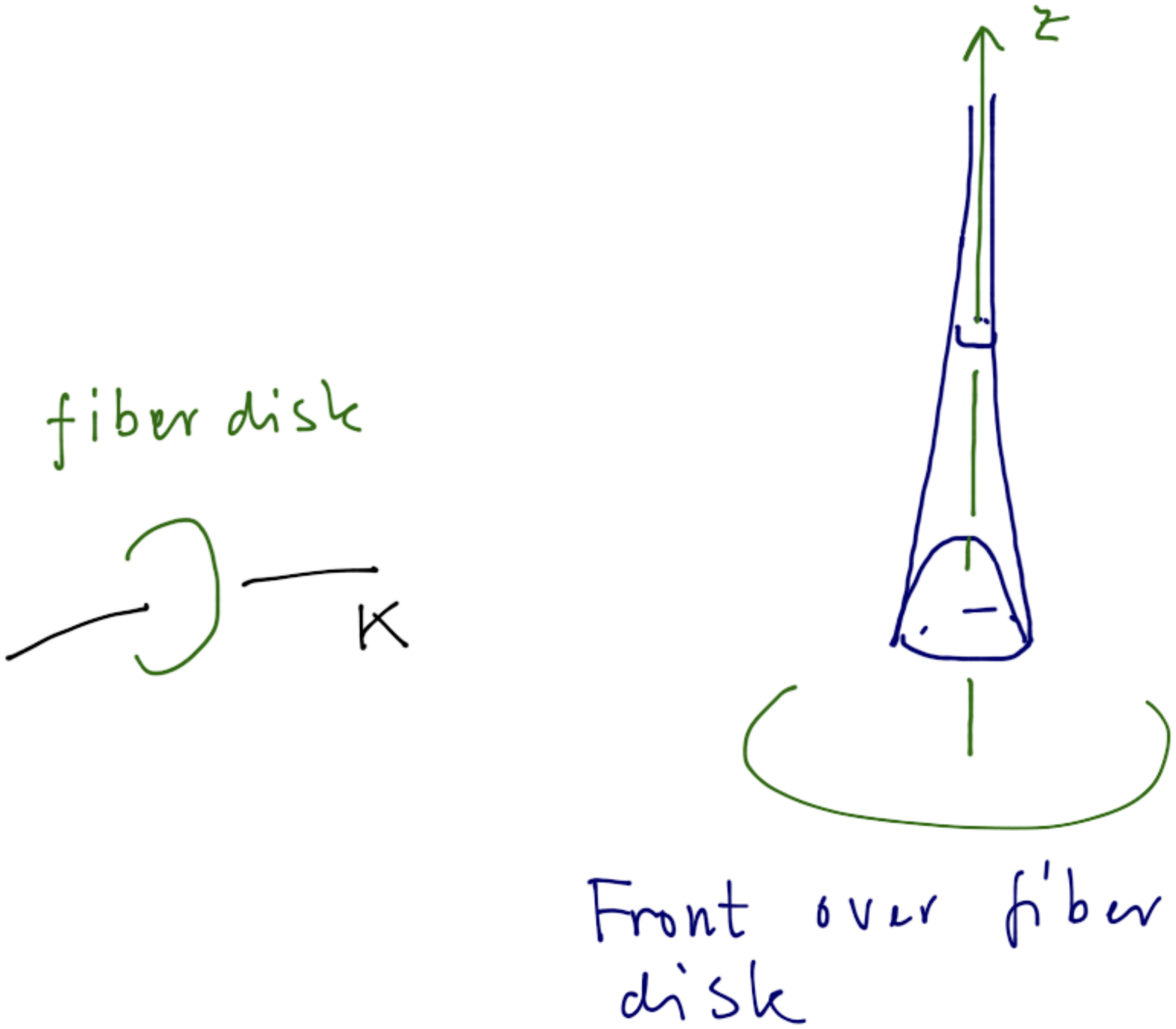}
\caption{The front of the conormal Lagrangian.}
\label{fig:conormalfront}
\end{figure}

\begin{figure}[htp]
\centering
\includegraphics[width=.5\linewidth, angle=90]{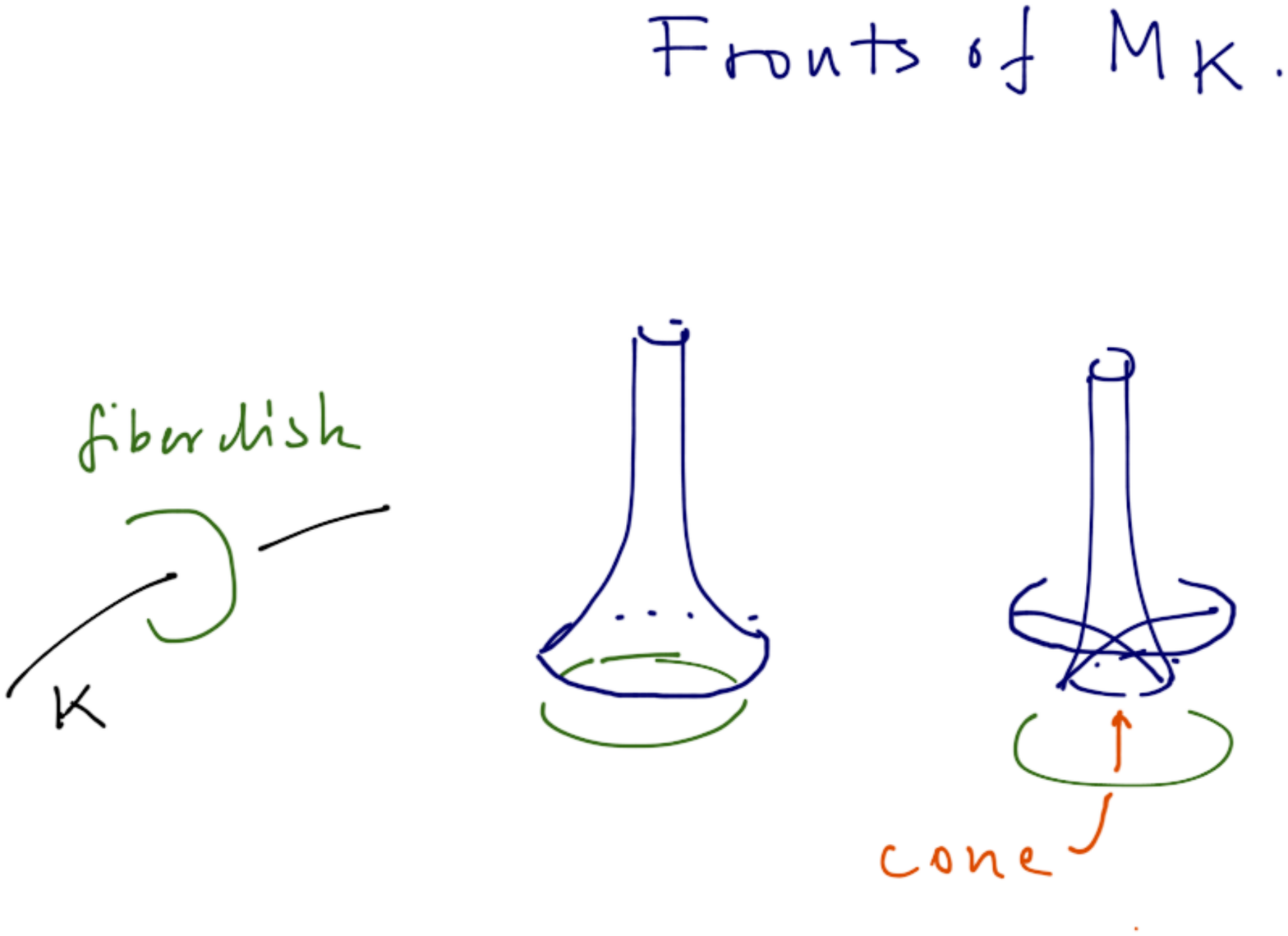}
\caption{Two fronts giving exact Lagrangian-isotopic Lagrangian versions of the link complement.}
\label{fig:complementfront}
\end{figure}

The second relative homologies $H_{2}(T^{\ast} S^{3},L_{K})$  and $H_{2}(T^{\ast} S^{3},M_{K})$ are both isomorphic to $\bZ$ and are generated by the longitude class $x$ and the meridian class $p$, respectively. On coefficients, the augmentation map in \eqref{eq:exactaug} is by definition the map on homology induced by the inclusion. Thus, the following result is immediate.
\begin{theorem}
The augmentation induced by the exact conormal filling $L_K$ of $\Lambda_K$ takes $x$ to $x$ and $p$
to $0$, and the augmentation induced by the exact knot complement filling $M_K$ takes $p$ to $p$ and $x$ to $0$
(or in multiplicative notation $e^{x}\mapsto e^{x}$ and $e^{p}\mapsto
1$ respectively $e^{x}\mapsto 1$ and $e^{p}\mapsto e^{p}$). Both
augmentations take $Q$ to $1$.
\end{theorem}

We note that this explains the geometric origin of the trivial
branches $e^{x}=1$ and $e^{p}=1$ of the augmentation variety of
$K$ at $Q=1$.

The construction of the Lagrangian complement $M_K$ discussed above generalizes immediately to the case when $K$ is a several-component link $K=K_1\cup\dots\cup K_\n$. In this case the map on homology is a bit more involved and we get the following result.
\begin{theorem}
If $K=K_1\cup\dots\cup K_\n$ is an $\n$-component link,
then the augmentation induced by the exact knot complement filling $M_K$ takes the meridian class $p_j$ of the $j^{\rm th}$ component to $p_j$ and the corresponding longitude class $x_j$ to $\sum_{k\ne j}\lk_{kj}p_k$, where $\lk_{kj}$ is the linking number of $K_k$ and $K_j$, and as above takes $Q$ to $1$.
\end{theorem}

\subsection{Non-exact Lagrangian fillings}\label{ssec:nonexact}
Our next goal is to generalize the discussion of Section \ref{ssec:exactaug} to a certain class of non-exact Lagrangian fillings. We discuss concrete constructions of
non-exact Lagrangian fillings of this type in Section \ref{ssec:constrlag}. Here we concentrate on basic properties that are crucial for the augmentations to be constructed in Section \ref{sec:qchmaps}.

We will consider Lagrangian submanifolds in $T^{\ast}S^{3}$ or $\R\times U^{\ast}S^{3}$, or in the resolved conifold $Y$. For the general discussion all these ambient spaces can be treated on equal footing and we will denote the ambient symplectic manifold $E$ and have either one of these in mind. In particular, outside a compact set, $E$ agrees with the complement of a finite-radius disk bundle of $T^{\ast}S^{3}$, which we can think of as the symplectization of $U^{\ast}S^3$, and we write $E_{R}\subset E$ for the complement of a disk bundle of radius $R$.

Consider the conormal Lagrangian $L_K$ of a knot $K$, represented as a front, see Figure \ref{fig:conormalfront}. Then $L_K\cap E_R$ projects to a punctured tubular neighborhood $U$ of $K$ in $S^{3}$. If $\eta$ is a closed $1$-form in $U$ then the shift map
\[
S_{\eta}\colon T^{\ast} U\to T^{\ast} U,\quad S_{\eta}(q,p)=(q,p+\eta(q))
\]
is a symplectomorphism.

We will consider non-exact Lagrangian submanifolds $F_{K}$ and adjusted almost complex structures $J$ on $E$ that satisfy the following three conditions:
\begin{itemize}
\item The Maslov class of $F_{K}$ vanishes.
\item For any $A>0$ there is a compact subset $C_{A}\subset E$ such that any $J$-holomorphic disk with boundary on $F_{K}$ of area at most $A$ lies in $C_A$.
\item There exists $R_0$ such that $F_{K}\cap E_{R_0}$ is the image of a part of $L_K$ under a shift map $S_\eta$ for some uniformly bounded closed $1$-form $\eta$ on $U$.
\end{itemize}
We call such Lagrangians \emph{augmentation admissible fillings of $\Lambda_K$}.

\begin{lemma}\label{lma:Gromov-Hofer}
If $F_{K}$ is an augmentation admissible filling of $\Lambda_K$, then the following version of SFT compactness holds. Any holomorphic disk of energy $\le A<\infty$ either is contained in the compact subset $C_A$ or has at least one positive puncture where it is asymptotic to a Reeb chord. Furthermore, any sequence of holomorphic disks of energy $\le A$ has a subsequence that converges to a several-level holomorphic building with a first level in $E$ and higher levels consisting of holomorphic disks in the symplectization of $U^{\ast} S^{3}$.
\end{lemma}

\begin{proof}
This is a small extension of \cite{BEHWZ}. Consider the case when a non-trivial part of the disk goes to infinity. In order to recover the limit, we use the invariance of the almost complex structure under rescaling. Take $t$ to $t-T$ for some large $T>0$, so that level $T$ maps to level $0$ in the symplectization. Under this rescaling the image of $S_\eta$ maps to the image of $S_T=S_{e^{-T}\eta}$. For such small shift we find that the adjusted almost complex structure $J$ is at distance of the order of magnitude $e^{-T}$ from an almost complex structure $J'=\tilde S_T\circ J\circ \tilde S_T^{-1}$, where $\tilde S_{T}$ is a diffeomorphism that agrees with $S_T$ and is suitably cut off in $T^{\ast} U$, and the holomorphic disks for the latter structure agree with the ordinary holomorphic disks with boundary on $L_K$. Standard bootstrap arguments show that $J$- and $J'$-holomorphic disks lie within distance of order of magnitude $e^{-T}$.
\end{proof}

The condition on vanishing Maslov class of $F_{K}$ and the fact that $E$ is Calabi-Yau shows that the formal dimension of any closed holomorphic disk with boundary on $F_{K}$ equals zero. Assume that a perturbation scheme for rigid holomorphic disks with boundary $F_{K}$ has been fixed. The moduli space $\scM(F_K)$ of rigid disks in $E$ with boundary on $F_K$ then constitutes a weighted branched $0$-manifold, see \cite{Hofer, FO3}.


\subsection{Quantum corrected chain maps}\label{sec:qchmaps}
We consider first our main geometric example: let $K$ be a knot and take $F_K$ to be the conormal filling $L_K\approx S^{1}\times\R^{2}$ shifted off of the zero section, $L_K\subset E$. In this case $H_{2}(E,L_{K})$ is generated by the class of $t$ in $H_{2}(E)$, represented by the fiber 2-sphere, and the longitude class $x\in H_{1}(\Lambda_{K})$. (In case $E=T^{\ast}S^{3}$, the fiber class is trivial and $Q=1$ below.)

In order to motivate the use of obstruction chains below let us try to carry over the construction of augmentations from the case of an exact filling to the current more general case in the most naive way. Define $\epsilon\colon\mathcal{A}\to \C[H_{2}(E,L_K)]$ as in \eqref{eq:exactaug}. As there, we try to prove the chain map equation $\epsilon\circ\partial=0$ by looking at the boundary of 1-dimensional moduli spaces $\scM(b)$. Here, however, not all boundary components contribute to $\epsilon\circ\partial$ because of boundary bubbling with nodal disks with two components, one in $\scM(b)$ and the other in $\scM(L_K)$, see Figure \ref{fig:boundarybubble}. Consequently $\epsilon$ as defined here is not a chain map. In order to remedy this we will introduce a change of variables on the coefficients  $\mathcal{A}$, guided by obstruction chains much like in \cite{FO3}.

\begin{figure}[htp]
\centering
\includegraphics[width=.5\linewidth, angle=90]{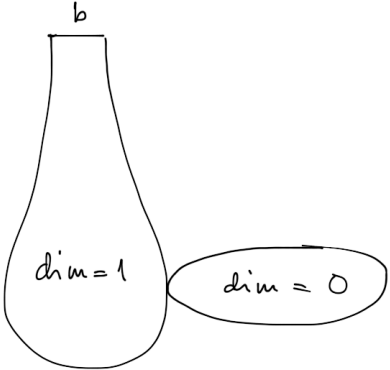}
\caption{Boundary bubbling gives new codimension one boundary strata that destroys the chain map property.}
\label{fig:boundarybubble}
\end{figure}

Consider a holomorphic disk $u\in \scM(L_K)$. Then $u\colon\partial D\to L_K$ is a cycle that represents the homology class $kx\in H_{1}(L_K)$, where $k$ is some integer and $x$ is the longitude class in $\Lambda_K$. Fix a circle $\xi$ in $\Lambda_K$ representing the class $x$. For each $u$ as above fix a smooth \emph{obstruction chain} $\sigma_u$ that connects the cycle $u(\partial D)$ to the fixed circle in $\Lambda_K$ traversed $k$ times. More precisely, recall that outside a compact subset of $E$, $L_K$ is the image of
$\R\times\Lambda_K$ 
 under a shift map $S_{\eta}$; we take the chain $\sigma_u$ to agree with 
$\R\times\xi^k$,
 where $\xi^{k}$ denotes the curve $\xi$ traversed $k$ times, with the appropriate weight, reflecting the weight of the curve $u$, in this region, see Figure \ref{fig:obstrchain}.

\begin{figure}[htp]
\centering
\includegraphics[width=.5\linewidth, angle=90]{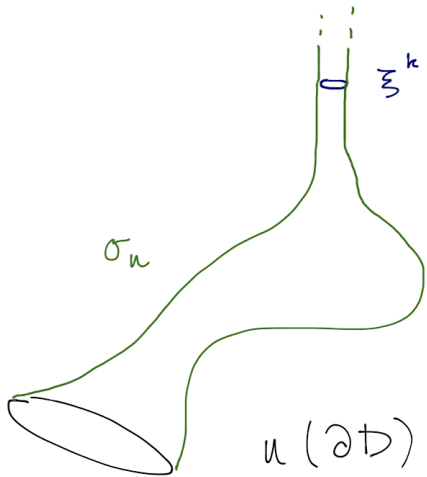}
\caption{An obstruction chain connecting the boundary of a holomorphic disk to a fixed homology generator at infinity.}
\label{fig:obstrchain}
\end{figure}

Let $\sigma=\{\sigma_{u}\colon u\in\scM(L_K)\}$ denote the collection of all the fixed obstruction chains. The collection $\sigma$ then associates to each Reeb chord $a$
a moduli space of \emph{quantum corrected disks} with positive puncture at $a$ and boundary on $L_{K}$. An element in this space is a holomorphic disk $u\in\scM(a)$, i.e.~a disk with boundary on $L_K$ and with positive puncture at $a$,
\[
u\colon (D-\{1\},\partial D-\{1\})\to (E,L_K),
\]
and with additional boundary marked points $\zeta_1,\dots,\zeta_s\in\partial D$ such that $u(\zeta_j)\in\sigma_{u_j}$ for some $u_{j}\in\scM(L_K)$. Furthermore, any disk $u_j$ attached in this first level also has additional boundary marked points which map to $\sigma_{u_{j'}}$ for disks $u_{j'}$ distinct from the ones already used, which are the disks attached at the second level. If the $r^{\rm th}$ level disks has been defined the $(r+1)^{\rm th}$ level is obtained by attaching obstruction chains at marked points in the $r^{\rm th}$ level. A quantum disk is such an object with finitely many levels, see Figure \ref{fig:qdisk}.
\begin{figure}[htp]
\centering
\includegraphics[width=.5\linewidth, angle=90]{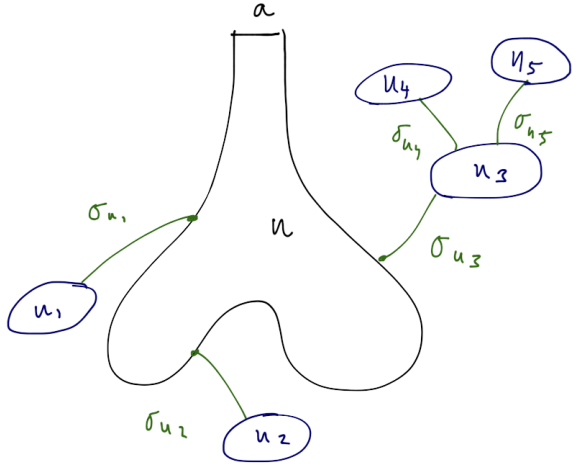}
\caption{Schematic picture of a quantum corrected disk with positive puncture. (In reality $\sigma_{u_j}$ is 2-dimensional and intersects the boundary $u(\partial D)$ or $u_{i}(\partial D)$ transversely inside the 3-dimensional $L_K$.)}
\label{fig:qdisk}
\end{figure}
Schematically, we can view a quantum corrected disk as a rooted tree, where the root is the holomorphic disk $u$ and the other vertices are the disks $u_j$. If we orient this tree so that all edges point away from the root, then an edge from $u_j$ to $u_{j'}$ corresponds to a way to map a boundary marked point on $u_j$ to a point in the cycle $\sigma_{u_{j'}}$.

We define the homology class represented by the quantum corrected disk to be the homology class in $H_{2}(E,L_K)$ given by the sum of the classes of $u$ (after capping) and all the attached disks $u_{j}$. We then write
\[
\scM_{kx+rQ}(a;\sigma)
\]
for the moduli space of quantum corrected disks with positive puncture at $a$ that represent the homology class $kx+rQ$, and $\scM(a)$ for the union of such spaces over all homology classes. In order to make sense out of counting quantum corrected disks, we note that there is a natural filtration by total area inside a fixed compact set where all closed disks lie, and that below every energy level there are only finitely many configurations that contribute. The following lemma shows how quantum corrected holomorphic disks, although their definition is rather involved, give rise to a simplification of the boundary of moduli spaces.

\begin{lemma}\label{l:qcorrboundary}
The boundary of a $1$-dimensional moduli space $\scM(b;\sigma)$ of quantum corrected disks consists of two-level quantum corrected disks, with one level being a disk in the symplectization in a $1$-dimensional moduli space, and the other level being rigid disks in $E$.
\end{lemma}

\begin{proof}
This is a standard application of obstruction chains: Figure \ref{fig:killboundary} indicates why boundary bubbling no longer is to be considered as a boundary of the moduli space, while in Figure \ref{fig:qboundary} a typical boundary configuration is depicted.
\end{proof}

\begin{figure}[htp]
\centering
\includegraphics[width=.5\linewidth, angle=90]{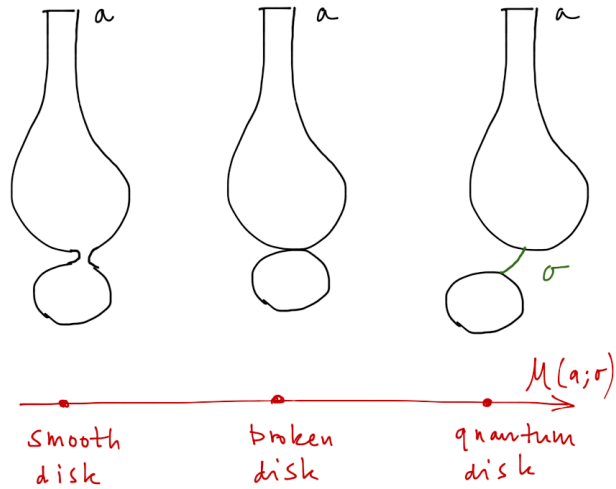}
\caption{Obstruction chains turn boundary bubbles to interior points in the moduli space.}
\label{fig:killboundary}
\end{figure}

\begin{figure}[htp]
\centering
\includegraphics[width=.5\linewidth, angle=90]{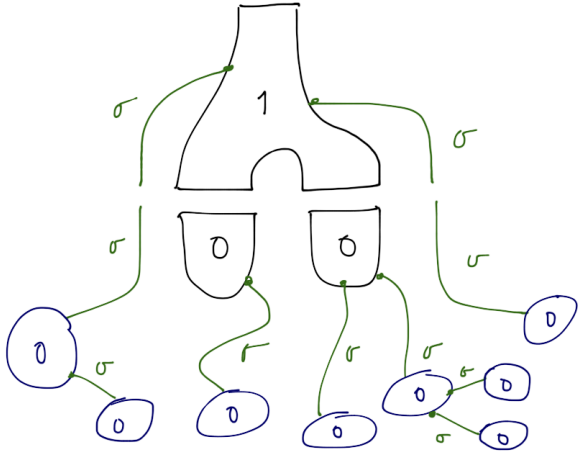}
\caption{A two-level quantum corrected disk in the boundary of the moduli space. The top level lies in $\R\times U^{\ast}S^3$, while the bottom level lies in $E$. Numbers inside the disks refer to the dimension of their moduli space.}
\label{fig:qboundary}
\end{figure}

\subsection{The GW-potential and chain maps}\label{s:GWdef}
We define the GW-potential as the generating function of the holomorphic disk configurations that are attached to a quantum corrected disk. More precisely, we count trees whose vertices are rigid holomorphic disks $u_i$ with boundary on $L_K$. To each such tree, we associate a weight as follows. First, each vertex $u_i$ in a tree has a weight given by the weight of $u_i$ viewed as a point in the moduli space of holomorphic disks. Next, let $u_{i_1}$ and $u_{i_2}$ be disks connected by an edge, and let $\sigma_{u_{i_1}}$ and $\sigma_{u_{i_2}}$ be the corresponding obstruction chains; then let the weight of the edge between $u_{i_1}$ and $u_{i_2}$ be the linking number of $u_{i_1}$ and $u_{i_2}$, which can be defined as the intersection number between $u_{i_1}(\partial D)$ and $\sigma_{u_{i_2}}$ in $L_K$, or equivalently the intersection number between $u_{i_2}(\partial D)$ and $\sigma_{u_{i_1}}$. (If necessary, we can shift the obstruction cycle off of itself before counting this intersection number.)
Finally, the weight of a tree is defined to be the product of the weights at all of its vertices and edges.

To motivate this definition of weight, suppose that we choose a vertex of the tree to be the root, and orient the tree so that all edges point away from this root. An edge from $u_{i_1}$ to $u_{i_2}$ should then be viewed as a choice of marked point on the boundary of $u_{i_1}$ that is mapped to $\sigma_{u_{i_2}}$, as in the definition of quantum corrected disk above. The weight of a tree measures the number of ways to choose these marked points. It is important for the definition of the GW-potential to note that because of the symmetry of the linking number, this weight is independent of the choice of root for the tree.


Write $\scM(L_K,\sigma)$ for the moduli space of (unrooted) trees, and define the homology class of a tree in $H_2(E,L_K)$ to be the sum of the homology classes of the disks in the tree.
We now define the GW-potential as follows:
\[
W(e^{x},Q)=\sum_{k,r\in\bZ} C_{k,r}\,e^{kx}Q^{r},
\]
where $C_{k,r}$ is the sum of the weights of the tree configurations in $\scM(L_K;\sigma)$ that represent the homology class $kx+rt$.

We claim that this GW-potential parametrizes a part of the augmentation variety of $K$.
Consider 
the specialization $(\hat{\mathcal{A}},\hat\partial)$  of the DGA $(\mathcal{A},\partial)$ to an algebra $\hat{\mathcal{A}}=\mathcal{A}|_{p=\frac{\partial W}{\partial x}}$ over (a completion of) $\C[e^{\pm x},Q^{\pm 1}]$, obtained by setting
\begin{equation}
p=\frac{\partial W}{\partial x} = \sum_{k,r\in\Z} kC_{k,r}\,e^{kx}Q^{r}.
\end{equation}
Define the graded algebra map
\[
\epsilon_{L_K}\colon \hat{\mathcal{A}}\to\C[e^{\pm x},Q^{\pm 1}]
\]
as follows on generators $a$ with $|a|=0$:
\begin{equation}\label{eq:nonexactaug1}
\epsilon_{L_K}(a)=\sum_{k,r\in\Z} |\scM_{kx+rt}(a;\sigma)| \,e^{kx}Q^{r}.
\end{equation}

\begin{theorem}\label{thm:shiftedL_K}
The map $\epsilon_{L_K}$ is a chain map, i.e.~$\epsilon_{L_K}\circ\hat\partial=0$.
\end{theorem}

\begin{proof}
This follows from the fact that contributions to $\epsilon_{L_K}\circ\hat\partial(b)$ correspond to two-level disks in the boundary of $\scM(b;\sigma)$. To see this, consider a holomorphic disk $u$ in the symplectization $\R\times U^{\ast}S^3$
with positive puncture at $b$ and negative punctures at $a_1,\ldots,a_m$ which (after capping) represents the homology class $kx+mp+rt$ in $H_1(\Lambda_K)$.
To obtain a two-level structure as in Figure~\ref{fig:qboundary}, we attach a quantum corrected disk in $E$ with positive puncture at $a_i$ (i.e., a point in $scM(a_i)$) to each negative puncture $a_1,\ldots,a_m$, and a tree of holomorphic disks in $E$ at marked points on the boundary of $u$ (note that such a tree can be rooted and oriented by choosing the root to be the disk whose obstruction cycle contains the marked point). Now in the boundary of the one-dimensional moduli space $\scM(b;\sigma)$, this precisely gives a contribution of $e^{kx}e^{m\frac{\partial W}{\partial x}}Q^{r}$: first note that each rooted tree of disks in $E$ is counted with multiplicity given by the weight of the tree, and second recall that at infinity the obstruction cycles are products of the form $\R\times\xi^{l}$ that intersect any curve in the class $kx+mp$, $m\cdot l$ times. (In the formula $m$ is explicitly visible and $l$ comes from the differentiation.) By definition of the specialized differential, this contribution (summed over $k$ and $r$) is $\epsilon_{L_K}(\hat\partial(b))$, which must then be equal to $0$.
\end{proof}

\begin{remark}\label{rmk:proofA=A}
As a consequence of Theorem \ref{thm:shiftedL_K}, the equation $p=\frac{\partial W}{\partial x}$ is a parametrization of a branch of the augmentation variety of $K$, which is in agreement with the parametrization \eqref{eq:p=dW/dx} derived from physical arguments, and which in fact proves the mathematical and physical definitions of the variety agree if the augmentation variety is irreducible.
\end{remark}

In general it is difficult to find augmentation admissible Lagrangian fillings with the topology of a link complement that do not intersect the zero section, but e.g.~for fibered links it is possible; see Section \ref{ssec:antibrane}, where we will discuss alternative constructions. Here we will discuss such Lagrangian fillings in order to see the main properties of the more general fillings in the simplest possible context.    Thus, let $K=K_1\cup\dots\cup K_\n$ be an $\n$-component link. Let $M_{K}\subset E$ be an augmentation admissible Lagrangian filling with the topology of the  link complement $M_K\approx S^{3}-K$. Then $H_{2}(E,M_{K})$ is generated by the meridian classes $p_j\in H_1(\Lambda_{K_j})$ ($j=1,\dots,\n$) of the components and the fiber class $t$. In analogy with the case of the conormal discussed above, we fix circles $\eta_j$ in $\Lambda_{K_j}$ that represent the class $p_j$, $j=1,\dots, n$, and for each $u\in\scM(M_K)$ we fix an obstruction chain $\sigma_u$ connecting $u(\partial D)$ to a linear combination of multiples of the fixed curves $\eta_j$ in $\Lambda_{K_{j}}$, where as above we take the obstruction chains to be (weighted) cylinders on the fixed curves outside a compact subset.  Let again $\sigma$ denote the collection of all the fixed obstruction chains $\sigma_{u}$, $u\in\scM(M_K)$.

Let $p=(p_1,\dots,p_n)$ and let $\tilde U(p,Q)$ denote the GW-potential of $M_K$, defined as a sum over tree configurations exactly as above:
\[
\tilde U(p,Q)=\sum_{k\in\bZ^{n},\;r\in\bZ} B_{k,r}\;e^{k\cdot p}Q^{r},
\]
where $B_{k,r}$ is the sum of the weights of the configurations in $\scM_{k\cdot p+rt}(M_K;\sigma)$.

 Furthermore, exactly as above, if $a$ is a Reeb chord of $\Lambda_{K}$, then let the quantum corrected moduli space of holomorphic disks with positive puncture at $a$ and boundary on $M_{K}$ that represent the homology class $k\cdot p+rt$ be denoted $\scM_{k\cdot p+rt}(M_{K};\sigma)(a)$. Then Lemma \ref{l:qcorrboundary} holds for such moduli spaces of dimension 1.

Unlike in the case of knots, for links there is also a classical part of the potential which we discuss next. If $\lk_{ij}=\lk(K_i,K_j)$ denotes the linking number between $K_i$ and $K_j$, then the homology class in $H_{1}(M_{K})$ represented by the longitude $x_j$ of $K_j$ is $\sum_{i\ne j}\lk_{ij}p_i$. Define the superpotential $U(p,Q)$ of $M_K$ as
\begin{equation}
U(p,Q)=\sum_{1\le i<j\le n}\lk_{ij}p_ip_j + \tilde U(p,Q).
\end{equation}

After adding this quadratic form (related to classical rather than quantum homology) the construction is parallel to that for the conormal.
Consider the specialization $(\hat{\mathcal{A}},\hat\partial)$  of the DGA $(\mathcal{A},\partial)$ to $\hat{\mathcal{A}}=\mathcal{A}|_{x_j=\frac{\partial U}{\partial p_j}}$ over (a completion of) $\C[e^{\pm p_1},\dots,e^{\pm p_n},Q^{\pm 1}]$, obtained by setting
\begin{equation}
x_j=\frac{\partial U}{\partial p_j},\quad j=1,\dots,n.
\end{equation}
Define the graded algebra map
\[
\epsilon_{M_K}\colon \hat{\mathcal{A}}\to\C[e^{\pm p_1},\dots,e^{\pm p_n},Q^{\pm 1}]
\]
as follows on generators $a$ with $|a|=0$:
\begin{equation}\label{eq:nonexactaug2}
\epsilon_{M_K}(a)=\sum_{k\in\Z^{\n},r\in\Z} |\scM_{k\cdot p+rt}(a;\sigma)| \,e^{k\cdot p}Q^{r}.
\end{equation}

\begin{theorem}\label{thm:shiftedM_K}
The map $\epsilon_{L_K}$ is a chain map, i.e.~$\epsilon_{L_K}\circ\hat\partial=0$.
\end{theorem}

\begin{proof}
The proof is analogous to the proof of Theorem \ref{thm:shiftedL_K}: a disk in the symplectization with positive puncture at $b$ and negative punctures at $a_1,\dots,a_m$ which represents the homology class $m\cdot x +k\cdot p+rt$ should be counted with coefficient $e^{k\cdot p}e^{m\cdot\left(\frac{\partial U}{\partial x_1},\dots,\frac{\partial U}{\partial p_\n}\right)}Q^{r}$ in the boundary of $\scM(b)$, where the classical part
\[
\sum_{i=1}^{n} m_i\left(\sum_{j\ne i}\lk_{ij} p_j\right)
\]
corresponds to the homology class represented by the disk itself and
\[
m\cdot\left(\frac{\partial \tilde U}{\partial x_1},\dots,\frac{\partial \tilde U}{\partial p_\n}\right)
\]
corresponds to intersections with the obstruction chains at infinity. By definition, this is also the contribution to $\hat\partial$.
\end{proof}

\begin{remark}
It follows from Theorem \ref{thm:shiftedM_K} that the equations
\[
x_j=\frac{\partial U}{\partial p_j},\quad
j=1,\dots,\n,
\]
give a parametrization of an $\n$-dimensional branch of the augmentation
variety of the link $K$, in agreement with \eqref{eq:x_i=dU/dp_i}. In general this branch
involves mixed augmentations (where Reeb chords between different
components are augmented) and hence 
appears to correlate with $V_K(\n)$.
\end{remark}

\begin{remark}\label{rmk:Hopf}
As an example of the above, consider the Hopf link $K_1\cup K_2$. It follows from \eqref{eq:Hopfpotential} that the part of the augmentation variety corresponding to mixed augmentations is given by the equations
\[
x_1=p_2,\quad x_2=p_1.
\]
In terms of potentials this corresponds to $U(p_1,p_2)=p_1p_2$, which means that $\tilde U(p_1,p_2)=0$. This is related to the fact that the Hopf link admits a exact Lagrangian filling $M_\Hopf$ in $T^{\ast} S^{3}-S^3$, on which no nonconstant holomorphic disks can form. To see this, note that there exists a Bott function on $S^{3}$ with maxima along $K_1$ and minima along $K_2$ and no other critical points. Such a Bott function can be obtained by pulling back a Morse function on $S^{2}$ with exactly two critical points by the Hopf map $S^{3}\to S^{2}$. Then adding cones along these Bott manifolds, we get a front that gives an exact Lagrangian $M_\Hopf$ that does not intersect the $0$-section.
\end{remark}

\subsection{Embedded augmentation admissible Lagrangians}\label{ssec:constrlag}
In this section we consider basic constructions of embedded Lagrangian fillings in $T^{\ast}S^{3}-S^3$. Together with the conjecture on intersections of branches of the augmentation varieties corresponding to partitions of links, we get a (correspondingly conjectural) geometric source of the augmentations corresponding to the trivial partition of any link.

Consider first the case of the conormal filling.   First we move $L_K$ off of the $0$-section using shifts by the angular form $d\theta$ supported in a tubular neighborhood of the knot.
\begin{lemma}\label{lma:L_Kadmiss}
For non-planar links, the shifted version of $L_K$ is augmentation admissible.
\end{lemma}

\begin{proof}
We must check the second property for augmentation admissibility. Let $R>0$ be large and write $E_R$ for the complement in $T^{\ast}S^{3}$ of a radius $R$ disk bundle.
We argue using monotonicity: if a disk does not lie entirely outside $E_{\frac{R}{2}}$ but still leaves $E_R$ then its area is bounded below by $C R^{2}$ for some constant $C>0$, which is a contradiction for $R$ sufficiently large. Thus the entire disk must lie outside $E_{\frac{R}{2}}$. However, in this region the symplectic form is exact and the area is given by some constant times the number of times its boundary goes around $S^{1}$. Assuming that the link does not lie in a plane, the length of such a curve is bounded below by $\delta R$ for some $\delta>0$, and monotonicity for holomorphic disks again gives a lower area bound of the form $C\delta^{2}R^{2}$, which gives a contradiction.
\end{proof}

We next consider moving link complements off of the $0$-section.

\begin{lemma}\label{lma:fiberedadmiss}
If a link $K$ is fibered then $\Lambda_{K}$ admits an augmentation admissible Lagrangian filling $M_K\approx S^{3}-K$ in $T^{\ast}S^{3}-S^3$.
\end{lemma}

\begin{proof}
Since $K$ is fibered, there exists a non-vanishing $1$-form on $S^{3}-K$ that can be taken to agree with the standard circular form $d\theta$ of the meridian circles  in a small neighborhood of a knot. Consider a small Bott function with maximum along the link, and add in the conormal along the Bott maximum as in Figure \ref{fig:complementfront}. Then the exact Lagrangian is parametrized as follows over the fiber disks:
\[
x(r,\xi) = r\xi, \qquad y(r,\xi) = \alpha(r)\xi,
\]
where $(r,\xi)\in (-\delta,\delta)\times S^{1}$ and $\alpha(r)>0$.
Consequently, the Lagrangian shifted by $c\cdot d\theta$ is given by
\[
x(r,\xi) = r\xi, \qquad y(r,\xi) = \alpha(r)\xi+cR_{\frac{\pi}{2}}\xi,
\]
where $R_{\frac{\pi}{2}}$ denotes rotation by $\frac{\pi}{2}$. Thus adding
the conormal along the Bott maximum does not introduce any intersections with the $0$-section. Finally, the argument confining closed disks is a repetition of the proof of Lemma \ref{lma:L_Kadmiss}.
\end{proof}

If a link $K$ is not fibered then we can move off the complement $M_{K'}$ of a link $K'=K\cup K_0$ where $K_0$ is a braid axis for $K$. More precisely, fix a braid presentation of $K$ and let $K_0$ be a braid axis.
\begin{lemma}\label{lma:shiftbraidaxis}
With $K'$ as above, $\Lambda_{K'}$ admits an augmentation admissible Lagrangian filling $M_{K'}\approx S^{3}-K'$ in $T^{\ast}S^{3}-S^3$.
\end{lemma}

\begin{proof}
There is a non-vanishing $1$-form $d\theta$ along a tubular neighborhood of an unknot dual to the braid axis, i.e.~along the tubular neighborhood $N$ of an unknot $U$ such that the link lies in $N$ and projects with everywhere nonzero derivative to $U$. Add a small Bott maximum along the link; a similar analysis as in the proof of Lemma \ref{lma:fiberedadmiss}, where this time the shift is in the direction of the knot which is perpendicular to the $\xi$-plane, shows that there are no intersections introduced by adjoining the conormal.
\end{proof}

We next combine Lemma \ref{lma:shiftbraidaxis} with Conjecture \ref{cnj:codim1} on codimension-one intersections of augmentation varieties (see also the discussion in Section \ref{sec:Dmirror}) to get a description of the part of the augmentation variety of the original link $K$ that corresponds to the trivial partition in terms of that of $K'=K\cup K_0$. For convenience we express our result in terms of potentials. Let $W_{\bigcirc}(x)$ and $U_\bigcirc(p)$ denote the potentials of the unknot. Setting the augmentation polynomial of the unknot (see \eqref{eq:V_O}) to zero determines $e^{p}$ uniquely in terms of $e^{x}$, or symmetrically $e^{x}$ in terms of $e^{p}$. We thus find that $p=\frac{\partial W_\bigcirc}{\partial x}(x)$ implies $x=\frac{\partial W_\bigcirc}{\partial x}(p)=\frac{\partial U_\bigcirc}{\partial p}(p)$ and that the potentials are equivalent.

Let $K$ be a non-split link with $\n$ components $K=K_1\cup\dots\cup K_\n$. As in Lemma \ref{lma:shiftbraidaxis}, let $K'=K\cup K_0$, where $K_0$ is a braid axis for $K$. Then also $K'$ is non-split and for an appropriate choice of orientation of $K_0$, $\lk_{0j}=\ell_j>0$ for all $j>0$. Write $x=(x_1,\dots,x_n)$, $p=(p_1,\dots,p_n)$ and $\ell=(\ell_1,\dots,\ell_n)$. Let $U_{K'}(p_0,p)$ denote the superpotential determined by $M_{K'}$,
\[
U_{K'}(p_0,p)=p_0(\ell\cdot p)+\sum_{i<j}\lk_{ij}p_ip_j+\tilde U_{K'}(p_0,p).
\]
As $M_{K'}$ is connected, this gives rise to the following local parametrization of $V_{K'}(n+1)$:
\[
x_0=\frac{\partial U_{K'}}{\partial p_0},\qquad
x = \frac{\partial U_{K'}}{\partial p}.
\]
Consider now the intersection $V(n+1)\cap (V(1)\times V(n))$. Assuming 
Conjecture~\ref{cnj:codim1}, this intersection has dimension $n$. On the other hand $K_0$ is an unknot, we know the first factor $V(1)$ explicitly, and equating the $x_0$-coordinates of a point in the intersection gives
\[
\frac{\partial U_{K'}}{\partial p_0}(p_0,p)=\frac{\partial U_\bigcirc}{\partial p}(p_0)
\]
or equivalently
\[
\frac{\partial W_\bigcirc}{\partial x}\left(\frac{\partial U_{K'}}{\partial p_0}(p_0,p)\right)=p_0.
\]
Solving for $p_0$ we get solutions $p_0=g(p)$, say. Letting $U_K(p)$ be a potential for $V_{K}(n)$, we find, by equating the $x$-coordinates of a point in the intersection, that
\[
\frac{\partial U_{K'}}{\partial p}(g(p),p) = \frac{\partial U_{K}}{\partial p}(p).
\]
Now, our assumption on the dimension of the intersection implies that this equation must be trivially satisfied for all $p$ since otherwise the intersection would have dimension $<n$. This thus gives the following formula for the potential of $K$ corresponding to $V_{K}(n)$:
\begin{equation}\label{eq:U_K}
U_{K}(p) = U_{K'}(g(p),p) -U_\bigcirc(g(p)),
\end{equation}
which is an example of ``Lagrangian reduction'' (see Section \ref{ssec:augvarcomps}), rephrased in the language of potentials.

\subsection{Immersed Lagrangian fillings}\label{ssec:immersed}
In this section we first discuss immersed Lagrangians, and then 
indicate a geometric counterpart of intersecting augmentation varieties that physically corresponds to adding an anti-brane along a Lagrangian. This conjecturally yields a connected geometric Lagrangian which is a source for the augmentations in the branch of the augmentation variety corresponding to the trivial partition of any link, and indicates a geometric origin for Conjecture \ref{cnj:codim1}.

We use the notion $E$ for the ambient symplectic manifold as in Section \ref{ssec:nonexact}. Consider a generically immersed Lagrangian submanifold $F_K$ which satisfies the Maslov class condition and the condition at infinity of an augmentation admissible Lagrangian. Then $F_K$ has a finite number of transverse double points in a compact part of $E$. Let $q$ be a double point of $E$; then $q$ is the transverse intersection of two local branches of $F_K$. We write $q^{\pm}$ for $q$ considered as an ordered intersection of the two branches. One can associate a Maslov type index $|q^{\pm}|$ to $q^{\pm}$, see \cite{EES}, such that the dimension of the moduli space $\scM(q^{\pm})$ of holomorphic disks with boundary on $F_K$ and one puncture where the disk is asymptotic to $q^{\pm}$, with the boundary orientation inducing the order of the sheets corresponding to the decoration, equals
\[
\dim(\scM(q^{\pm}))=|q^{\pm}|-1,
\]
and where $|q^{-}|=(\dim(E)-2)-|q^{+}|=1-|q^{+}|$, see Figure \ref{fig:lagpuncture}.

\begin{figure}[htp]
\centering
\includegraphics[width=.5\linewidth, angle=90]{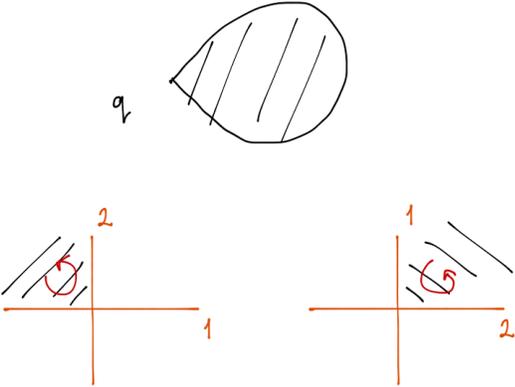}
\caption{Holomorphic disks with one puncture at an intersection point and the induced ordering of the local sheets.}
\label{fig:lagpuncture}
\end{figure}

Under additional conditions, the curve counting formula corresponding to \eqref{eq:nonexactaug1} and \eqref{eq:nonexactaug2} defines an augmentation also for immersed $F_K$. In order to understand the additional condition, we consider the boundary of a $1$-dimensional moduli space $\scM(b;\sigma)$ of quantum corrected holomorphic disks with boundary on $F_K$ and a positive puncture at a Reeb chord $b$. Here a new boundary phenomenon appears: there are two-component broken disks where one component is a disk in $\scM(b,q^{\mp};\sigma)$ with two punctures, a positive puncture at $b$ and another puncture at $q^{\mp}$ for some double point $q$, and the other is a disk in $\scM(q^{\pm};\sigma)$, see Figure \ref{fig:lagbreak}.
\begin{figure}
\centering
\includegraphics[width=.5\linewidth, angle=90]{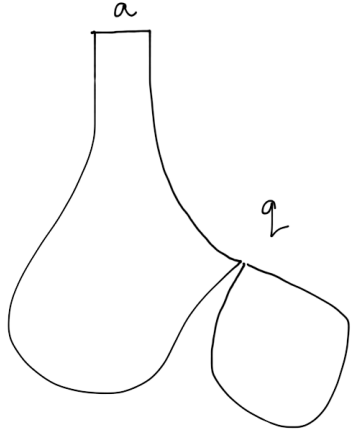}
\caption{In the presence of double points, disks with one puncture at the double point can break off, leading to new codimension-one boundary phenomena that can destroy the chain map property.}
\label{fig:lagbreak}
\end{figure}
We have
\[
\dim(\scM(b,\sigma))=\dim(\scM(b,q^{\mp}))+\dim(\scM(q^{\pm}))+1,
\]
which by genericity equals $1$ only if $\dim(\scM(b,q^{\mp}))=\dim(\scM(q^{\pm}))=0$. Conversely, any such configuration can be glued uniquely to a $1$-dimensional quantum corrected disk. In particular we see that if for every double point $q$ with $|q^{+}|=1$, the count of elements in the moduli space of quantum corrected disks $\scM(q^{+};\sigma)$ equals zero, then there is no contribution to the boundary from this type of breaking and consequently the chain map equation holds.

We call an immersed Lagrangian $F_K$ as above, which can be equipped with an obstruction chain $\sigma$ so that it has this property with respect to rigid once-punctured disks, \emph{unobstructed}. In these terms the above arguments can be summarized in the following way.

\begin{lemma}\label{lma:unobstr}
Unobstructed Lagrangian fillings define augmentations.
\end{lemma}

\subsection{Geometric Lagrangian fillings and anti-branes}\label{ssec:antibrane}
We next construct a connected Lagrangian filling of any link (or more precisely, of the split union of the link with an unknot). Let $K$ be a link and consider the auxiliary link $K'=K_0\cup K$, where $K_0$ is a braid axis for $K$ as in Lemma \ref{lma:shiftbraidaxis}. We use notation as there. Let $M_{K'}\subset E$ denote the Lagrangian complement of $K'$ shifted off the $0$-section. Recall that the braid axis is an unknot which links all components of $K$ positively. Consider a Hopf link $K''=K_{0}\cup K_{-1}$, where $K_{-1}$ is a meridian circle of the unknot $K_{0}$ and hence in particular $K_{-1}$ is itself an unknot. Let $M_{K''}$ denote the complement of the Hopf link $K''$, see Remark \ref{rmk:Hopf}. Noting that $p_0$ and $\ell\cdot p$ determine the shift of the conormal torus of $K_{0}$ in the direction of the meridian and the longitude, respectively, in $M_{K'}$, we find that for $p_{-1}=\ell\cdot  p$, the shifts of the conormal tori in $M_{K'}$ and $M_{K''}$ agree.

The next step of our construction is to glue $M_{K'}$ to $M_{K''}$ along $\Lambda_{K_0}$. From a physical perspective this corresponds to adding an anti-brane along $M_{K''}$. In Figure \ref{fig:antibrane}, the gluing operation is described in terms of fronts.
\begin{figure}[tp]
\centering
\includegraphics[width=.8\linewidth]{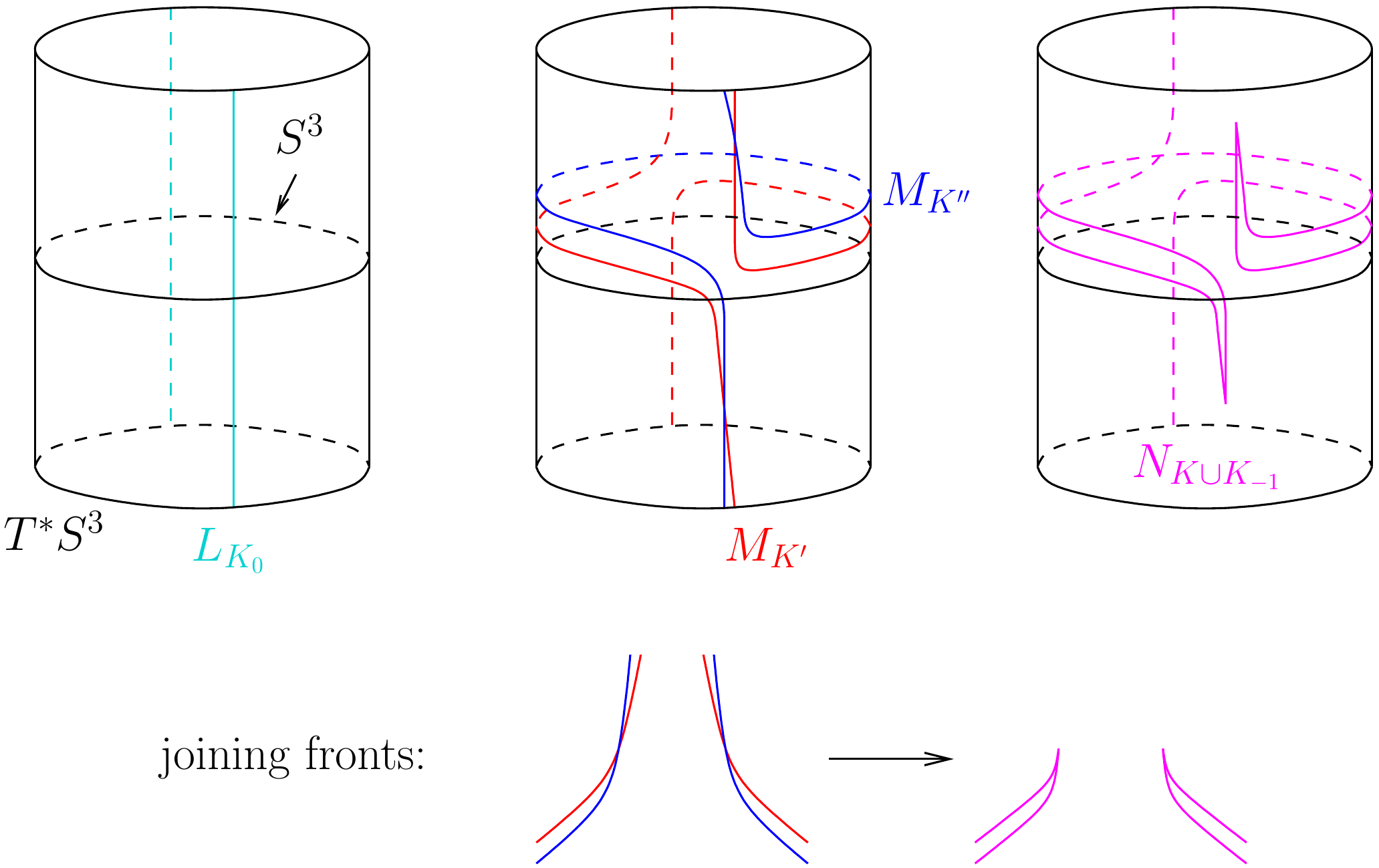}
\caption{Gluing the Lagrangian complements $M_{K'}$ and $M_{K''}$ along $\Lambda_{K_0}$, where $K'=K_0 \cup K$ and $K''=K_0\cup K_{-1}$, to obtain a new Lagrangian $N_{K\cup K_{-1}}$.}
\label{fig:antibrane}
\end{figure}
The result of this gluing is an immersed Lagrangian filling  $N_{K\cup K_{-1}}$. Since the  complement of the unknot is diffeomorphic to $S^{1}\times\R^{2}$, we find  that $N_{K\cup K_{-1}}$ has the topology of the complement of a link $K\cup K_{-1}$ in $S^{1}\times S^{2}$. Here $S^{2}$ is the union of two hemispheres $S^{2}=H_{-}\cup H_{+}$, $K$ is a braid in $S^{1}\times H_{-}$, and $K_{-1}$ is the trivial one-strand braid $S^{1}\times h_{+}\subset S^{1}\times H_{+}$ for $h_{+}$ a point in the interior of $H_{+}$. Viewed as a link in $S^{3}$, $K\cup K_{-1}$ is a split link formed by adding a distant unknot component $K_{-1}=\bigcirc$ to the link $K$.

Consider the GW-potential $\tilde U_{K\cup K_{-1}}$ of $N_{K\cup K_{-1}}$. Holomorphic disks with boundary on $N_{K\cup K_{-1}}$ (i.e.~closed disks without punctures) lie on either $M_{K'}$ or $M_{K''}$. However, $M_{K''}$ is the complement of a Hopf link which supports algebraically zero disks, see Remark \ref{rmk:Hopf}. Noting that $x_{-1},p_1,\dots,p_n$ generate $H_{1}(N_{K\cup K_{-1}})$ and that $p_0$ is homologous to $x_{-1}$, we then find that
\[
\tilde U_{K\cup K_{-1}}(x_{-1},p) = \tilde U_{K'}(x_{-1},p).
\]
Furthermore, since $p_{-1}$ is homologous to $\ell\cdot p$ in $N_{K\cup K_{-1}}$, we get the corresponding superpotential
\[
U_{K\cup K_{-1}}(x_{-1},p)=U_{K'}(x_{-1},p).
\]

We consider finding conditions that guarantee that the immersed Lagrangian filling $N_{K\cup K_{-1}}$ of $\Lambda_{K\cup K_{-1}}$ is unobstructed. Assuming that is the case, we find that
\[
p_{-1}=\frac{\partial U_{K'}}{\partial x_{-1}},\quad x=\frac{\partial U_{K'}}{\partial p}
\]
is a parametrization of $V_{K\cup K_{-1}}$. On the other hand, since $K\cup K_{-1}$ is split, we know that $V_{K\cup K_{-1}}=V_{K}\times V_{K_{-1}}$, and since $K_{-1}$ is an unknot, we conclude that
\begin{equation}\label{eq:intersO}
\frac{\partial U_{K'}}{\partial x_{-1}}=\frac{\partial W_{\bigcirc}}{\partial x}(x_{-1})
\end{equation}
or equivalently
\[
\frac{\partial W_\bigcirc}{\partial x}\left(\frac{\partial U_{K'}}{\partial x_{-1}}(x_{-1},p)\right)=x_{-1}.
\]
Solving this equation, we get $x_{-1}=g(p)$, and in complete agreement with \eqref{eq:U_K}, we find that
\[
U_{K}(p)=U_{K'}(g(p),p)-U_{\bigcirc}(g(p))=U_{K\cup K_{-1}}(g(p),p)-U_{\bigcirc}(g(p)),
\]
is a potential for $V_{K}(\n)$.

This leads us to conjecture that the immersed Lagrangian filling $N_{K\cup K_{-1}}$ is unobstructed provided that
\begin{equation}\label{eq:conjunobstr}
\frac{\partial W_{\bigcirc}}{\partial x}\left(\frac{\partial U_{K\cup \bigcirc}}{\partial x_{-1}}(x_{-1},p)\right)=
x_{-1},
\end{equation}
in the case that $K$ is an non-split link. This conjecture gives a connected unobstructed Lagrangian filling $N_{K\cup \bigcirc}$ of $K\cup \bigcirc$ for any link $K$, where $\bigcirc$ is an unknot split from $K$. The corresponding potential $U_K$ for $K$ is then obtained by subtracting the unknot potential from the geometrically defined potential $U_{K\cup \bigcirc}$ and gives a geometric source for augmentations in $V_{K}(\n)$.

We next consider the classical limit of $U_K$. In this limit, $\frac{\partial U_{K'}}{\partial p_0}=\ell\cdot p\ne 0$ and \eqref{eq:intersO} becomes
\[
\ell\cdot p=p_{-1},\
\]
where $(x_{-1},p_{-1})$ is a point on the augmentation variety of the unknot. But in the classical limit  the augmentation variety of the unknot is $\{x_{-1}=0\}\cup\{p_{-1}=0\}$, so $x_{-1}=0$ and the potential $U_K$ limits to $\sum_{i<j}\lk_{ij}p_ip_j$, which from the viewpoint of augmentations agrees with the Lagrangian link complement in $T^{\ast}S^{3}$.

We will often denote $N_{K\cup K_{-1}}$ equipped with the potential
$$
U_K(p)=U_{K\cup K_{-1}}(g(p),p)-U_{\bigcirc}(g(p))
$$
by $M_K$, see Section \ref{sec:linkcompl}.

We next consider a possible geometric source of Conjecture \ref{cnj:codim1} in a similar spirit. Consider a  link $K=K_1\cup K_2$ where $K_1$ is a knot and $K_2$ a link such that both $K$ and $K_2$ are non-split. Let $M_K$ and $M_{K_1}$ be Lagrangian fillings of $K\cup K_{-1}$ and $K_1\cup K_{-1}$, respectively. Here we use the same braid axis $K_0$ for $K$ and $K_{1}$ and for suitably values of $p_{1}$, the shifts of $\Lambda_{K_1}$ in $M_K$ and $M_{K_1}$ agree. We next glue $M_{K_1}$ to $M_K$ along $\Lambda_K$ (add an anti-brane along $M_{K_1}$). Denote the resulting Lagrangian filling $M_{K/K_1}$. Then $M_{K/K_1}$ is a filling of $K_2\cup \bigcirc\cup\bigcirc$, where $\bigcirc\cup\bigcirc$ is a two-component unlink split from $K_2$.

Assuming that $M_{K/K_1}$ is unobstructed, we consider conditions under which it would define an augmentation. As above, the holomorphic disks with boundary on $M_{K/K_1}$ are exactly the holomorphic disks in $M_K$ and those in $M_{K_1}$ (the latter counted with opposite signs), and we have obstruction chains from the two pieces. In $M_{K/K_1}$, however, the obstruction chains going to $p_1$ no longer go off to infinity, but end at a finite distance. In order to have the chain map property (see the proof of Theorem \ref{thm:shiftedL_K}), we must deal with the discontinuity that may appear when a disk crosses the reference cycle $p_1$. This is however straightforward: by equating the one-point functions of the potential of the pieces,
\beq\label{eq:reductioncond}
\frac{\partial U_{K\cup K_{-1}}}{\partial p_1}=
\frac{\partial U_{K_1\cup K_{-1}}}{\partial p_1},
\eeq
the total homology class of one parameter families of quantum corrected disks varies continuously as the  punctured disk crosses $p_1$. As in the proof of Theorem \ref{thm:shiftedM_K}, the classical part of the superpotentials in \eqref{eq:reductioncond} is related to the homology class of the boundary of the punctured disk, and the Gromov-Witten part is related to insertions. Imposing this condition, $M_{K/K_1}$ thus gives an augmentation of $K_2$.

Note that in \eqref{eq:reductioncond} we can replace $U_{K\cup K_{-1}}$ with $U_K$ since
\begin{align*}
\frac{\partial U_{K}}{\partial p_1}(g(p),p)
&=
\frac{\partial U_{K\cup K_{-1}}}{\partial p_1}(g(p),p)+
\left(\frac{\partial U_{K\cup K_{-1}}}{\partial p_{-1}}(g(p),p)
-\frac{\partial U_{\bigcirc}}{\partial p_{-1}}(g(p),p)
\right)\frac{\partial g}{\partial p_1}\\
&=\frac{\partial U_{K\cup K_{-1}}}{\partial p_1}(g(p),p)
\end{align*}
and similarly for $U_{K_1}$. We thus find that the augmentation variety $V_{K_2}$ contains a branch given by the potential
\[
U_{K_2}(p_2)=U_K(p_1(p_2),p_2)-U_{K_1}(p_1(p_2)),
\]
where $p_1(p_2)$ is a solution of the equation
\[
\frac{\partial U_K}{\partial p_1}(p_1,p_1)-\frac{\partial U_{K_1}}{\partial p_1}(p_1)=0.
\]
This is exactly ``Lagrangian reduction'' rephrased in terms of potentials, see Conjecture \ref{conj:reduction}, and thus gives a possible explanation for the codimension-one intersection property of Conjecture \ref{cnj:codim1}. In order to turn this into an actual proof, the main point would be to verify the assumptions we made on the immersed Lagrangians being unobstructed under the conditions stated.

\subsection{Geometric constructions in a simple example}\label{ssec:alggeomex}
In this section, we study a basic example that illustrates Conjecture~\ref{cnj:codim1} as well as the geometric constructions in Section \ref{ssec:antibrane}.

Consider two knots $K_1$ and $K_2$ and their connected sum $K_1\# K_2$. Assume that we have Lagrangian fillings $M_1=M_{K_1}$ and $M_2=M_{K_2}$ with GW-potentials $U(p_1,Q)$ and $U_2(p_2,Q)$ respectively. Then we construct a Lagrangian filling $M_{12}=M_{K_1\# K_2}$ with potential $U_{12}(p,Q)$. Studying the augmentations of $\Lambda_{K_1\# K_2}$ presented with a short Reeb chord near the connecting point, one can show that the GW-potentials are related by
\[
U_{12}(p,Q)=U_1(p,Q)+U_2(p,Q)-U_\bigcirc(p,Q),
\]
where $U_\bigcirc(p,Q)$ is the potential of the Lagrangian complement of the unknot (which is isomorphic to the conormal). To see this, recall from Section \ref{Sec:knotch} how the augmentation polynomial behaves under connected sum, and note that is a restatement of that in terms of GW-potentials.

We now consider the example $K$ of an $(n+1)$-component link consisting of an unlink $K'$ on $n$ components with a braid axis $K_0$, see Figure \ref{fig:unlinkaxis}.
\begin{figure}[htp]
\centering
\includegraphics[width=.5\linewidth, angle=90]{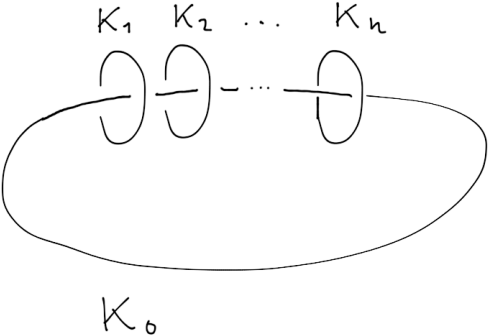}
\caption{The $n+1$ component link $K$}
\label{fig:unlinkaxis}
\end{figure}
Consider adding an anti-brane along $M_{K_0}$ to $M_K$. This results in a Lagrangian that lies over a tubular neighborhood of an unknot $K_{-1}$ dual to (Hopf linked with) $K_0$ and with front as depicted in Figure \ref{fig:frontunlinkaxis}.
\begin{figure}[htp]
\centering
\includegraphics[width=0.8\linewidth]{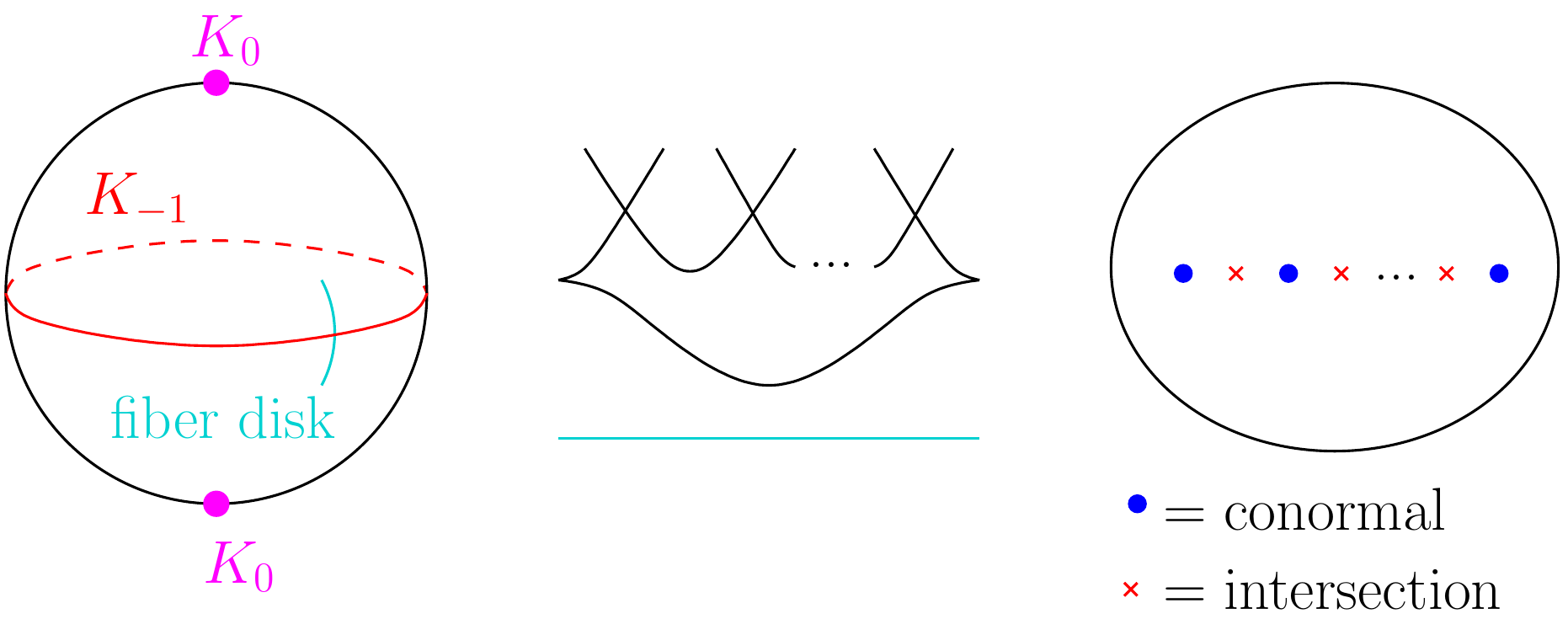}
\caption{Left, the configuration of $K_0$ and $K_{-1}$; middle, the front of $N_{K'\cup K_{-1}}$ over a fiber; right, the front of $N_{K'\cup K_{-1}}$ seen from above.}
\label{fig:frontunlinkaxis}
\end{figure}
By the above connected sum result and our knowledge of the potential of the Hopf link (see Section~\ref{ssec:Hopf}), we can compute the superpotential $U_{K}$ for $K$ as
\[
U_{K}(p_0,p_1,\dots,p_n)= p_0(p_1+\dots+p_n)-(n-1)U_\bigcirc(p_0).
\]
Adding an anti-brane along $M_{K_0}$ introduces the homology relation $p_1+\dots+p_n=0$ and gives the superpotential
\[
U_{K/K_0}(x_1,p_2,\dots,p_n)=x_1(p_1+\dots+p_n)-nU_\bigcirc(x_1),
\]
where we keep additional cycles in order to keep track of obstruction chains. The continuity argument from Section \ref{ssec:antibrane} requires that we remove $x_1$, where different obstruction cycles meet by passing to the critical point. Using the symmetry of the unknot potential under exchange of $x$ and $p$, we find the critical point condition
\[
x_1=\frac{\partial U_\bigcirc}{\partial p}\left(\frac{1}{n}(p_1+\dots+p_n)\right).
\]
Re-substituting this, we find that a ``potential'' for $K'$ is given by
\[
U_{K'}=n\,U_\bigcirc\left(\frac{1}{n}(p_1+\dots+p_n)\right),
\]
which, in the light of our knowledge of augmentations for split links, gives correct augmentations only if $p_1=\dots=p_n$.

Looking at the Lagrangian filling, we can see these phenomena from a geometric perspective. Applying \cite{E}, we get a description of holomorphic disks in terms of flow trees and in this case this is straightforward. The Lagrangian $M_{K/K_0}$ obtained by anti-brane addition has Bott circles of self-intersection points. Perturbing to a Morse situation leaves two intersection points on each circle, and the rigid disks have a puncture at the double point with the smallest action. By \cite{E}, these rigid disks correspond to flow lines starting at saddle points and hitting the cusp edge as shown in Figure \ref{fig:flowtrees}.
\begin{figure}[htp]
\centering
\includegraphics[width=.6\linewidth, angle=90]{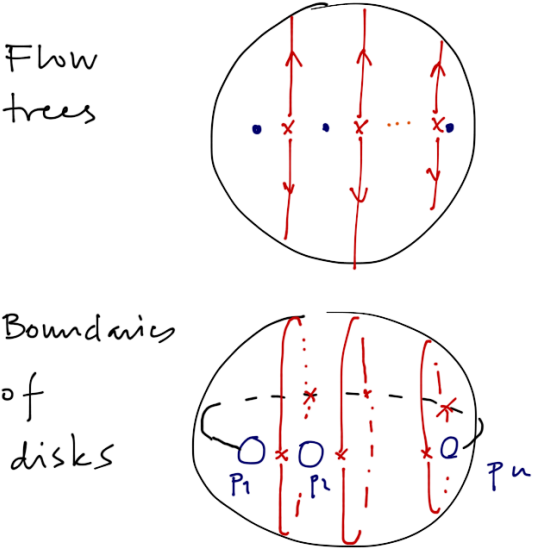}
\caption{The flow trees that correspond to rigid disks, and the boundaries of the corresponding rigid disks lying in a fiber punctured sphere of $N_{K/K_1}$.}
\label{fig:flowtrees}
\end{figure}
These rigid disks are indeed obstructions, with difference classes $p_1$, $p_1+p_2$, \dots, $p_1+\dots+p_{n-1}$ as indicated there. In the exact case, we find that the obstructions then vanish only if $p_1=\dots=p_n=0$. In the general case, one would expect these equations to be quantum corrected to
\[
p_j=\frac{\partial U_\bigcirc(p_0)}{\partial p_0},\quad j=1,\dots,n,
\]
which is indeed the case if we use a model where the rigid disks lie close to each conormal piece and look like the disks on the conormal of the corresponding unknot, i.e.~the disks lie near the cones in Figure \ref{fig:frontunlinkaxis}. This is in agreement with the above algebraic analysis.

\subsection{Legendre transform of superpotentials from a geometric perspective}\label{sec:legtransf}
In this section we consider the duality between the potentials of different Lagrangian fillings mentioned in Section \ref{sec:framingoflag} from a geometric perspective. Let $K$ be a knot, let $L_K$ be its conormal filling, and let $M_K$ be a filling that fills in the longitude, e.g.~the knot complement if the knot is fibered or the filling $N_{K\cup \bigcirc}$ constrained by \eqref{eq:conjunobstr}.  Consider a branch of the augmentation variety defined by $L_{K}$ and parametrized by $x$:
\[
x\mapsto \left(x, \frac{\partial W}{\partial x}\right).
\]
Parametrizing the curve by $p$ instead, we get a function $p\mapsto x(p)$, and the corresponding potential is given by
\begin{align*}
U(p)&=\int x(p)\,dp = x(p)p-\int p\,dx= x(p)p - W(x(p))\\
&=\frac{\partial U}{\partial p}(p)\cdot \frac{\partial W}{\partial x}(x(p)) - W(x(p)).
\end{align*}
This indicates that the GW-potential of $M_K$ should be Legendre dual to the GW-potential of $L_K$. We next discuss this duality heuristically from a geometric viewpoint.

First consider holomorphic disks with boundary on $L_{K}$. We fix a perturbation scheme so that no holomorphic disk intersects the central $S^{1}$ in the conormal. Note that the construction of $M_{K}$ involves gluing $L_{K}$ to the 0-section. We assume that the disks on $M_{K}$ are exactly those that come from $L_{K}$. In order to shift $M_{K}$ off the 0-section, we glue with a twist, taking a curve in the class $x+p$ in $\Lambda_K\subset M_K$ to the curve in class $x$ in $\Lambda_K\subset L_K$.

Looking at the holomorphic disks of $M_K$ in a neighborhood of the knot, we see that they look like the disks in $L_K$ except they all link with linking number $1$. Consider now a $1$-parameter deformation which changes this, energy level by energy level. First we let the disk on the second energy level cross the disk on the first energy level to change the linking number, and inductively we let the disk on energy level $k$ cross the disks of all lower energy levels. Below a fixed energy level we thus find three boundary components of the parametrized moduli spaces: disks of $L_K$, disks of $M_K$, and broken disks of two components that intersect at a boundary point. In the latter configuration, one of the disks is a disk of higher energy and can be viewed as a disk on $L_K$, and the other has already been moved in place and can be viewed as a disk on $M_K$. Summing over the ends of the weighted 1-dimensional moduli space, we find
\[
U(p)=\frac{\partial U}{\partial p}\cdot\frac{\partial W}{\partial x}-W(x(p)),
\]
where $x(p)$ is the function discussed above. See Figure~\ref{fig:legtransf}.
\begin{figure}
\centering
\includegraphics[width=.5\linewidth, angle=90]{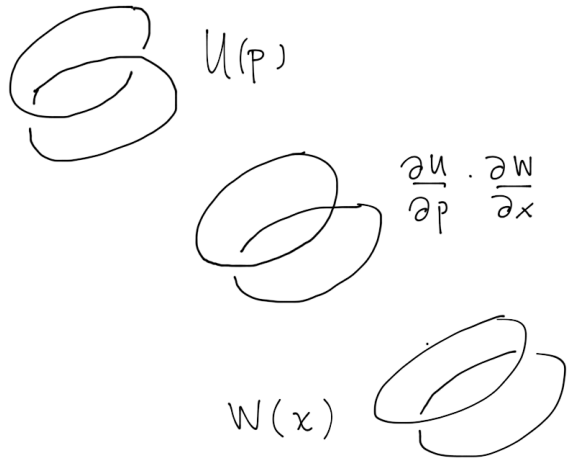}
\caption{The three boundary configurations.}
\label{fig:legtransf}
\end{figure}

Differentiating with respect to $p$ gives
\[
p\left(\frac{\partial x}{\partial p}-\frac{\partial^{2} U}{\partial p^{2}}\right)=0.
\]
Hence $x(p)=\frac{\partial U}{\partial p}$,
\[
U(p)=\left.\bigl(xp - W(x)\right)\bigr|_{x\text{ critical}},
\]
and $U(p)$ is the Legendre transform of $W(x)$.

\subsection{Augmentations, the cord algebra, and flat connections on exact fillings}\label{ssec:exactcord}
In order to find other geometric sources for augmentations, we consider the exact Lagrangian filling $M_K$ with topology of the link complement $S^{3}-K$ from a different perspective, in line with the discussion at the end of Section \ref{sec:linkcompl}.

Let $\bZ\pi$ denote the group ring of the fundamental
group $\pi:=\pi_1(S^{3}-K)=\pi_1(M_K)$. We define a chain map
$\Phi\colon \scA(\Lambda_{K})\to \bZ\pi$, where $\bZ\pi$ is endowed
with the trivial differential mimicking the usual augmentation
map. More precisely, let $\gamma\in\pi$ and let $c$ be a Reeb chord of
$\Lambda_{K}$. Write $\scM_{\gamma}(c)$ for the moduli space of
holomorphic disks $u\colon D\to T^{\ast}S^{3}$ with $u(\partial
D)\subset M_{K}$, with one positive puncture where the disk is
asymptotic to the Reeb chord strip of $c$, and with $u(\partial D)$
representing the homotopy class $\gamma\in\pi$ (after being closed up
by a capping path for $c$). Define $\Phi$ as follows: $\Phi$ takes $Q$ to $1$ and $\lambda = e^x$ and $\mu = e^p$ to the corresponding elements in $\bZ\pi$; and for a Reeb chord $c$,
\[
\Phi(c)=\sum_{\gamma\colon\dim\scM_{\gamma}(c)=0}|\scM_{\gamma}(c)|[\gamma],
\]
where $|\scM_{\gamma}(c)|$ denotes the algebraic number of disks in $\scM_{\gamma}(c)$ and where $[\gamma]$ is the element in $\bZ\pi$ corresponding to $\gamma\in\pi$. The usual argument using the description of the boundary of the 1-dimensional moduli space as two-level disks gives the chain map property $\Phi\circ\partial=0$.

Recall the concrete construction of $M_{K}$ (see Figure \ref{fig:complementfront}), where the conormal of the link was added to the $0$-section via Lagrange surgery.  Compactness properties of the space of holomorphic disks with boundary on $L_{K}\cup S^{3}\subset T^{\ast}S^{3}$ were studied in \cite{CEL}, where it was shown that for generic data there is a constant $C>0$ such that the number of corners (where the disk boundary switches between $L_{K}$ and $S^{3}$) for any disk is at most $C$. The corners of disks with boundary on $L_{K}\cup S^{3}$ can be smoothed to give a disk on $M_{K}$, and conversely, disks with boundary on $M_{K}$ limit to disks with corners on $L_{K}\cup S^{3}$ as the surgery degenerates. For isolated double points this relation was studied in \cite{FO3}.

In order to describe the relation between holomorphic disks with boundary on $L_{K}\cup S^{3}$ and disks with boundary on $M_{K}$, we first characterize the corners of disks. At any corner the boundary orientation of the disk gives one incoming and one outgoing arc. We say that a corner with incoming arc along $S^{3}$ and outgoing along $L_{K}$ is positive and that a corner with the opposite behavior is negative. Generically, rigid disks have only non-degenerate corners. In $1$-parameter families of disks with boundary on $L_K\cup S^{3}$, there are also degenerate corners where the disk has a vanishing derivative in the fiber direction and in a neighborhood of which two nondegenerate corners collide.

The relation between disks on $L_{K}\cup S^{3}$ and disks on the ``surgered'' version $M_{K}$ were studied in \cite{FO3}. Using this result, or alternatively the correspondence between flow trees and holomorphic disks established in \cite{E} together with standard gluing theorems, one establishes the following results, beginning with the case of rigid disks.

\begin{theorem}\label{thm:rounding0}
For all sufficiently small surgery parameters,
there is the following many-to-one correspondence between rigid holomorphic disks with boundary on $M_K$ and disks with boundary on $L_{K}\cup S^{3}$:

To each disk with $2k$ corners corresponds $2^{k}$ disks obtained by smoothing the corners of the disk. At a positive corner the smoothing is unique, see Figure \ref{fig:pospunct}, and at each negative corner there are two smoothings and the homotopy class of the corresponding boundaries differ by the meridian class in $\Lambda_K$, see Figure \ref{fig:negpunct}.  
\end{theorem}

\begin{figure}[htp]
\centering
\includegraphics[width=.5\linewidth, angle=0]{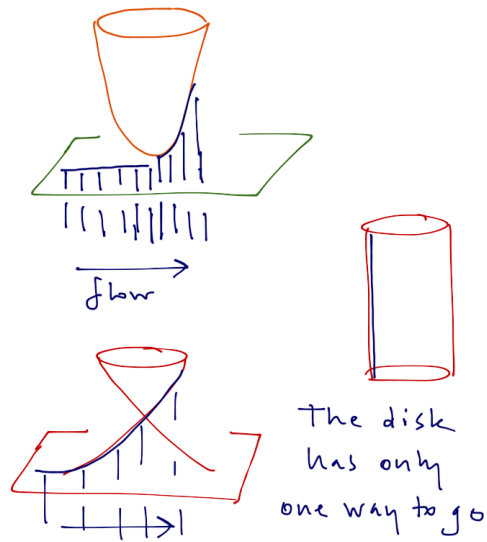}
\caption{Smothing at a positive corner.}
\label{fig:pospunct}
\end{figure}

\begin{figure}[htp]
\centering
\includegraphics[width=.5\linewidth, angle=0]{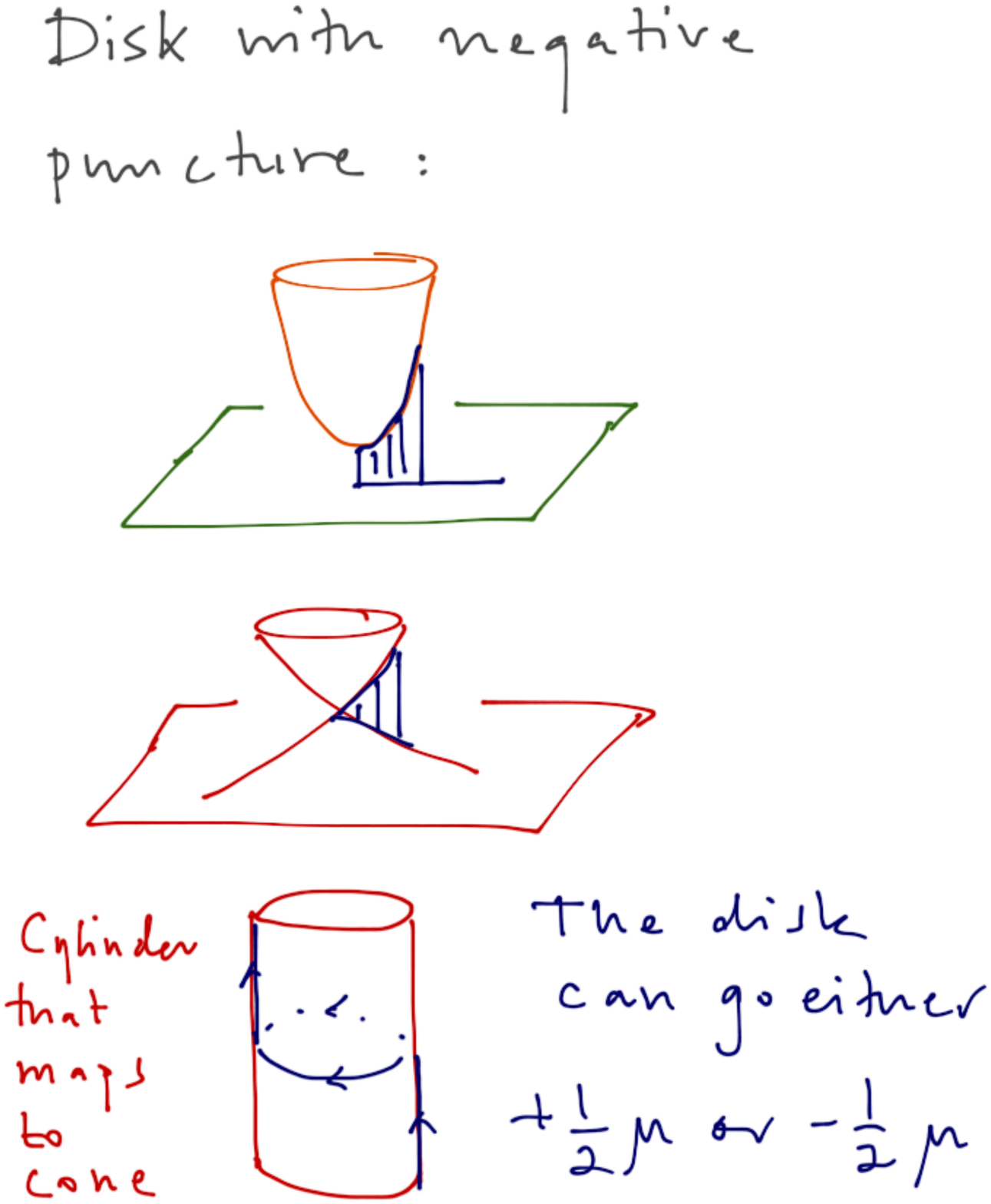}
\caption{Smothing at a negative corner.}
\label{fig:negpunct}
\end{figure}

We next consider $1$-parameter families of disks. Consider such a $1$-parameter family $\scM$. Outside of a finite set of points in $\scM$, the disks have nondegenerate corners that can be smoothed exactly as in Theorem \ref{thm:rounding0}. It thus remains to describe the smoothed family for parameters in a small neighborhood of a disk with a degenerate corner.
Consider such a family parametrized by $t\in[-\delta,\delta]$, with two nondegenerate corners for $t\in[-\delta,0)$, a degenerate corner at $t=0$, and no corner at $t\in(0,\delta]$. Note that the disk with degenerate corner also lies in a moduli space of disks without corners that simply cross through $K$, see Figure \ref{fig:degenerations}.
\begin{figure}[htp]
\centering
\includegraphics[width=.5\linewidth, angle=90]{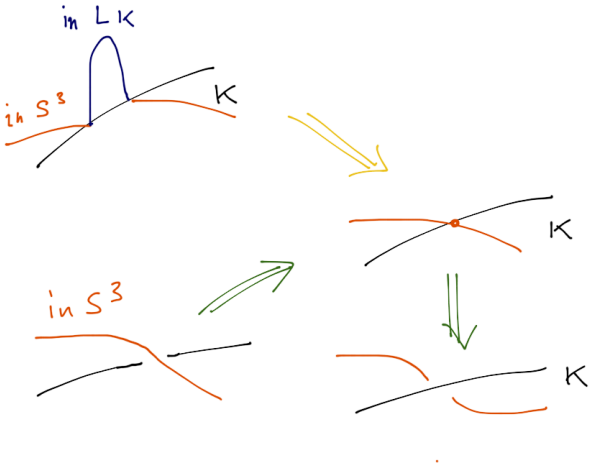}
\caption{Holomorphic disks with boundary on $L_K\cup S^{3}$ near a disk intersecting the knot.}
\label{fig:degenerations}
\end{figure}
Choose smoothings at all corners except for the corners near the degenerate corner. Write $\scM_{\pm}$ for the two branches of the moduli space corresponding to the two smoothings of corners of disks $(-\delta,0)$, see Figure \ref{fig:collision}. Also write $\scN_{\pm}$ for the two components of the moduli space without corners in the complement of the disk that intersects $K$, see Figure \ref{fig:cross}.

\begin{figure}[htp]
\centering
\includegraphics[width=.5\linewidth, angle=0]{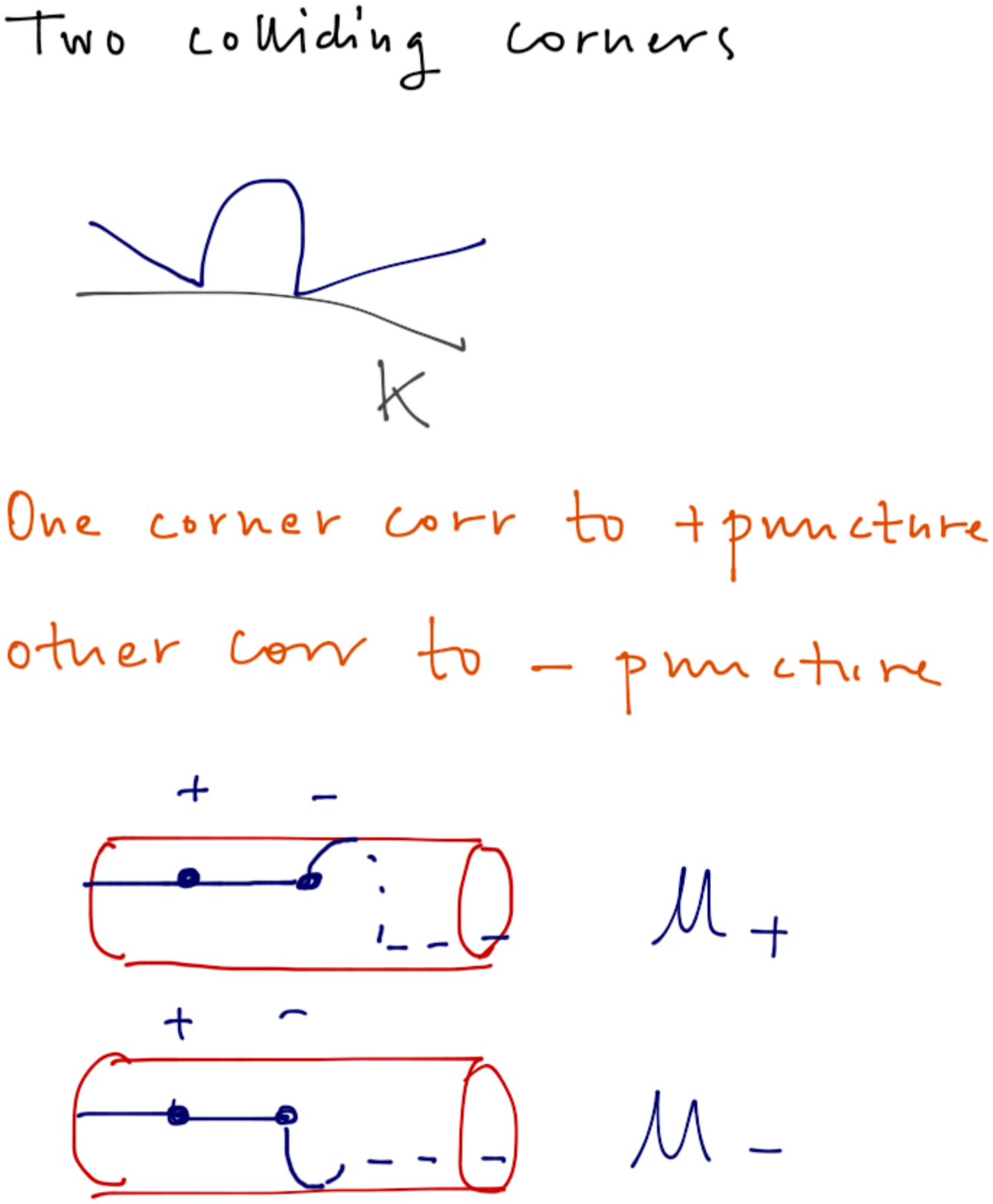}
\caption{Smoothings of the family of disks with two colliding corners. The lower picture shows the smoothed boundary in $M_K$.}
\label{fig:collision}
\end{figure}

\begin{figure}[htp]
\centering
\includegraphics[width=.5\linewidth, angle=90]{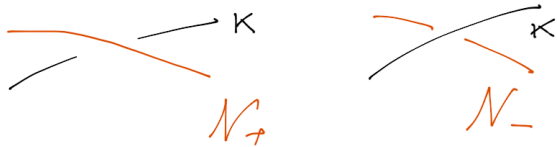}
\caption{Subdivision of the 1-parameter family $\mathcal{N}$.}
\label{fig:cross}
\end{figure}

\begin{theorem}\label{thm:rounding1}
For all sufficiently small surgery parameters, the moduli space of disks with boundary on $M_{K}$ corresponding to the intersecting family and the crossing family is obtained by joining $\scM_{+}$ to $\scN_+$ and $\scM_{-}$ to $\scN_{-}$, see Figures \ref{fig:collision} and \ref{fig:cross}. 
\end{theorem}

There is a topological interpretation for knot contact homology in homological degree $0$ (and $Q=1$) known as the cord algebra \cite{Ngframed}.
Theorems \ref{thm:rounding0} and \ref{thm:rounding1} precisely match with the generators and relations of the cord algebra.
Indeed, work in progress \cite{CELN} aims at establishing that the map
counting holomorphic disks with boundary on $L_K\cup S^{3}$ gives an isomorphism
between the degree zero contact homology of $\Lambda_{K}$ and the
cord algebra using a version of string topology. Together with Theorems \ref{thm:rounding0} and \ref{thm:rounding1}, this would lead to a proof of the following result that we state as a conjecture:

\begin{conjecture}\label{thm:cordiso}
The chain map $\Phi$ induces an isomorphism between the degree $0$ knot contact homology (with $Q=1$) and the cord algebra.
\end{conjecture}

\noindent The fact that these two algebras are isomorphic is already known \cite{Ngframed}, but this conjecture would give a geometric reason for the isomorphism.

As an aside, it is straightforward to show that the chain map $\Phi$ is surjective: over each geodesic binormal chord $\overline{\gamma}$ of the knot sits a trivial half strip with boundary on $L_K\cup S^{3}$ and which ends at the corresponding Reeb chord $\gamma$. After smoothing the corners, this gives the element $(1-\mu)[\gamma]$ in the cord algebra, proving surjectivity.

As a consequence of Conjecture \ref{thm:cordiso}, we would have that
augmentations of $\scA(\Lambda_{K})$ to $\bC$ (at $Q=1$) coincide with
algebra morphisms of the cord algebra of $K$ into $\bC$. In
particular, as shown in \cite{Ngframed}, any flat
$\mathrm{SL}(2,\bC)$-connection on $M_{K}$ induces such a map from
the cord algebra, which explains why the $A$-polynomial divides the
(classical) augmentation polynomial. Further, in examples some of the
factors of the augmentation polynomial which do not appear in the
$A$-polynomial arise from flat $\mathrm{GL}(n,\bC)$ connections for
$n>2$ in a similar way \cite{Ngsurvey,Cornwell}.

\subsection{The cord algebra and small non-exact deformations}\label{ssec:nonexactcord}
In this section we discuss a possible generalization of the material in Section \ref{ssec:exactcord} and
consider maps from the DGA of $\Lambda_{K}$ into the group ring of the fundamental group of a knot complement in presence of quantum corrections. The picture studied here applies for small shifts (i.e.~for Lagrangian fillings that are sufficiently close to being exact).

Let $K=K_{1}\cup\dots\cup K_{\n}$ be an $\n$-component link with complement $M_K$ that can be shifted off of the $0$-section; for example, $K$ could be a link with an adjoined braid axis, as discussed in Section~\ref{ssec:antibrane}. In the limit where periods are small, it is reasonable to expect that
the boundaries of all rigid holomorphic disks lie in a small neighborhood of the link. Here the disks look basically like the corresponding disks in the conormals of the link components that are glued in with a twist. We restrict attention to this case and write $x=(x_1,\dots,x_\n)$ and $p=(p_1,\dots,p_\n)$ for the longitude and meridian homology generators in $\Lambda_{K}$, where $H_{1}(\Lambda_K)=H_{1}(\Lambda_{K_1})\oplus\dots\oplus H_1(\Lambda_{K_\n})$ and $x_j,p_j\in H_{1}(\Lambda_{K_{j}})$. Since boundaries of holomorphic disks lie in a (punctured) tubular neighborhood of $K$ we fix obstruction chains from the disks' boundaries to linear combinations of fixed homology generators $x_j,p_j$ at infinity.
As above tree configurations then define a generating function (GW-potential)
\[
W(x,p)=\sum_{k,l,r} N_{k,l,r}e^{k\cdot x+l\cdot p}Q^{r},
\]
where $k=(k_1,\dots,k_\n)$ and $l=(l_1,\dots,l_\n)$ both range over all vectors in $\bZ^{n}$ and $r$ ranges over $\bZ$.

 In analogy with the above, we change variables in the coefficient ring of the DGA of $\Lambda_{K}$ by
\begin{equation}\label{eq:changeofvar}
x_j\mapsto X_j=x_j-\frac{\partial W}{\partial p_j},\qquad
p_j\mapsto P_j=p_j-\frac{\partial W}{\partial x_j},
\end{equation}
and write $\scA_{ml}$ for the DGA with transformed coefficient variables.

Let $\bC_{Q}\pi$ denote the group ring of the fundamental group of the link complement $\pi=\pi_1(S^{3}-K)$ over $\bC$, with the commutative variable $Q^{\pm 1}$ adjoined.
Define the algebra map $\Phi\colon \scA_{ml}\to \bC_{Q}\pi$ on Reeb chord generators as
\[
\Phi(a)=\sum_{\dim\scM_{\gamma}(a;\sigma)=0} |\scM_{Q^{k};\gamma}(a;\sigma)|[\gamma],
\]
where $\gamma\in\pi$ is the homotopy class of the boundary of the configuration, $\sigma$ is an obstruction chain as above, and $\scM_{\gamma}(a;\sigma)$ is the moduli space of quantum corrected holomorphic disks with boundary on $M_K$ and one positive puncture at the Reeb chord $a$. Note that this map makes sense as a map into $\bC_{Q}\pi$: since this is a module over the commutative ring $\bC[\pi_1(\Lambda_{K})]$, we need not keep track of the order of intersections of the reference chains in the boundary when all holomorphic disks lie near the link components.

The usual argument analyzing the boundaries of $1$-dimensional moduli spaces shows that
the map $\Phi\colon\scA_{ml}\to\bC_Q\pi$ is a chain map. Thus if $R_{Q}\pi$ denotes the subalgebra of $\bC_{Q}\pi$ that is generated by the image of $\Phi$, then any algebra map $R_{Q}\pi\to\bC$ induces an augmentation.
We next argue that the parts of the augmentation variety which come from representations of $R_{Q}\pi$ are Lagrangian with respect to the symplectic form $\sum_{j=1}^{n}dx_j\wedge dp_j$.

To see this, write $X_j$ and $P_j$ for the ``logarithmic coefficient variables'' in the algebra $\scA_{ml}$, see \eqref{eq:changeofvar}. Fix a link component $K_j$ of $K$ and write the function $\lk_j$, defined on generators $\gamma\in\pi$ of $\bC\pi$, for the linking number with $K_j$:
$\lk_j(\gamma)=\lk(K_j,\gamma)$. Write $\lk^{-1}(0)$ for the set of generators with vanishing linking number with each $K_j$. Since any generator $\gamma$ in $\lk^{-1}(0)$ is zero in homology, it is homotopic to a product of commutators: $\gamma=\Pi_{s=1}^{m}\alpha_j\beta_j\alpha^{-1}_{j}\beta^{-1}_{j}$ in $\bC\pi$.

Consider now an algebra morphism parametrized by $P_j$ and note that
\[
\frac{1}{\gamma}\frac{\partial \gamma}{\partial P_j}=
\sum_{s=1}^{m}\left(\frac{1}{\alpha_{s}}\frac{\partial \alpha_{s}}{\partial P_j}+
\frac{1}{\beta_{s}}\frac{\partial \beta_{s}}{\partial P_j}-
\frac{1}{\alpha_{s}}\frac{\partial \alpha_{s}}{\partial P_j}-
\frac{1}{\beta_{s}}\frac{\partial\beta_{s}}{\partial P_j}\right)=0
\]
We have
\[
e^{X_k}=e^{\sum_{j\ne k} \lk_j(X_k) P_j}\xi_k,
\]
where $\xi_k\in\lk^{-1}(0)$. Thus
\[
\frac{\partial X_k}{\partial P_j}=\lk_j(X_k)=\lk_k(X_j)=\frac{\partial X_j}{\partial P_k},
\]
since by symmetry of linking numbers, $\lk(K_j,K_k)=\lk(K_k,K_j)$.
Hence $X=X(P)$ is Lagrangian with respect to the symplectic form $\sum_{j=1}^{n} dX_j\wedge dP_j$, and since the change of coordinates $(X,P)$ to $(x,p)$ takes $\sum_{j=1}^{n} dX_j\wedge dP_j$ to $\sum_{j=1}^{n} dx_j\wedge dp_j$, we find that the original augmentation variety is Lagrangian as well.

\section{Examples}\label{sec:ex}

Here we provide a few computations of augmentation varieties for
links. The augmentation variety and the mirror-symmetry variety coincide in all cases where they have both been computed, and we denote both by $V_K$.
In addition, we present evidence for the conjectures presented in sections $4$ and $5$.

Let us summarize what we expect, based on the discussion of the previous sections. Let $K$ be an $\n$-component link.
The mirror/augmentation variety $V_K$
has a Zariski-closed subset (union of irreducible components) $V_K(P)$ for each primitive partition $P$ of $\{1,\ldots,n\}$, where we recall that primitive partitions divide $K$ into non-split sublinks $K_\alpha$. For fixed $Q$, we have
$$V_K(P)=\prod_\alpha V_{K_\alpha} (P_\alpha)$$,
where $K_\alpha$ corresponds to the subset $P_\alpha \subset \{1,\ldots,n\}$.
Moreover, for fixed $Q$, $V_{K_\alpha}\subset {(\C^*)}^{2n_\alpha}$ is expected to be a complex-Lagrangian variety of dimension $n_\alpha$, where $n_\alpha$ is the cardinality of $P_\alpha$, and the corresponding potential defining $V_{K_\alpha}$ is associated to the Lagrangian filling $M_{K_\alpha}$ as follows.  Let
$e^{x_i},e^{p_i}$ denote the coordinates of ${(\C^*)}^{2n_\alpha}$; then
$$x_i=\frac{\partial}{\partial p_i} W(p_1,...,p_{n_\alpha})$$
and
$$W(p_1,...,p_{n_\alpha})={\rm lim}_{\lambda\rightarrow 0}\ (\lambda\log
H_{{p_1/\lambda},\ldots,p_{n_\alpha}/\lambda}(K(P))),$$
%
where
$$
H_{{p_1/\lambda}, \ldots,p_{n_\alpha}/\lambda}(K(P))
$$
are the HOMFLY invariants of the link $K(P)$  colored by symmetric representations  of length $p_{n_i}/\lambda$.

Furthermore, let $P_1,P_2$ be primitive partitions such that $P_2$ is a refinement of $P_1$. If there is no primitive $P'$ such that $P'$ is a refinement of $P_1$ and $P_2$ is a refinement of $P'$, then we expect
\[
{\rm codim}(V_K(P_1) \cap V_K(P_2)) = 1.
\]
More generally, if we have a sequence $P_1,\ldots,P_r$ of primitive partitions such that each is a refinement of the previous one and the sequence cannot be lengthened by inserting another partition in the middle, then we expect
\[
{\rm codim}(V_K(P_1) \cap V_K(P_r)) \leq r.
\]
In particular, we expect $V_K(\n)= V_K(1\cdots \n)$ and $V_K(1^\n) = V_K((1)\ldots(\n))$ to intersect at least in a curve.

Finally, the case of knots and links colored by representations higher than the symmetric representation translates back to the case of links colored by totally symmetric representations. For any sublink $K_{\alpha}$ colored by a representation with $m$ rows corresponding to $m$ branes on the corresponding Legendrian, we get an equivalent description in terms of a link containing $m$ parallel copies of $K_{\alpha}$ with a single brane on each Legendrian. The only distinction between the two are the contributions of short strings, which can potentially kill some of the $V_K(P)$.

We will see that all the examples we have considered are consistent with the above predictions.
The polynomials determining the augmentation varieties in this section, along with a number of other examples, are available in a \textit{Mathematica} notebook at the third author's web site:

\centerline{
\url{http://www.math.duke.edu/~ng/}
}

\noindent
A number of the links considered in this section are depicted, along with their graphs, in Figure~\ref{fig:link-ex}.
Details of the computations are mostly suppressed; all computations use a \textit{Mathematica} package also available at the web site.
For notational purposes, note that the varieties $V_K$ and $V_K(P)$ are subsets of $(\C^*)^{2\n+1}$ with coordinates given by $\lambda_i=e^{x_i}$, $\mu_i=e^{p_i}$, and $Q=e^t$.

\begin{figure}[ht]
\centerline{
\includegraphics[width=\linewidth]{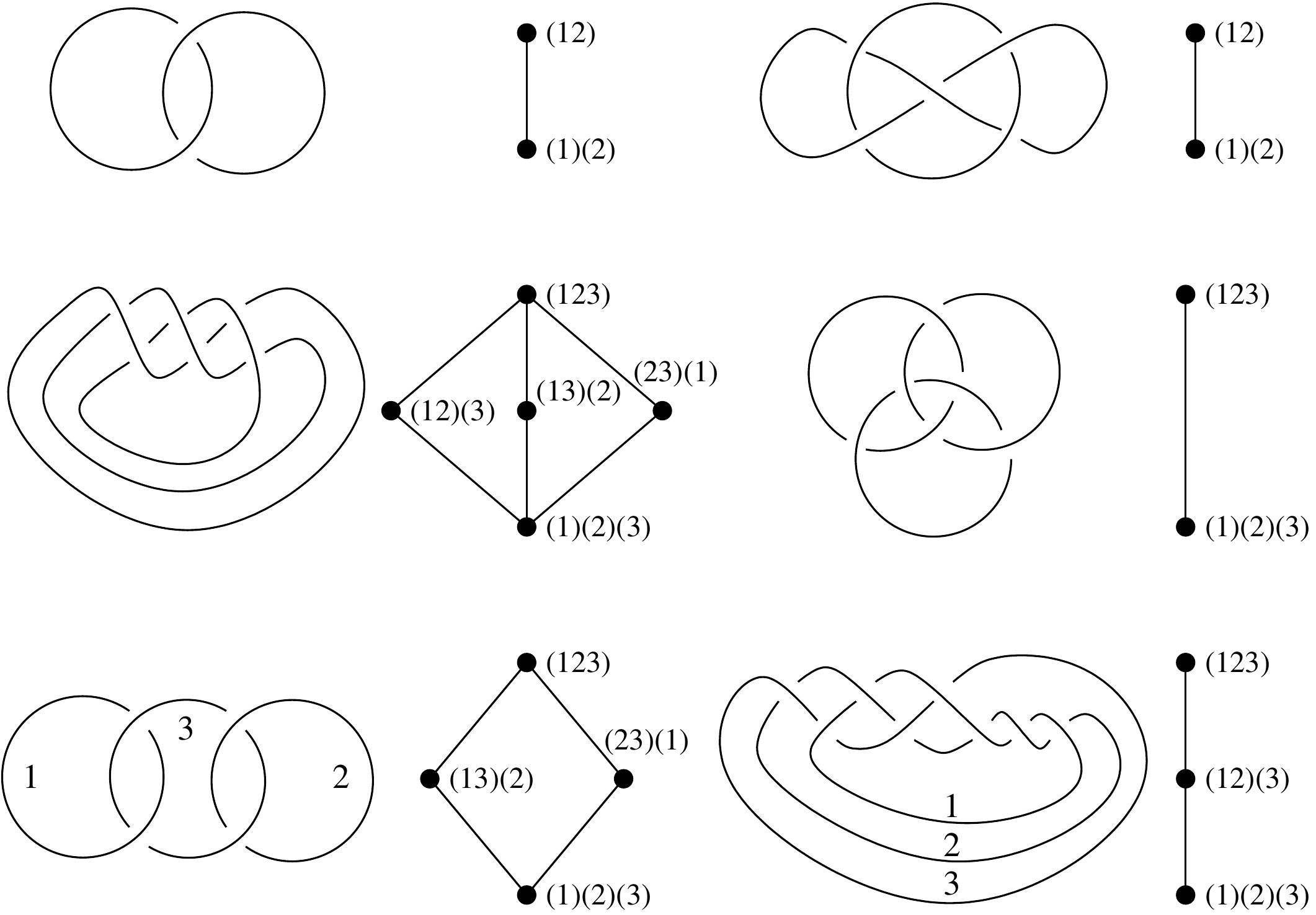}
}
\caption{
Examples of links $K$ and their graphs $\Gamma_K$. Left to right, top to bottom: Hopf link ($(2,2)$ torus link), Whitehead link; $(3,3)$ torus link, Borromean link;
connected sum of two Hopf links, L8n5. Where relevant, the components of the link are numbered; where no numbers appear, there is a symmetry given by interchanging components.
}
\label{fig:link-ex}
\end{figure}

\subsection{Torus links}
\label{ssec:toruslinks}

Here we give the augmentation varieties for various torus
links, obtained from knot contact homology.  In all cases, each link component is an unknot, and so
$V_K(1^n)$ is the vanishing set of
\[
(Q-\lambda_1-\mu_1+\lambda_1\mu_1,\ldots,
Q-\lambda_n-\mu_n+\lambda_n\mu_n),
\]
where $n$ is the number of components.

\begin{itemize}
\item
Hopf link ($(2,2)$ torus link): $V_K = V_K(1^2) \cup V_K(2)$ with
\[
V_K(2) = \left(\lambda_1-\mu_2,\lambda_2-\mu_1\right).
\]
\item
$(2,4)$ torus link: $V_K = V_K(1^2) \cup V_K(2)$ with
\[
V_K(2) = \left(\lambda_1\mu_1^2-\lambda_2\mu_2^2,~
\mu_1\lambda_1^2+(-\mu_1\mu_2^2+Q\mu_2)\lambda_1-\mu_2^3\right).
\]

\item
$(2,6)$ torus link: $V_K = V_K(1^2) \cup V_K(2)$ with
\begin{align*}
V_K(2) &= \left(\lambda_1\mu_1^3-\lambda_2\mu_2^3,~
\mu_1^3\lambda_1^3+(-\mu_1^3\mu_2^3+Q\mu_1^2\mu_2^2-\mu_1^2\mu_2^2)\lambda_1^2
\right.\\
& \left.\qquad +
(Q\mu_1\mu_2^4-\mu_1\mu_2^4-Q^2\mu_2^3)\lambda_1+\mu_2^6\right).
\end{align*}

\item
$(2,8)$ torus link: $V_K = V_K(1^2) \cup V_K(2)$ with
\begin{align*}
V_K(2) &= \left(\lambda_1\mu_1^4-\lambda_2\mu_2^4,~
\mu_1^6\lambda_1^4+(-\mu_1^6\mu_2^4+Q\mu_1^5\mu_2^3-\mu_1^5\mu_2^3+Q\mu_1^4\mu_2^2)\lambda_1^3\right.\\
&\qquad +
(-\mu_1^4\mu_2^6+2Q\mu_1^3\mu_2^5-2\mu_1^3\mu_2^5-Q^2\mu_1^2\mu_2^4)\lambda_1^2
\\
&\qquad
 \left.+(\mu_1^2\mu_2^8+Q^2\mu_1\mu_2^7-Q\mu_1\mu_2^7-Q^3\mu_2^6)\lambda_1+\mu_2^{10}\right).
\end{align*}
\end{itemize}

All of the listed examples of $(2,2f)$ torus links satisfy the Lagrangian reduction property, Conjecture~\ref{conj:reduction}. That is, if we write $K = K_1 \cup K_2$ and fix $Q$, then the intersection of $V_K(2)$ with $(\C^*)^2 \times V_{K_2}(1) = \{Q-\lambda_2-\mu_2+\lambda_2\mu_2=0\}$ projects in $\lambda_1,\mu_1$ space to a curve containing $V_{K_1}(1) = \{Q-\lambda_1-\mu_1+\lambda_1\mu_1=0\}$. Note that in general this projection contains more than just $V_{K_1}(1)$. For instance, when $K$ is the $(2,4)$ torus link, a quick calculation shows that the reduction is the union of $V_{K_1}(1)$ and $\{(Q+\mu_1)^2+\lambda_1\mu_1(1+\mu_1)^2=0\}$.


 We now also compare these computations to our Conjectures \ref{cnj:auglink} and \ref{cnj:codim1} relating the structure of the augmentation variety to the geometry of the link.  The $(2,2f)$ torus links have two knot components which are linked for $f\neq 0$ and unlinked for $f=0$. For $f=0$, we get a single primitive partition $P=(1)(2)$ (which we have written as $1^2$), and for $f\neq 0$ we get two: $P=(12)$ (which we have written as $2$) and $P=(1)(2)$. Thus, by our conjecture \ref{cnj:auglink}, we expect one (for $f=0$) or two (for $f\neq 0$) components of the augmentation variety in these cases, one of which, $V_K(1^2)$, is the variety of the unlink. This is exactly what we found.

For $f=0$, the graph $\Gamma_K$ has one vertex and no edges. For $f\neq 0$, $\Gamma_K$ has two vertices, $(12)$ and $(1)(2)$, and one edge connecting them. (This is depicted for $f=1$ in Figure~\ref{fig:link-ex}.) The intersection of the subvarieties joined by this edge has codimension $1$, in accordance with Conjecture~\ref{cnj:codim1}.


\begin{itemize}

\item
$(3,3)$ torus link: Here we have
\[
V_K = V_K((1)(2)(3)) \cup V_K((12)(3)) \cup V_K((13)(2)) \cup V_K((23)(1)) \cup V_K((123)),
\]
where $V_K((1)(2)(3)) = V_K(1^3)$ is given at the beginning of this section,
\[
V_K((123)) = V_K(3) = (\lambda_1\mu_1-\lambda_2\mu_2,~
\lambda_1\mu_1-\lambda_3\mu_3,~
\mu_1\lambda_1^3-\mu_1\mu_2\mu_3\lambda_1^2-Q\mu_2\mu_3\lambda_1+
\mu_2^2\mu_3^2),
\]
and the other components are given by
\[
V_K((12)(3)) =
(\lambda_1-\mu_2,\lambda_2-\mu_1,Q-\lambda_3-\mu_3+\lambda_3\mu_3)
\]
and cyclic permutation of indices.
\end{itemize}

One can show that $V_K(3)$ intersects $V_K((12)(3))$ in a surface, which projects in $(\lambda_1,\mu_1,\lambda_2,\mu_2)$ space to a curve containing $V_\Hopf(2)$, where the Hopf link is seen here as the sublink of the $(3,3)$ torus link consisting of components 1 and 2. Similarly, $V_K((12)(3))$ intersects $V_K(1^3)$ in a surface; this fact is a direct result of the fact that $V_\Hopf(2)$ and $V_\Hopf(1^2)$ intersect in a curve. Finally, $V_K(3)$ and $V_K(1^3)$ intersect in a curve.

The graph $\Gamma_K$ for the $(3,3)$ torus link is shown in Figure~\ref{fig:link-ex}, and the intersection dimensions are in accord with Conjecture~\ref{cnj:codim1}.

\subsection{Torus links from Chern-Simons theory}
\label{ssec:torusCS}

Before presenting other examples, in this subsection we will study torus links from the perspective of Chern-Simons theory.
We will show that for the $(2,2f)$ torus links considered in the previous subsection, the mirror variety of the link, obtained from Chern-Simons theory and large $N$ duality, agrees with the augmentation variety. Moreover, we will see that the geometry of the mirror variety is indeed captured by the link graph $\Gamma_K$. We expand our discussion to $(n,nf)$ torus links, which
provide a good testing ground of the general relation we spelled out earlier between a link ${\tilde K}$ that is an $n$-parallel of a knot $K$, and the knot $K$ itself, just colored by representations with $n$ rows.
In this case, if we let ${\tilde K}$ be the $(n,nf)$ torus link, the corresponding knot $K$ is the $(1,f)$ torus knot, i.e., the unknot in framing $f$>

Recall from Section~\ref{ssec:multiplebranes} that there is a relation between the HOMFLY polynomial $H_R(K)$ of a knot $K$ colored by representations $R$ with $n$ rows and certain colored HOMFLY polynomials of the link ${\tilde K}$ consisting of $n$ parallel copies of $K$.
One can see this by simply observing that, on the one hand, the HOMFLY polynomial in any representation is the expectation value of the Wilson loop
$$
H_R(K) = \langle {\rm Tr}_R U(K)\rangle,
$$
and on the other hand, tensoring $n$ copies of the symmetric representation, we can get any representation with up to $n$ rows, so
$$
 {\rm Tr}_{m_1}U\ldots {\rm Tr}_{m_n}U =  \sum_R N^R_{m_1,\ldots,m_n} {\rm Tr}_R U,
$$
where $N^R_{R_1,\ldots,R_n}$ are tensor product coefficients for a unitary group of rank greater than $n$. This holds inside the expectation values as well, leading to a relation between the HOMFLY's:
$$
H_{m_1,\ldots,m_n}({\tilde K}) = \sum_R N^R_{m_1,\ldots,m_n} H_R(K).
$$
The only novelty is that
$$H_{m_1,\ldots,m_n}(\tilde K) =  \langle {\rm Tr}_{m_1}U\ldots {\rm Tr}_{m_n}U\rangle,
$$
containing $n$ parallel Wilson loops along $K$, is interpreted as the expectation value of the link ${\tilde K}$ rather than the knot $K$.

For the unknot $K$ in framing $f$,
colored by an arbitrary representation $R$, the HOMFLY polynomial is
given by a well-known expression:
$$
H_R(K) = T^f_R \,S_{0R}/S_{00},
$$
where $S_{RR'}$ and $T_{RR'} = T_R \delta^R_{R'}$ are the $SU(N)$ WZW $S$ and $T$ matrices, as before. In particular, $S_{0R}/S_{00}$ is just the quantum dimension of the $SU(N)$ representation $R$. From this, the partition function of $n$ branes on $K$ is easy to compute:
$$
Z({ K})(x_1, \ldots , x_n) =
\sum_{R}H_R(K)\;{\rm Tr}_{R}(e^{x_1},\ldots,e^{x_n}) \Delta(e^{x_1},\ldots,e^{x_n}).
$$
As we show in Appendix~\ref{app:toruslink}, it follows that
$$Z({K})(x_1,\ldots,x_n) = \sum_{m_1,\ldots,m_n\geq 0}\prod_{i=1}^n  B_{f}({m}_i)e^{m_1 x_1 +\ldots +m_n x_n}\Delta(q^{m_1},\ldots,q^{m_n}).$$
In the above,  $\Delta(z_1,\ldots,z_n)$ is the usual Vandermonde determinant $ \prod_{1\leq i<j\leq n}({z_i} - {z_j})$, and $B_f(m)$ is the HOMFLY polynomial for an unknot in framing $f$, colored by a totally symmetric representation of rank $m$:
$$
B_{f}(m) \propto q^{(f-1) m^2/2}\prod_{i=1}^m \frac{1 - {\hat {Q}}q^{i-1}}{1-q^{i}}.
$$
By $\propto$ we mean equality up to overall normalization. We also drop factors that ultimately only shift what we mean by holonomy.

In the classical limit, the Vandermonde determinant does not contribute, correspondingly,
up to terms suppressed by $g_s$ relative to the leading order:
$$
Z(K)(x_1,\ldots,x_n) \sim \prod_{i=1}^n Z(K)(x_i).
$$
This is simply a product of $n$ copies of a single brane partition function
$$
Z(K)(x) = \sum_{m\geq 0}  B_{f}({m})e^{m x} \sim \exp\left({1\over g_s} W_K(x)\right).
$$
It is easy to show that $(\mu, \lambda) = (e^p,e^x)$ with $p$ defined by
$$p= \frac{\partial}{\partial x} W_K(x)$$
lies on a Riemann surface that is mirror to the unknot with framing $f$ (the unusual framing replaces $\lambda$ by $\lambda\mu^f$):
$$
V_K:\qquad (1-\mu) - \lambda \mu^f(1- Q \mu) =0.
$$

Thus, with rank $n$ representations coloring the knot $K$,
$$Z({K})(x_1,  \ldots x_n) \sim \exp\left({1\over g_s} (W_K(x_1)+\ldots +W_K(x_n))\right).
$$
Geometrically, this 
comes from $n$ Lagrangians conormal to the unknot, which corresponds to the partition
$$
P = (1)(2)\ldots(n).
$$
The corresponding mirror variety is just $n$ copies of $V_K$:
$$
V_{\tilde{K}}(1^n) = (V_{K})^n.
$$
This is consistent with the analysis of higher representations based on knot contact homology which we gave in Section~\ref{ssec:higherrank}.

Now consider the partition function of the link ${\tilde K}$:
$$
Z({\tilde K})(x_1, \ldots , x_n) =\sum_{m_1,\ldots,m_n} H_{m_1,\ldots,m_n}({\tilde K})  \;{\rm Tr}_{m_1}(e^{x_1}) \ldots {\rm Tr}_{m_n}(e^{x_n}).
$$
By either the reasoning of Section~\ref{ssec:multiplebranes}, or the above relation of the HOMFLY's and
the fact that
$$
{\rm Tr}_{R}(e^{x_1}, \ldots ,e^{x_n}) = \sum_{m_1, \ldots ,m_n} N^R_{m_1,\ldots,m_n} e^{m_1 x_1} \cdots e^{m_n x_n},
$$
it follows that
$$
Z({\tilde K})(x_1, \ldots , x_n) =Z({ K})(x_1, \ldots , x_n)/\Delta(e^{x_1},\ldots,e^{x_n}).
$$

We will consider the case $n=2$ in some detail. The partition function has two saddle points. At one, corresponding to
$P = (1)(2),$
we have
$$Z({\tilde K})(1^2)(x_1,  x_2) \sim \exp\left({1\over g_s} (W_K(x_1)+W_K(x_2))\right),
$$
where $W_K$ is the potential of a single knot, and
$$
V_{\tilde K}(1^2) = (V_K)^2.
$$
This agrees with our statements in Sections~\ref{ssec:knotparallels} and \ref{ssec:higherrank} that rank $n$ representations of a knot are a subset of the augmentations of the link which is the $n$-copy of the knot.  In this filling, both meridians are filled, $\mu_{1,2} = e^{p_{1,2}} \sim 1$, and we get two disconnected Lagrangians, as befits the fact that the corresponding partition is $P=(1)(2)$.

There is one more saddle point,  which as we will see corresponds to
$P=(12),$
and which gives a filling that is a single Lagrangian, the link complement.

As we will see, at this saddle point,  we can make the holonomy around the meridian, $\mu_{1,2}$, as large as we wish, as befits the complement of a link. In this case it is better to consider the Fourier transform of the amplitude, and consider  the coefficient $H_{m_1, m_2}$ of $e^{m_1x_1 + m_2 x_2}$. Assuming $m_1\geq m_2$,
we can write
$$
H_{m_1, m_2} (\tilde K)=\sum_{k=0}^{m_2} B_f({m_1 + k}) B_f({m_2 -k})( q^{m_1+k} -  q^{m_2-k}).
$$
Now we will show that this has a saddle point for $m_{1,2} \gg 0$.  Assume that the saddle point exists in this regime. Then, in the perturbative $g_s$ expansion around it,
we can replace the sum by an integral:
$$
H_{p_1/g_s, p_2/g_s}(\tilde K) = \int dq\; B_f(p_1+ q) B_f({p_2 -q})( e^{p_1+q} - e^{p_2-q}).
$$
It is easy to find the saddle point of this integral; we will show that it indeed occurs for  $\mu_{1,2}\gg1$. If we let
$$B_f(p) \sim \exp(U_K(p)/g_s+\ldots ),
$$
then $U_f(p)$ is the Legendre transform of the framed unknot potential $W_f(x)$, corresponding to working in terms of the meridian variable $p$, which is canonically conjugate to $x$:
$$
U_K(p) =_{crit} W_K(x) + px.
$$
This means that the potential $U_{{\tilde K}}(p_1, p_2)$ corresponding to $V_{{\tilde K}}$ is
$$
U_{{\tilde K}}(p_1, p_2) =_{crit} U_K(p_1+q) + U_K(p_2 -q),
$$
where the right hand side is taken at the critical point, obtained by extremizing with respect to $q$.
The corresponding variety $V(2)$ is given by the intersection of
\beq\label{feq}
\lambda_1=\lambda_2
\eeq
and a variety obtained by eliminating $z=e^q$ from the equations
\beq\label{zeq}\begin{aligned}
(1-\mu_1 z)(\mu_1 z)^{f} - \lambda_1 (1- Q \mu_1z) &=0\\
(1-\mu_2 /z)(\mu_2/ z)^{f} - \lambda_2 (1- Q \mu_2/z) &=0.
\end{aligned}
\eeq
We can solve for $\lambda_{1,2}$ and $z$ in these equations for any $\mu_{1,2}$, consistent with our assumption that treating
$\mu_{1,2}$ as a continuous variable was justified. Thus we have shown the saddle point corresponding to $V(2)$ indeed exists.

It is easy to show that for integer values of $f$, this reproduces the $V(2)$ varieties of $(2,2f)$ links obtained for integer values of $f$. Moreover, we can prove that, for any $f\neq 0$, the two varieties intersect on a curve.

For example, for $f=2$, corresponding to the $(2,4)$ link, the variety $V(2)$ is given by a pair of equations consisting of \eqref{feq}, which is the same for all $f$ and
$$\lambda_1^2  + \lambda_1 \mu_1 \mu_2 (1 - \mu_1 \mu_2 Q)- \mu_1^3 \mu_2^3=0.
$$
from solving \eqref{zeq}.
After a shift of framing
$$
\lambda_{1,2} \rightarrow \lambda_{1,2} \mu_{1,2}^2
$$
followed by a change of orientation $\lambda_{1,2},\mu_{1,2} \rightarrow \lambda_{1,2}^{-1},\mu_{1,2}^{-1}$,
this agrees with the $V(2)$ component of the augmentation variety of the $(2,4)$ link in the previous subsection.

For $f=3$, or the $(2,6)$ link, we get \eqref{feq} together with
$$\lambda_1^3 -
\lambda_1^2 \mu_1^2 \mu_2^2 (1 - Q + \mu_1 \mu_2 Q^2) +
 \lambda_1 \mu_1^3 \mu_2^3 (-1 + \mu_1 \mu_2 (-1 + Q)) + \mu_1^6 \mu_2^6=0.
$$
After a shift of framing
$$
\lambda_{1,2} \rightarrow \lambda_{1,2} \mu_{1,2}^3
$$
followed by inverting $\lambda_{1,2},\mu_{1,2}$, this agrees with the $V(2)$ component of the augmentation variety of $(2,6)$ link we gave earlier.


Finally, for $f=4$, or the $(2,8)$ link, we get \eqref{feq} and
$$\begin{aligned}
\lambda_1^4  +
 \lambda_1^3 \mu_1^2 \mu_2^2 (1 + \mu_1 \mu_2 (-Q + Q^2) - &\mu_1^2 \mu_2^2 Q^3)-
 \lambda_1^2 \mu_1^4 \mu_2^4 (1 - 2 \mu_1 \mu_2 (-1 + Q) + \mu_1^2 \mu_2^2 Q^2)  \\
+  &\lambda_1 \mu_1^6 \mu_2^6 (1 + \mu_1 \mu_2) (-1 + \mu_1 \mu_2 Q)+ \mu_1^{10} \mu_2^{10}=0,
\end{aligned}
$$
which gives the surface of the previous section, upon replacing
$$
\lambda_{1,2} \rightarrow \lambda_{1,2} \mu_{1,2}^4
$$
and inverting $\lambda_{1,2},\mu_{1,2}$.

The intersection of the two varieties $V_{\tilde K}(2)$ and $V_{\tilde K}(1^2)$ 
corresponds to asking that $z=z(\mu_i, \lambda_i)$, which solves the equations \eqref{zeq} provided $\mu_i$, $\lambda_i$ live on $V(2)$, in fact takes a specific value: $z=1$. This is because along the locus $z=1$, $V(2)$ and $V(1^2)$ coincide, as is easily seen from the two equations. Since a priori $z$ is a function of $\mu_i$, $\lambda_i$, this puts one additional constraint on these variables, which reduces $V(2)$ to a curve---this is the curve of its intersection with $V_{\tilde K}(1^2)$. Moreover, the curve itself is easily seen to be nothing but a copy of $V_K(1)$:
$$
V_{{\tilde K}}(2) \cap V_{{\tilde K}}(1^2) = V_{K}(1).
$$

\subsection{Whitehead link}

The Whitehead link, the closure of the $3$-braid
$\sigma_1\sigma_2^{-1}\sigma_1\sigma_2^{-1}\sigma_1$, consists of two
unknot components that have linking number $0$ but are nonetheless
linked. Here we find that $V_K
=  V_K(1^2) \cup V_K(2)$ with $V_K(1^2)$ as usual and
\begin{align*}
V_K(2) &= \left(
(\lambda_1\mu_1-\lambda_1\mu_2-\mu_1\mu_2+Q)\lambda_2
+ \lambda_1\mu_1\mu_2-Q\lambda_1-Q\mu_1+Q\mu_2, \right. \\
& \qquad \mu_1^2\mu_2\lambda_1^4 +
(-\mu_1^2\mu_2^2-Q\mu_1^2\mu_2+Q\mu_1\mu_2^2+2Q\mu_1\mu_2-Q\mu_1^2
-Q^2\mu_2+Q^2\mu_1)\lambda_1^3 \\
& \qquad
+(Q\mu_1^2\mu_2^2+Q\mu_1^2\mu_2-2Q\mu_1\mu_2^2+Q^2\mu_1^2-4Q^2\mu_1\mu_2+Q^2\mu_2^2-2Q^2\mu_1+Q^2\mu_2+Q^3)\lambda_1^2
\\
& \qquad \left.
+(-Q^2\mu_1^2\mu_2+Q^2\mu_1\mu_2^2+2Q^2\mu_1\mu_2-Q^2\mu_2^2+Q^3\mu_1-Q^3\mu_2-Q^3)
\lambda_1
+Q^3\mu_2 \right).
\end{align*}

One can show easily that $V_K(1^2)$ and $V_K(2)$ intersect on codimension $1$; this is consistent with our conjectures.
As depicted in Figure~\ref{fig:link-ex}, the graph $\Gamma_K$ has two vertices, corresponding to the partitions $(12)$ and $(1)(2)$, connected by an edge. (Note that $K(P) = K$ for $P = (12)$; for $P=(1)(2)$, $K(P)$ is the split link consisting of two unknots.)

It would be interesting to compare this with the direct computation of the mirror variety $V_K$ of the Whitehead link, based
on HOMFLY coming from symmetric representations. The recent progress in computing HOMFLY invariants of links colored by totally symmetric representations should make this possible \cite{morozov, kawagoe, ramadevi1,Zhu:2012tm, Morozov1, Morozov2, ramadevi2, ramadevi3, Morozov3}.

\subsection{Borromean link}

The Borromean link, the closure of the $3$-braid $\sigma_1\sigma_2^{-1}\sigma_1\sigma_2^{-1}\sigma_1\sigma_2^{-1}$, consists of three
unknot components, any two of which are unlinked. Here we have $V_K = V_K(1^3) \cup V_K(3)$ (note that there is no separate $V_K(P)$ for $P = (12)(3)$, $(13)(2)$, or $(23)(1)$, since any two-component sublink is split) with
\begin{align*}
V_K(3) &= \left(
(\lambda_1\mu_2-\lambda_1\mu_1+\mu_1\mu_2-Q)\lambda_2
+(-\lambda_1\mu_1\mu_2+Q\lambda_1+Q\mu_1-Q\mu_2), \right.
\\
& \qquad (\lambda_1\mu_3-\lambda_1\mu_1+\mu_1\mu_3-Q)\lambda_3
+(-\lambda_1\mu_1\mu_3+Q\lambda_1+Q\mu_1-Q\mu_3), \\
& \qquad \left. P(\lambda_1,\mu_1,\mu_2,\mu_3,Q) \right),
 \end{align*}
where $P$ is a certain polynomial of degree $6$ in $\lambda_1$ (see the \textit{Mathematica} notebook
for the precise polynomial).
It would be interesting to compare this with the direct computation of $V_K$ based
on HOMFLY coming from symmetric representations on the Borromean link.

For the Borromean link, $V_K(3)$ and $V_K(1^3)$ intersect on a surface, i.e. in codimension 1,
as expected from the D-module picture discussed in Section~\ref{sec:largeN}. For generic $Q$, one can calculate that the following surface in $(\bC^*)^6$ is contained in $V_K(3) \cap V_K(1^3)$:
\[
\{(\mu_1,\mu_2,\mu_3,\lambda_1,\lambda_2,\lambda_3) =
(Q^{1/2},\mu_2,\mu_3,-Q^{1/2},\frac{\mu_2-Q}{\mu_2-1},\frac{\mu_3-Q}{\mu_3-1})\}
\]
with two parameters $\mu_2,\mu_3 \in \bC^*$. (In fact, $V_K(3) \cap V_K(1^3)$ contains a union of three surfaces; the other two are obtained by cyclically permuting the indices $1,2,3$.)

This structure agrees with the conjectures of Sections~\ref{sec:largeN} and \ref{Sec:knotch}: in this case, the link graph $\Gamma_K$ has only two vertices, corresponding to the two primitive partitions of the link. One partition, $(123)$, is the original Borromean link itself. Removing any strand of the link, we get the unlink, so the only other primitive partition is $(1)(2)(3)$, corresponding to the unlink. The two vertices of the graph are connected, as we need to move only one component to go between the two links. The fact that there are only two vertices, one of which is $(1)(2)(3)$, is consistent with the fact that the augmentation variety has two components, one of which corresponds to the $V$ of the unlink with three components. Finally, the fact that the vertices of the graph are connected by an edge is consistent with the fact that $V_K(3)$ and $V_K(1^3)$ intersect in codimension 1. See Figure~\ref{fig:link-ex}.

\subsection{Other three-component links}

Here we briefly consider two more 3-component links. The first is the
link $K$ given by the connected sum of two Hopf links. In other words, $K$ consists of three unknots, with two of the unknots, $K_1$ and $K_2$, being
meridians of the third, $K_3$. Thus $K_1,K_3$ are Hopf linked, as are $K_2,K_3$, and $K_1,K_2$ are unlinked. This is a special case of the link considered in Section~\ref{ssec:alggeomex} (where $K_3$ here is $K_0$ there). Here there is no component of $V_K$ corresponding to the partition $(12)(3)$, since $K_1 \cup K_2$ is split, but all other partitions are primitive, and we have:
\begin{align*}
V_K(3) = V_K((123)) &= (\lambda_1-\mu_3,\lambda_2-\mu_3,(Q-\mu_3)\lambda_3+\mu_1\mu_2(\mu_3-1)) \\
V_K((13)(2)) &= (\lambda_1-\mu_3,\lambda_3-\mu_1,Q-\lambda_2-\mu_2+\lambda_2\mu_2) \\
V_K((23)(1)) &= (\lambda_2-\mu_3,\lambda_3-\mu_2,Q-\lambda_1-\mu_1+\lambda_1\mu_1) \\
V_K(1^3) = V_K((1)(2)(3)) &= (Q-\lambda_i-\mu_i+\lambda_i\mu_i)_{i=1}^3.
\end{align*}
The graph for this link is shown in Figure~\ref{fig:link-ex}, and intersections are in accordance with Conjecture~\ref{cnj:codim1}.

Another illustrative example is the 3-component link L8n5, also shown in Figure~\ref{fig:link-ex}. Here the components are all unknots, and $K_1,K_3$ are unlinked, as are $K_2,K_3$. We omit the calculation of the augmentation variety here, but simply state that there are three components to the augmentation variety, corresponding to primitive partitions $(123)$, $(12)(3)$, and $(1)(2)(3)$, and these intersect pairwise in codimension $1$. In particular, it is the case that $V_K(3)$ and $V_K(1^3)$ intersect in codimension $1$, showing that equality need not hold in the inequality
\[
{\rm codim}(V_K(P_1) \cap V_K(P_r)) \leq r
\]
in Conjecture~\ref{cnj:codim1}: choose $r=3$, $P_1=(123)$, $P_2=(12)(3)$, $P_3=(1)(2)(3)$.

\section{${\cal D}$-modules and quantization:  the D-model and the D-mirror}\label{Sec:Dmodel}
One of the most central objects of our study in this paper is the mirror variety $V_K$, the moduli space of disk instanton corrected Lagrangian fillings of the Legendrian conormal $\Lambda_K$ in the resolved conifold $Y$, associated to a link $K$. The purpose of this section is to propose a way to quantize the mirror variety $V_K$, and thereby obtain quantum invariants of $Y$ with Lagrangian branes that  approach $\R\times\Lambda_K$ at infinity.  In fact this leads to a new way of thinking
about mirror symmetry (at least for non-compact branes) which may be more general than the usual
picture and which we call the D-model on $V_K$.  Even though our primary examples
will be the mirrors for links and knots, it should be clear that our proposal is more general.

We start in the case when $K$ is a knot. Then we have a single D-brane probing the conifold $Y$. The mirror of $Y$ as probed by the brane was conjectured in \cite{AV} to be the topological B-model string on the Calabi-Yau threefold
\beq\label{bmod}
Y_K: \qquad A_K(e^x,e^p, Q)=uv.
\eeq
We consider a single B-brane wrapping the subvariety $\{(w,u,v)\in Y_K\colon w\in V_K, v=0\}$, where $V_K$ is the curve with equation $A_K=0$.  By large $N$ duality, the partition function of the B-brane should compute the HOMFLY polynomial of the knot colored by totally symmetric representations,
$$
\Psi_K(x, g_s) =  \sum_{m}H_{m}(K) \; e^{-m x}.
$$

The topological B-model on Calabi-Yau manifolds of type \eqref{bmod} was shown in \cite{IH} to be a theory of a single fermion $\psi(x)$ on the corresponding Riemann surface
$$
V_K: \qquad A_K(e^x,e^p, Q)=0.
$$
The fermion $\psi(x)$ is the operator that inserts the D-brane at a point $x$ on the Riemann surface\footnote{The Calabi-Yau manifolds that arise from mirror symmetry are a fair bit subtler, as the underlying Riemann surfaces are singular, and one has to describe how the singularities get resolved as a part of the data of quantization. Note that the fermion number conservation requires that a conjugate fermion $ \psi^*(\infty)$ be placed at infinity on $V_K$, which we will often include only implicitly.}, so that
$$
\Psi_K(x)=\langle \psi(x)\rangle_{V_K}.
$$
The expectation value of the fermion is determined by the geometry of the Riemann surface. In particular, near the classical limit,
$$
\langle \psi(x)\rangle_{V_K} \sim \exp\left(\frac{1}{g_s}W_K(x)+\ldots\right),
$$
which corresponds to $(p,x)$ lying on the Lagrangian submanifold $V_K$
$$
p =\frac{\partial W_K}{\partial x}(x).
$$

In quantum theory, as shown in \cite{IH}, $x$ and $p$ do not commute, but instead
\beq\label{Oper}
[x,p] = g_s,
\eeq
and we can view them as operators arising by quantization of $V_K$ viewed as
a Lagrangian submanifold of
$$
{\cal M}_1= T^*(T^2)
$$
with symplectic form
$$
\Omega= dx\wedge dp,
$$
which has its origin in the holomorphic three-form ${du \over u} \wedge dx\wedge  dp$ on $Y_K$.
From the present perspective, the quantization \eqref{Oper} expresses the fact that $x$ and $p$ are mirror to the holonomies of the brane on $Y$ around the longitude or $(1,0)$ and the meridian or $(0,1)$ cycles, respectively, of the torus $\Lambda_K$ at infinity. Indeed, as discussed in Section \ref{sec:SYZmirror}, the theory on the brane is Chern-Simons theory, and the commutation relation between the monodromies along the dual cycles in quantum Chern-Simons theory is well known.

The existence of a fermion and the above non-commutative structure for the space
that the Riemann surface lives on was rederived in \cite{DHSV} (see also \cite{Dijkgraaf2, Dijkgraaf3}) using string dualities, mapping
topological string amplitudes to the partition function of type IIA strings in flat space in a background with $D4$ and $D6$ branes intersecting
on a Riemann surface, leading to a fermion living on the intersection.  This picture was later simplified \cite{vstring} and reformulated in terms of  yet another topological string theory on ${\cal M}_1$.  This reformulation involves coisotropic
branes for A-models introduced in \cite{Kapustin-Orlov}
and studied in a similar context in \cite{Kapustin-Witten,Gukov-Witten}.  In that context,
one studies only the conformal theory.  However, in the context of topological strings, we consider the same worldsheet theory, but couple it to 2d gravity, i.e.~we view it as a string theory and in the usual way couple
topological strings to the worldsheet metric.
 The usual topological
string theory amplitudes can be reformulated as amplitudes of this simpler topological string theory,
where one is studying topological strings on ${\cal M}_1$ with a canonical coisotropic
brane filling ${\cal M}_1$ and a Lagrangian brane wrapped on the Riemann surface. This string
theory is simpler because there are no higher string loop corrections and the theory is exact
at the level of one-loop (annulus) computations!  We call this picture the `D-mirror' and the corresponding
model the `D-model' (so labeled for its relation to D-modules).  Note that the A-model, B-model, and D-model are dual
to each other. The A-model is the most difficult to solve, since even the genus 0 amplitude is
hard to compute.   The  B-model is simple in genus 0, but
more difficult at higher loops, and the D-model is exact at the level of one loop (and disk amplitudes).  In this way the D-model picture represents a satisfactory solution to the A-model.

We next consider the case of  many-component links. In this case we have branes with many moduli and multiple brane probes, and this naturally leads to higher dimensional mirrors.
This imposes the problem of how to define a critical topological strings in higher dimensions. We will
now argue that the above reformulation of topological strings for the case of knots can be extended to the
case of links leading to the resolution of this problem. For the case of knots or ordinary mirror
symmetry, the D-model can be viewed as simply a restatement of what one had learned in the context
of topological vertex formalism.  However, for the case of links, the D-model provides the only candidate for the mirror theory.

The D-mirror variety of a $\n$-component link $K$,
$$
V_K = \bigcup_P \;V_K(P),
$$
is a Lagrangian submanifold of the torus
$$
{\cal M}_\n= T^*{{T}^{2\n}}
$$
with coordinates
$$
x_i, \;\; p_j, \qquad i,j=1,\ldots,\n,
$$
corresponding to the complexified holonomies around longitudes and meridians of the components of $\Lambda_K$ at infinity of $Y$ (recall that $\Lambda_K$ is a union of $\n$ $T^2$'s, one for each knot component of $K$). The quantum theory on the $\R\times\Lambda_K$ brane is again $GL(1)$ Chern-Simons theory,  so we know that quantization leads to the standard commutation relations (with $x_i$ and $p_i$ non-commuting):
\beq\label{comm}
[x_i, p_j] = g_s \delta_{ij},\quad 1\le i,j\le \n.
\eeq
This means that quantization of the theory must be quantization of ${\cal M}_\n$, viewed as a phase space with symplectic form
$$
\Omega = \sum_i dx_i\wedge dp_i.
$$
Moreover, we must associate to the Lagrangian submanifold $V_K$ a wave function that corresponds to
\beq\label{wag}
\Psi_K(x_1,\ldots,x_\n) = \sum_{m_1,\ldots,m_\n}H_{m_1,\ldots,m_\n}(K) e^{-m_1 x_1-\ldots -m_\n x_\n}
\eeq
whose various semi-classical expansions lead to the Lagrangians $V_K(P)$,
$$
\Psi_K(x_1, \ldots, x_\n) \sim \exp\left(\frac{1}{g_s}W_K(P)(x_1,\ldots , x_\n)+\ldots\right).
$$

Generalizing what we already discussed in the context of mirror symmetry and knots, we propose that the mirror is based on the D-model on $V_K$.
To define it, we start with the topological A-model on ${\cal M}_\n$  studied in \cite{Kapustin-Witten,Gukov-Witten}.
This leads to quantization of the phase space, by putting a single coisotropic A-brane wrapping all of ${\cal M}_\n$. Studying the algebra of open string states on this A-brane leads to viewing coordinates of ${\cal M}_\n$ as operators satisfying \eqref{comm}. Adding an A-brane on $V_K$, the bifundamental strings with one boundary on $V_K$ and one on the coisotropic brane have a unique ground state which leads to a fermion
$\psi(x_1,\ldots, x_\n)$ living on $V_K$.  The theory of this fermion is the D-model on $V_K$. In particular, we conjecture that computing the expectation value of this fermion living on $V_K$ leads to $\Psi_K$:
$$
\Psi_K(x_1, \ldots , x_\n) = \langle \psi(x_1,\ldots, x_\n)\rangle_{V_K}.
$$
This conjecture is naturally motivated by extending the case
of knots, which we already know works, to the case of links.

Note that the D-model is exact at the annulus level. In this context we interpret the worldsheet of bifundamental string states between the coisotropic brane and the Lagrangian brane as corresponding to holomorphic annuli with one boundary on ${\cal M}_\n$ and the other on the $V_K$. Such annuli lie in the intersection ${\cal M}_\n\cap V_K$, which is similar to what happens for holomorphic curves with boundary components mapping to either ${\cal  M}_\n$ or ${V_K}$ considered in a suitable complexification of ${\cal M}_\n$, and in that setup the formal dimension of a curve of genus $g$ with $h$ boundary components
is, see e.g.~\cite{CEL},
\[
\dim=(\n-3)(2-2g-h).
\]
In particular, $\dim=0$ for the annulus $g=0, h=2$, and thus all the holomorphic annuli that contribute are rigid. (This in principle also gets a contribution from the Maslov index; the latter should vanish in this case since the annulus amplitude cannot be zero.) This is parallel to the situation when the Maslov index vanishes in the Calabi-Yau case for $m=3$, when the first factor in the dimension formula vanishes.

Our proposal for quantization of $V_K$ in terms of the A-model topological string on ${\cal M}_\n$ naturally leads to D-modules  \cite{Kapustin-Witten}, and provides an a priori way to associate a quantum D-module to the classical Lagrangian $V_K$.
 Namely, the fermion $\psi$ is a bifundamental string with one boundary on $V_K$ and the other on the coisotropic brane on ${\cal M}_\n$. As such, $\psi$, and hence $\Psi_K$ as well, provides a module for strings with two endpoints on the coisotropic brane. The latter simply correspond to functions on ${\cal M}_\n$, whose algebra is deformed by \eqref{comm}. Thus, $\Psi_K(x_1, \ldots, x_\n)$ generates a D-module for the algebra \eqref{comm}.

A strong piece of evidence for our conjecture is the fact, recently proven in \cite{Garoufalidis}, that the HOMFLY polynomial of a link  $H(K)_{m_1, \ldots, m_\n}$ colored by totally symmetric representations
 generates a D-module for the Weyl algebra ${\cal D}$, acting by
$$
e^{p_i} \,H_{{m_1}\ldots m_i \ldots{m_\n}}(K) = q^{m_i}  H_{{m_1}\ldots {m_i} \ldots {m_\n}}(K),
$$
and
$$
e^{x_i} \, H_{{m_1}\ldots {m_i} \ldots {m_\n}}(K)  =   H_{{m_1}\ldots (m_{i}+1) \ldots {m_\n}}(K),
$$
whose characteristic variety is a Lagrangian submanifold. This is equivalent,  by the discrete Fourier transform in \eqref{wag}, to the action of $x_i$ and $p_i$ on $\Psi_K(x_1, \ldots, x_\n)$.\footnote{Another way of saying the same thing is to say that $\Psi_K$ is $q$-holonomic \cite{GL}. The work \cite{Garoufalidis} showing that the HOMFLY polynomial colored by symmetric representations is holonomic generalizes earlier work of \cite{GL} showing that the colored Jones polynomial is $q$-holonomic.} This fact is predicted by our conjecture for quantization of $V_K$.
 As an aside, note that the present perspective provides a prediction
that the HOMFLY polynomial of the link $K$ colored by representations of $n$ rows is $q$-holonomic as well.


In the rest of this section we will review the work of \cite{Kapustin-Witten} and show how this applies to our context.

\subsection{Review of A-model on ${\cal M}_\n$ }

To begin with, notice that the torus
$$
{\cal M}_\n= T^*{{T}^{2\n}}
$$
is in fact hyper-K{\"a}hler. In the complex structure we have used so far on ${\cal M}_\n$, which we call $I$,
there is a holomorphic $(2,0)$ form
$$
\Omega =  \sum_i dx_i \wedge d{p}_i
$$
and a $(1,1)$ form
$$
k = i \sum_i (dx_i \wedge d{\bar x}_i +d p_i \wedge d{\bar p}_i).
$$
With this choice, the torus $T^{2\n}$ is Lagrangian, and the variety $V_K$ is holomorphic. However, if we let
$$
\Omega = \omega_J+i \omega_K,
$$
$$
k= \omega_I,
$$
then the triplet of forms
$$
\omega_I, \quad \omega_J,\quad \omega_K
$$
are the three K\"ahler forms of ${\cal M}_\n$, viewed as a hyper-K\"ahler manifold.  
Note that, in the other two complex structures, $J$ and $K$, $V_K$ becomes Lagrangian rather than holomorphic as it is for $I$. This fact will be important later.

We consider A-model defined on ${\cal M}_\n$ with real symplectic structure

$$\omega = {\rm Im}(\Omega) =\omega_K.
$$
The A-model admits, in addition to the usual Lagrangian branes, which are half the dimension of the symplectic manifold, coisotropic branes \cite{Kapustin-Orlov} that can have lower codimension. On ${\cal M}_\n$, which is hyper-K\"ahler, there is a canonical example of such a coisotropic brane of rank one that in fact wraps all of ${\cal M}_\n$. We will denote the canonical coisotropic brane by ${\cal B}_{c.c}$.  The discussion in this subsection borrows from \cite{Kapustin-Witten}.

On any coisotropic brane there is a line bundle whose curvature $F$ satisfies
$$
 (\omega^{-1} F)^2=-1
$$
where we view $\omega^{-1}F$ as a map that takes $T{\cal M}_\n$ to itself. In the present case, on the canonical coisotropic brane, $F$ is simply
$$
F = {\Re}(\Omega) = \omega_J,
$$
and the coisotropic condition is satisfied since ${\cal M}_\n$ is hyper-K\"ahler, with triplet of K\"ahler forms $\omega_I$, $\omega_J$, $\omega_K$
that satisfy
\begin{align*}
\omega_I^{-1}\omega_K &= J, \\
\omega_J^{-1}\omega_I &= K, \\
\omega_K^{-1}\omega_J &= I.
\end{align*}
Here $I,J,K$ satisfy $I^2=J^2=K^2 =-1$, $IJ = -JI= K$, etc., and are the three complex structures on ${\cal M}_\n$.
Thus, the ${\cal B}_{c.c}$ brane is a canonical coisotropic brane in complex structure $I$. In fact, with respect to the three complex structures $I,J,K$, the canonical coisotropic brane is an
$$
(A,B,A)
$$
brane. We have already seen it is an A-brane in complex structure $I$. It is also an A-brane in complex structure $K$, where we take
$\omega_I$ as the symplectic form, simply because $\omega_I^{-1} F = \omega_I^{-1} \omega_J = K$. Finally, in complex structure $J$ it is a B-brane, since in that complex structure $F=\omega_J$ is the $(1,1)$ form.

Consider now topological open string with two boundaries on the ${\cal B}_{c.c}$ brane.
The action of the A-model on a disk can be written up to $Q$-exact terms as
$$
S =  \int_{\Sigma} \Phi^*(\Omega),
$$
where we used $\Omega =  F+i \omega$. This reflects the fact that the instantons of the A-model are holomorphic in complex structure $K$,
$\omega= \omega_K$. It is easy to show that only degenerate maps contribute \cite{Kapustin-Witten}.

The algebra of open strings is associative, as usual, but due to the presence of the flux $F$, it becomes non-commutative. Since $\Omega =  \sum_i d x_i \wedge d p_i $ is exact, the action above introduces a boundary coupling
$$
S =- \sum_i\int_{\partial \Sigma}  p_i d x_i
$$
so that $p_i$ and $x_i$ are canonically conjugate and thus satisfy the standard commutation relations
$$
[p_i , x_j] = g_s \delta_{ij}.
$$
Note that $g_s$ corrections in this duality frame are not string loop corrections, but $\alpha'$ corrections instead.
Locally, this simply implies that
$$
p_i = g_s \partial_{x_i},
$$
and the coordinates $e^{x_i}$ and $e^{p_j}$ of the torus ${\cal M}_\n=({\C}^*)^{2\n}$ become operators, with
$$
e^{p_j} = e^{g_s {\partial_{x_j}}}\text{ or equivalently }e^{p_i}e^{x_j}=e^{g_s\delta_{ij}}e^{x_j}e^{p_i},
$$
and Laurent polynomials in these operators generate the (exponentiated) Weyl algebra 
\beq\label{Weyl}
{\cal D}_\n = \left\langle e^{x_1}, \ldots ,e^{x_\n}, e^{g_s {\partial_{x_1}}}, \ldots ,e^{ g_s {\partial_{x_n}}}\right\rangle.
\eeq
The algebra of open strings on ${\cal M}_\n$ is locally simply the space of differential operators of the Weyl algebra ${\cal D}_\n$ above acting on functions of $x_1,\dots,x_\n$. The elements of this algebra are operators $A\in{\cal D}_\n$ of the form
$$
A = \sum_{\k_1, \ldots, \k_\n} a_{\k_1, \ldots, \k_\n}(e^{ x_1},\ldots, e^{x_\n}) e^{ \sum_i \k_i  g_s {\partial_{x_i}} }.
$$
The quantization of ${\cal M}_\n$ defined by the A-model on ${\cal M}_\n$ is nothing but deformation quantization, where the operator product is simply the star product of functions on ${\cal M}_\n$.

\subsubsection{ Branes and ${\cal D}$-modules}

Every A-brane in ${\cal M}_\n$ in complex structure $K$
provides a module on which the elements of ${\cal D}_\n$ act. This follows from a general fact in string theory: given a pair of branes ${\cal B}$, ${\cal B}'$, the $({\cal B}, {\cal B}')$ strings form a module on which $({\cal B}, {\cal B})$ strings act by multiplication on the left.
In the present case, we apply this with the brane ${\cal B}$ being the canonical coisotropic brane.
For the ${\cal B}'$ brane, we will eventually take a Lagrangian brane defined by $V_K$ of the $\n$-component link $K$; for now, we will simply take the brane to wrap a Lagrangian submanifold $V$ of the same type, namely $V$ is given by equations\footnote{The set of equations is taken to be such that it defines a Lagrangian submanifold, by assumption.}
$$
V : \qquad A_{\alpha}(e^{x_i},e^{p_i})=0, \qquad \alpha = 1,\ldots, \n
$$
in ${\cal M}_\n$.  This brane is a Lagrangian $A$-brane in complex structures $J$ and $K$:  both $\omega_J$ and $\omega_K$, which are the $(1,1)$ forms in those complex structures,
vanish along $V$. Moreover, it is clearly holomorphic in complex structure $I$, so it defines a $B$ brane for the B-model on ${\cal M}_\n$. Thus ${V}$ gives rise to a
$$
(B,A,A)
$$
brane. It was shown  in \cite{Kapustin-Witten} that such an A-brane provides a ${\cal D}_\n$-module for the differential operators acting on the square root of the canonical bundle of $V$, ${\cal K}_{V}^{1/2}$.

%
%
 %


From quantization of $({\cal B}_{c.c}, {\cal B}')$ strings we get a single fermion living on $V$, which is naturally a section of  ${\cal K}^{1/2}_V$ \cite{Kapustin-Witten}.
Consider first the case where the $(B,A,A)$ brane is given simply by
$p_i =0$ for $i=1,\ldots \n$,
so that in this case $V$ is a flat Lagrangian, $V=N=({\mathbb C}^*)^\n$. The problem of quantizing open strings translates into a problem in supersymmetric quantum mechanics with target space $N$. To compute the spectrum of open string states, we need to find the ground states of the quantum mechanics problem. Fermions in the supersymmetric quantum mechanics on $N$ are sections of $TN$. When we quantize this, we get spinors on $N$. All together, the states of open $({\cal B}_{c.c}, {\cal B}')$ are sections of the spin bundle on $N$, which is
${\cal K}^{1/2}_{N}\otimes(\oplus_{j=0}^\n \Omega^{0,j}(N))$. These are $(0,j)$ forms with values in ${\cal K}^{1/2}_{N}$. Here we are using the fact that $N$ is a complex manifold in complex structure $I$. The supersymmetric ground states are states annihilated by the supersymmetry operator $Q$, modulo the $Q$-exact states. In this case, $Q$ is just the ${\bar \partial}$ operator acting on $\Omega^{(0,j)}$ forms on $N$; locally only the $j=0$ cohomology is nontrivial. This is where the fermion
$$
\psi(x)
$$
lives and is thus a section of ${\cal K}^{1/2}_{N}$.  When we consider more complicated $(B,A,A)$ branes $V$, locally the Lagrangian always looks like a copy of $N$. The question about the spectrum of the bifundamental strings is local, so we should always get a single fermion $\psi(x)$ at the intersection. 

\subsection{Second-quantizing D-modules}\label{sec:2ndquant}
We consider now the A-model string theory with a  ${\cal B}_{c.c}$ brane on ${\cal M}_\n$ and a Lagrangian ${\cal B}$ brane wrapping $V$. Since ${\cal M}_\n$ is hyper-K\"ahler, there are no higher string loop corrections. The theory is exact at the semi-classical level, and the fermions are free. The only corrections come from the sigma model loops, and these are controlled by $g_s$. Thus, we expect to get a theory of a single free fermion living on $V$, yet one that is a section of the ${\cal D}_\n$-module.

Consider the one point function of the fermion, viewed as the string field, inserted at a point $x$ on $V$, in the background of a single anti-fermion, which we insert at infinity,
\beq\label{corr}
\Psi(x) = \langle \psi(x)\rangle_V = \left\langle \psi(x) \psi^*(\infty)\right\rangle_V.
\eeq
As it stands, this is not gauge invariant, since $\psi$ is charged under the gauge fields on both the ${\cal B}_{c.c}$ and the ${\cal B}'$ brane. To fix this, we can modify the correlator by inserting
\beq\label{dressing}
\psi(x) \psi^*(\infty) \rightarrow \psi(x) \psi^*(\infty) \,e^{{1\over g_s}\int^x_{\infty} (A - A')},
\eeq
where $A$ is the connection on the ${\cal B}_{c.c}$ brane and $A'$ the connection on the Lagrangian brane ${\cal B}'$. The coupling to gauge fields is natural from the worldsheet perspective where the correlator is computed by the disk amplitude of the topological $A$ model, and the coupling to gauge fields comes from the standard way the target space gauge fields couple to the worldsheet of the string.\footnote{More precisely, we consider a worldsheet which is a ribbon with one boundary on the ${\cal B}_{c.c}$ brane and the other on the ${\cal B}'$ brane, and we integrate over all paths that start at the point with coordinates $x$ at the worldsheet time $\tau\rightarrow -\infty$ and go off to $x=\infty$ at $\tau \rightarrow +\infty$. The propagation in infinite time automatically projects to the vacuum states, and these correspond to fermion insertions.}

If the ${\cal B}'$ brane were the flat brane at $p_i=0$, in the flat target space ${\cal M}_\n = ({\C}^*)^{2\n}$, then we would find that $\Psi(x) = 1$, corresponding to expectation value of a fermion in a free field theory.
Here however the correlator is modified, and to leading order
\beq\label{onept}
\Psi(x)\sim \exp\left({1\over g_s} W(x)  +\ldots\right),
\eeq
where $W(x)$ is determined by the local equation of the ${\cal B}'$-brane in ${\cal M}_\n$, as the set of points with coordinates $(x,p)$ where
$$
p_i = \frac{\partial W}{\partial x_i}(x).
$$
Existence of such a function $W(x)$ is equivalent to vanishing of $\sum_i dp_i \wedge dx_i$ on $V$, which in turn follows by the definition of $V$ being a Lagrangian submanifold of ${\cal M}_\n$ in complex structure $K$.

We next discuss how this relates with \eqref{corr}. Recall that the connection on the coisotropic brane ${\cal B}_{c.c}$ does not vanish identically. Rather, the presence of curvature on the coisotropic brane,
$$
F = {\Re}\left(\sum_i d p_i \wedge dx_i\right),
$$
implies that $F=dA$, where
\beq\label{conn}
A ={\Re}\left( \sum_i p_i\,dx_i\right)
\eeq
and
$$
p_i = p_i(x)
$$
is the local equation of $V$ as embedded in ${\cal M}_\n$. In the topological A-model on any manifold, holomorphy of the amplitudes implies that the gauge field on the brane is always paired with the symplectic form $\omega$, so it is only the combination $F+i\omega = \sum_i dp_i \wedge dx_i$ that appears. Consequently, the dependence of the amplitude \eqref{corr} on $\int A$ implies it depends on $\int \sum_i p_i dx_i$.

This does not quite lead us to conclude that \eqref{corr} takes the value \eqref{onept} for the following reason: the expectation value \eqref{corr} depends on not only the gauge field $A$ on the coisotropic brane, but also the gauge field $A'$ on the Lagrangian brane. However, the condition $F=dA$ does not uniquely determine $A$ to satisfy \eqref{conn}, but only determines it up to a gauge transformation. As we shall see, the two deficiencies cancel.
Like on the coisotropic brane, on the Lagrangian brane wrapping $V$ one finds two sets of fields that are naturally paired: one set is the components $A'_i$ of the flat gauge field, and the other is scalar fields $\phi_i$ that are sections of the normal bundle to $V$ and describe deformations of the Lagrangian. Since $V$ is a Lagrangian and ${\cal M}_\n$ symplectic, the normal bundle is just the cotangent space, $T^*V$.
In the A-model in complex structure $K$, the connection on $V$ and the scalars $\phi_i$ always appear in the combination
$$
\sum_j (A_j' + i \phi_j)dx^j = \sum_j p_j'(x) dx^j,
$$
where we defined $p_j'(x) = A_j'(x) + i \phi_j(x)$ which in effect parametrizes a choice of $GL(1)$ connection on $V$.
This gives us two different yet equivalent ways of deforming the Lagrangian: one by varying the local defining equation $p_i = p_i(x)$, and the other by turning on the $GL(1)$ connection $p_j'(x)$. We can use the fact that that $p_j'$ itself is not gauge invariant to soak up the gauge dependence of the statement that $A_i = p_i$. We see that the only invariant combination is the difference $p_i-p_i'=\frac{\partial W}{\partial x_i}(x)$, and that describes the embedding of the branes.
Thus, gauge invariance of the amplitude \eqref{corr} and holomorphy of the topological A-model imply that the one point function takes the value \eqref{onept}.

We conclude that the A-model on ${\cal M}_\n$ with the two sets of A-branes, the canonically coisotropic brane ${\cal B}_{c.c}$ and the Lagrangian brane on ${\cal B}'$, leads to a \emph{holonomic} D-module for the Weyl algebra ${\cal D}_\n$. Among all the D-modules, the holonomic ones are special: they arise from quantization of Lagrangian (rather than coisotropic) submanifolds, and hence precisely correspond to ${\cal B}'$ branes that are Lagrangian. In particular, it is well known that holonomic D-module are cyclic, i.e.~generated by a single element.
This corresponds to the fact that among all the $({\cal B}_{c.c},{\cal B}')$ strings there is a unique special one, the ground state, which is our fermion.

The D-module ${\cal V}$ is the set of all operators in ${\cal D}_\n$ modulo the ideal ${\cal I}\subset {\cal D}_K$ of those that annihilate $\Psi(x)$:
$$
{\cal V} = {\cal D }/{\cal I}, \quad {\cal I}=\{A\in{\cal D}_\n\colon A\Psi=0\}.
$$
The fact that ${\cal V} $ is generated by a single element $\Psi$ means that any element $\Theta\in {\cal V}$ can be obtained by acting on $\Psi$ with elements of ${\cal D}_\n$, $\Theta=A\Psi$ for some $A\in{\cal D}$.
The ideal ${\mathcal I}$ is typically generated by 
a finite set of $m$ generators $A_\alpha$ satisfying
\beq\label{Dmod}
A_{\alpha}(e^{x_i}, \, e^{  g_s {\partial_{x_j}}})   \Psi =0,
\eeq
so that
$$
{\cal I} = {\cal D}_\n \cdot \langle A_1, \ldots, A_{m} \rangle.
$$
Recall that any element of ${\cal D}_\n$ can be written in the form
$$
A = \sum_{\k_1,\ldots,\k_\n} a_{\k_1, \ldots, \k_\n}(e^{ x_1},\ldots, e^{x_\n})\, e^{ \sum_i \k_i  g_s {\partial_{x_i}} }.
$$
From the semi-classical limit of $\Psi$ in \eqref{onept}, it follows that the classical Lagrangian $V$ that we started with is the set of points in ${\cal M}_\n$ satisfying
$$
A_{\alpha}(e^{x_i}, e^{p_j})=0, \qquad \alpha=1,\ldots m,
$$
which in the language of ${\cal D}$-modules means that $V$ is the characteristic variety of the ${\cal D}$-module ${\cal V}$.

There is another, equivalent perspective on the problem.
Every A-brane wrapping an arbitrary submanifold $V$ of ${\cal M}_\n =
({\C}^*)^{2\n}$ can be viewed as $r$ branes on  $N={(\mathbb C^*)}^\n$ for some $r$. Such a brane is equipped with a flat ${GL(r)}$ connection $A'$, which complexifies the $U(r)$ gauge group on the branes, with
$$
dA' + {1\over g_s}  A'\wedge A' =0.
$$
There is a set of $r$ fermions ${\Psi}_a$, $a=0, \ldots, r-1$, transforming in the fundamental representation of ${GL(\n)}$, that give rise to a covariantly constant section with respect to this connection:
$$
\left(\partial_i  +  {1\over g_s} A'_i\right) \begin{pmatrix}
\Psi_0 \\
\vdots \\
\Psi_{r-1}
\end{pmatrix}=0.
$$
The flatness of the $GL(\n)$ connection guarantees that there are $r$ linearly independent solutions to this equation. The $r$ linearly independent solutions naturally correspond to $r$ fermions living on the $r$ sheets of $V$.
The fact that these two different descriptions exist is well known in the theory of D-modules: the fact that they get related to flat connections is one of the key aspects of the theory. Moreover, the relation between the two pictures is also well known.  One can prove that there is a basis and a set of $r$ operators $P_0, P_1, \ldots, P_{r-1}$, which are generators of the module ${\cal V}$,  such that
$$
\Psi_0 = P_0\Psi, \qquad \Psi_1 = P_1 \Psi, \qquad \ldots \qquad ,\Psi_{r-1} = P_{r-1} \Psi.
$$
The operators $P_0, \ldots  ,P_{r-1}$ form a basis of the module ${\cal V}$ viewed as a free module over functions on $({\C}^*)^{\n}$ (that is, with coefficients given by functions of $e^{x_i}$).  Different choices of basis are related by $GL(\n)$ gauge transformations. In particular, we can always choose $1$ to be one of the generators, so up to $GL(\n)$ transformations, we can set $P_0=1$.
In this description, we effectively flatten out the Lagrangian brane, and make it manifest that the theory on the branes is a free theory.

\subsection{Knots and D-modules}
We next consider the general theory developed in Section \ref{sec:2ndquant} in the special case of Lagrangian submanifolds $V=V_K$ of ${\cal M}_\n$ where $V_K$ is the variety associated to a link $K$. Above we treated $V$ as smooth, whereas, as we saw in Sections \ref{sec:largeN}, \ref{Sec:knotch}, and \ref{sec:augmentations}, $V_K$ is a union of Lagrangians
$$
V_K = \bigcup_P V_K(P)
$$
intersecting according to the graph $\Gamma_K$. At the outset, classically,  the fact that $V_K$ is reducible implies that the connection $A'$ on the $V$-brane is reducible, and we get one fermion $\psi_P$ for each $P$. From the physical perspective, when the link is non-split, we expect to get a single theory that unifies all these different $V_K(P)$  into a single object, quantum mechanically. Thus one would expect all the different components of $V_K$ to play a role in the quantization.  The fact that the characteristic varieties of irreducible ${\cal D}$-modules are in general reducible is a general feature of ${\cal D}$-modules, and indeed, here it is naturally forced on us.

Now, recall that we observed that in all computed cases, the graph $\Gamma_K$ encodes the geometry of the intersections of the components in $V_K$. The vertices of the graph correspond to components $V_K(P)$ and for a pair of vertices  $P$, $P'$ connected by an edge of the graph, $V_K(P)$ and $V_K(P')$  intersect in codimension $1$. For every intersection between $V_K(P)$ and $V_K(P')$ we get additional states at the intersection which are bifundamental $(P,P')$ strings.
 If the intersection is in codimension $1$, then by giving expectation value to these bifundamental matter fields, one can fuse the two Lagrangians.\footnote{Near the  intersection we can model the two Lagrangians by $p_1 x_1=0$, and $p_2=\ldots=p_{\n}=0$, which can be smoothed out to $p_1 x_1 = \mu $, which describes the single Lagrangian. Here 
 $\mu\sim \lambda$, since at $\lambda=0$ we must get the classical Lagrangian that factors.}
One can think of these as corresponding to the off-diagonal components of the connection $A'$ on $({\C}^*)^{\n}$, connecting blocks corresponding to $P$ and $P'$.

This structure matches the mathematical theory very well: it is known \cite{kashiwara} that every ${\cal D}$-module whose characteristic variety is a union of two components that intersect in codimension greater than $1$ is in fact reducible. Thus, both physically and mathematically, the interaction between the component $V_K(P)$ for different $P$ comes from intersections of codimension $1$.

\subsection{Examples}

The simplest case to consider is when the link is totally split: each component is separated from all the rest.  In this case, as already noted in Section \ref{Sec:knotch}, the augmentation variety has a single component, with $P=1^\n$. This is the product of $V(1)'$s of the individual knots, and all the other $V_K(P)$'s are empty.  This also implies that
$$
V_K=\prod_i V_{K_i}(1)
$$
and that 
${\cal B}_{c.c}$ is the product
of the individual ${\cal B}_{c.c}$'s for each link component. Thus the partition function of the topological
string should be the product of the partition function for each knot component.  This is consistent with the factorization
property of HOMFLY for totally split links:
$$
H_K=\prod_{i=1}^\n H_{K_i},
$$
where $H_{K_i}$ is the HOMFLY invariant associated with the $i$-th knot component $K_i$.

Now let us consider more general links.  For every link $K$, $V_K$ will include $V_K(1^\n)$, which is indeed known to be part of the
full amplitudes, where we ignore worldsheet configurations that can end in distinct knots.
If there were no other components
to $V_K$, or if there were other components but all intersected over codimension higher than $1$, then we would again simply get the above product formula for the HOMFLY, since the theory on each component is that of a free fermion. However, there is a correction
to this formula, because there {\it are} components to $V_K$ other than $V(1^\n)$, with the 
graph $\Gamma_K$ describing how $V_K(P)$'s intersect each other pairwise over codimension-$1$ loci. On each of these codimension-$1$ loci, there is a bifundamental matter field coming from the off-diagonal components of the gauge field $A'$. If these bifundamentals get expectation values of order $g_s$, then they are not visible classically, so classically the connection is reducible and we get a free fermion for each $P$. At finite $g_s$ we end up unifying the different components. So the intersection
locus and the degrees of freedom of open strings localized at such intersections could in principle provide the necessary correction of order $g_s$.

To see how this may come about, it is instructive to consider the simple example of the Hopf link.  As described previously, the
two components of $V_K$ are given by:
\begin{align*}
V_K(1^2):  &\qquad Q-X_i-P_i+X_iP_i=0, \qquad i=1,2 \\
V_K(2): &\qquad  X_1=P_2 , \quad P_1=X_2.
\end{align*}

If the two unknot components were unlinked, we would only have $V_K(1^2)$, and we would get the product of the partition
function of two unknots.  However, the fact that $V_K(2)$ is also there changes the story.
From the HOMFLY polynomial, we can read off the exact answer for the ideal $I$ of operators that annihilate $\Psi$. We can write this in two different ways. One is from the perspective of the coisotropic brane, describing which of the ${\cal B}_{c.c}-{\cal B}_{c.c}$ strings act trivially on the fermion. This gives the ideal
$I$ determined by three quantum difference equations:
\begin{align*}
\Bigl((-X_1+X_2) - (1-Q X_1)P_1 + (1-Q X_2) P_2\Bigr)\Psi&=0
\\
\Bigl(  (1-q^{-1} X_2)X_2 -(1-q^{-1} X_2)P_1 - (1-Q X_2) q^{-1}X_2 P_2 + (1-Q X_2)P_1P_2  \Bigr)\Psi&=0
\\
\Bigl( (1-q^{-1} X_1)X_1 -(1-q^{-1} X_1)P_2   - (1-Q X_1) q^{-1}X_1 P_1 + (1-Q X_1)P_1P_2\Bigr)\Psi&=0.
\end{align*}
In the above, 
$X_{1,2} =e^{x_{1,2}}$ and $P_{1,2} = e^{g_s \partial_{x_{1,2}}}$ are operators that do not commute, due to the flux on the coisotropic brane.
Viewed as a set of linear differential equations, this has two solutions, as discussed in Section~\ref{sec:largeN}, associated to $V(1^2)$ and $V(2)$, so that the general solution is
$$
\Psi = c_{V(1^2)} \Psi_{V(1^2)} + c_{V(2)} \Psi_{V(2)},
$$
with two a priori arbitrary constants $c$. As previously discussed, $\Psi_{V(2)}$ simply solves
\begin{align*}
(P_1 -X_2)\Psi_{V(2)}&=0, \\
(P_2-X_1)\Psi_{V(2)}&=0,
\end{align*}
but  $\Psi_{V(1^2)}$ is more complicated.

Another way to characterize this is from the perspective of the Lagrangian brane. From this perspective, the fact that the brane has two components $V_K(1^2)$ and $V_K(2)$, each of which covers the $X_1,X_2$ plane once, means that we can describe the vacuum of the theory in terms of a flat rank $2$ connection on $N = ({\mathbb C^*})^2$. Indeed, one can show that the above system of equations is equivalent to a set of two linear equations:
\begin{align*}
P_1  \left(\begin{array}{c}
\Psi_0 \\
\Psi_{1}
\end{array}\right) &= \begin{pmatrix}
X_2 & (1-Q X_2)(1-X_1)/(1-QX_1) \\
0& (1-X_1)/(1-QX_1) \end{pmatrix} \left(\begin{array}{c}
\Psi_0 \\
\Psi_{1}
\end{array}\right) \\
P_2  \left(\begin{array}{c}
\Psi_0 \\
\Psi_{1}
\end{array}\right) &= \begin{pmatrix}
X_1 & 1-X_1 \\
0& (1-X_2)/(1-QX_2) \end{pmatrix}\left(\begin{array}{c}
\Psi_0 \\
\Psi_{1}
\end{array}\right),
\end{align*}
where $\Psi_0$ is what we previously called $\Psi$. 

In this language, we can see how the solution corresponding to $V_K(1^2)$ is corrected by the presence of $V_K(2)$. To get $\Psi_{V_K(1^2)}$, consider the projection sending $(\Psi_0, \Psi_1)$ to $\Psi_1$. This gives the $V_K(1^2)$ brane embedded in ${\cal M}_2$ in isolation, associated to two unlinked unknots.  This is not the exact solution, since the off-diagonal components of the connection do not vanish. But once we have picked a specific solution for $\Psi_1$, it determines $\Psi_0$ in perturbation theory for small $g_s$. In particular, one can easily show that classically, $\Psi_{V(1^2)} =\Psi_1(1+{\cal O}(g_s))$, where the subleading corrections are determined from $\Psi_1$ by off-diagonal components of $A'$. Finally, to this solution, we can add any multiple of a solution associated to $V_K(2)$.

Finally, note that in the case of the Hopf link, we can recover all the data of the link from studying more branes---or more fermions, on the Riemann surface $V_{\bigcirc}$ mirror to the unknot. The exact answer for the HOMFLY for the case of the Hopf link is known to be obtainable
by putting two branes on the curve associated to the conifold, one along the $x$-leg and one along the $p$-leg \cite{IH}. We have seen examples of this in Section \ref{sec:ex} as well, when we studied torus links. This should be an example of a more general phenomenon,
where the topological string on ${\cal M}_k$ for any $\k$  provides a way to get a hierarchy of ${\cal D}$-modules for free, simply by inserting more fermions. For any number $\k$ of fermions inserted, we get a module for the $\k\n$-th Weyl algebra, and moreover all the ${\cal D}$-modules that arise are holonomic, with the $\n$-fermion correlation function being the cyclic generator of the ${\cal D}$-module.

As an example, consider the link ${\tilde K}$ given by $\n$ parallel copies of a given knot ${ K}$. In this case, we have two descriptions of the theory. One is the standard one in terms of ${\cal M}_\n=T^*(T^{2\n})$ with a ${\cal B}_{c.c}$ brane and a single Lagrangian brane wrapping $V_{\tilde K}$. In this theory, the partition function of the system

\beq\label{lp}
Z_{\tilde K}(x_1, \ldots, x_\n) = \sum_{m_1, \ldots,m_\n} H_{m_1, \ldots, m_\n}({\tilde K} )e^{-m_1 x_1-\cdots -m_\n x_\n}
\eeq
should be the one point function of the fermion
\beq\label{opt}
Z_{\tilde K}(x_1, \ldots, x_\n) =  \langle\psi(x_1, \ldots, x_\n)\rangle_{V_{\tilde K}}
\eeq
inserted at a point in $V_{\tilde K}$.  One of the components of $V_{\tilde K}$ is $V_{\tilde K}(1^\n)$; this simply equals $\n$ copies of $V_K(1)$.  For most framed knots $K$, ${\tilde K}$ is a nontrivial link, and when this is the case, $V_{\tilde K}$ has additional components.\footnote{Only when $K$ is the unknot in trivial framing is ${\tilde K}$ a split link.} Correspondingly, the partition function \eqref{opt} does not factor. Instead, the $S_\n$ symmetry of the link translates into the permutation symmetry of $V_{\tilde K}$ and the fermion one point function is totally symmetric in its $\n$ variables.

The second description is based on the Riemann surface $V_K$ in ${\cal M}_1=T^*(T^{2})$ which is the mirror for the knot ${ K}$. We insert $\n$ fermions
$$\langle\psi(x_1) \ldots \psi(x_\n)\rangle_{V_{K}}
$$
at $\n$ points on the Riemann surface $V_{{\tilde K}}$ (with $\n$ fermions at infinity to cancel the charge).
Since the $\n$ fermions are totally identical particles that anticommute, the partition function is totally antisymmetric. As we discussed in Section \ref{sec:largeN},  the correspondence between the two pictures is
$$\langle\psi(x_1,\ldots,x_\n)\rangle_{V_{\tilde K}} = \langle\psi(x_1) \ldots \psi(x_\n)\rangle_{V_{K}} /\Delta(x_1, \ldots , x_\n).
$$
Dividing by the Vandermonde restores the invariance of the partition function under all permutations of $\n$ strands.
We saw an example of this in the case of $(n, nf)$ links we studied in Section~\ref{ssec:torusCS}. There we showed that the partition function \eqref{lp} of the link has the form:\footnote{In the present case,
$
w_i(x) =\sum_m B_f(m) q^{m i} e^{-m x},\;\;i=1, 2, \ldots,
$
and $B_f(m)$ is given in Section~\ref{ssec:torusCS} and Appendix \ref{app:toruslink}.}
\beq\label{lpb}
\langle\psi(x_1, \ldots, x_\n)\rangle_{V_{\tilde K}} ={\rm det}(w_{i-1}(x_j))/\Delta(x_1, \ldots , x_\n).
\eeq
The simplicity of this expression has a natural explanation in the second viewpoint: the right hand side is the wave function of $\n$ free fermions on $V_K$, provided we treat the fermions as sections of the ${\cal D}$-module. Here $w_j(x)$ is a basis of fermion wave functions that arises from $V_K$. These shift the vacuum of the theory on $V_K$ from the canonical one, where the single particle wave functions are $z_j(x) = e^{j x}$, to $w_j$, so that
$
\langle\psi(x_1) \cdots \psi(x_\n)\rangle_{V_{K}} ={\rm det}(w_{i-1}(x_j)).
$

\section*{Acknowledgments}
We are greatly indebted to Robbert
Dijkgraaf for collaboration at the early stages of this work.

 We would like to thank SCGP where part of this work was done in the tenth Simons Workshop on math and physics.   M.A. and C.V. would also like to thank the Simons Foundation for hospitality at the Simons Symposium on Knot
Homologies and BPS States, April 2012.
In addition we would like to thank Mohammed Abouzaid, Denis Auroux, Pavel Etingof, Sergei Gukov, Albrecht Klemm, Andrei Okounkov, Ivan Smith, and Maxim Zabzine for valuable discussions.
The research of M.A. is supported in part by the Berkeley Center for Theoretical Physics, by the National Science Foundation (award number 0855653), by the Institute for the Physics and Mathematics of the Universe, and by the US Department of Energy under Contract DE-AC02-05CH11231.
T.E. is supported in part by VR grant 2012-2365 and by the Knut and Alice Wallenberg Foundation.
L.N. is supported in part by NSF grant DMS-0846346.
C.V. is supported in part by NSF grant PHY-0244821.

\appendix
\section{Hopf link}\label{app:Hopf}
 The simplest nontrivial link is the Hopf link. The corresponding Chern-Simons expectation value is given by
the matrix element of the $SU(N)_k$ $S$ matrix in the representations coloring the two components of the link,  normalized as in Section~\ref{Sec:review} to set the partition function with no insertions equal to 1.
In the present case, we take in addition the representations to be totally symmetric:
\[
H_{m_1,m_2}^{\Hopf}=S_{m_1, m_2}/S_{00},
\]
where $m_{1,2}$ as before denote totally symmetric representations with $m_{1,2}$ boxes. The corresponding wave function
$$
\Psi_{\Hopf}(x_1,x_2)=\sum_{m_1, m_2} H_{m_1, m_2}^{\Hopf} e^{m_1 x_1} e^{m_2 x_2},
$$
after redefining $e^{-x_i}$ by a constant shift of $(Qq^{-1})^{1/2}$, can be written as
$$
\Psi_{\Hopf}(x_1,x_2) = \sum_{m_2} \frac{(Q,q)_{m_2}}{(q,q)_{m_2}} e^{m_2 x_2}\times \frac{1}{1- q^{m_2} e^{x_1}  }\times
 \sum_{m_1} \frac{(Qq^{-1},q)_{m_1}}{(q,q)_{m_1}} e^{m_1 x_1};
$$
while it is not manifest from this expression as written, one can show that the partition function is invariant under exchanging the two knots.

We will now show that this has two distinct saddle points. Note that we can rewrite the amplitude as
$$
\Psi_{\Hopf}(x_1,x_2) ={1\over 1-  e^{g_s \partial_{x_2} }e^{x_1}} \Psi_{\bigcirc}(x_1, Q)\Psi_{\bigcirc}(x_2, Qq^{-1}),
$$
where $ \Psi_{\bigcirc}(x, Q)$ is the partition function corresponding to the unknot in the shifted variable $x$:
$$
 \Psi_{\bigcirc}(x, Q) = \sum_{m} \frac{(Q,q)_{m}}{(q,q)_{m}} e^{m x}.
$$
One of the saddle points, corresponding to $V_{\Hopf}(1^2)$, corresponds to the wave function being dominated by
$$
\Psi_{\Hopf}(x_1,x_2) ={1\over 1-  e^{-p_2 }e^{x_1}} \exp({1\over g_s} W_{\Hopf}(1^2)(x_1, x_2)+\ldots),
$$
where the missing terms are suppressed by powers of $g_s$. The saddle point potential is
$$
W_{\Hopf}(1^2)(x_1,x_2) = W_{\bigcirc}(x_1)+W_{\bigcirc}(x_2),
$$
the sum of two unknot potentials $W_{\bigcirc}(x_i, Q)$, with longitude $x_i$. Denoting
$$
 e^{x_i} = \lambda_i , \qquad  e^{p_i} = e^{-\partial_{x_i} W_{\bigcirc}(x_i, Q)}=\mu_i,
$$
then $\lambda_i$, $\mu_i$ live on a $2$-complex-dimensional variety which is Lagrangian and a direct product of two copies of the Riemann surface mirror to the unknot:
$$
V_{\Hopf}(1^2): \qquad Q- \lambda_1 - \mu_1 + \lambda_1 \mu_1=0, \qquad Q- \lambda_2 - \mu_2 + \lambda_1 \mu_1=0.
$$

The potential $W_{\Hopf}(1^2)(x_1,x_2) $ dominates the saddle point as long as the factor ${1\over 1-  e^{-p_2 }e^{x_1}}$
is analytic there. To study the opposite case, consider the Fourier transform ${\tilde \Psi}_{\Hopf}$ of $\Psi_{Hopf}$:
$${\tilde \Psi}_{\Hopf}(p_1, p_2) = \int dx_1 dx_2\; e^{p_1 x_1/g_s + p_2 x_2/g_s} \Psi_{Hopf}(x_1,x_2).
$$
The integral over $x_2$ at fixed $x_1$ is a Fourier transform
$$
{\tilde \Psi}_{\Hopf}(p_1, p_2) = \oint dx_1 \; {1\over 1-  e^{-p_2 }e^{x_1}} \exp(-p_1 x_1/g_s) \Psi_{\bigcirc}(x_1) \Psi_{\bigcirc}^{-1}(p_2),
$$
where we used the fact that the Fourier transform of the unknot inverts the wave function. Now, due to the pole in the integrand, the integral will depend on the contour. The saddle point corresponding to $V_{\Hopf}(2)$ corresponds to choosing a contour which makes a small circle around the pole at $e^{p_2} = e^{x_1}$. Making this choice, the integral is the residue of the pole
$$
{\tilde \Psi}_{\Hopf}(p_1, p_2)= e^{-p_1 p_2 /g_s},
$$
which is exact up to non-perturbative corrections. The potential corresponding to this is simply
$$
W_{\Hopf}(2)(p_1,p_2)=p_1p_2
$$
and the corresponding mirror variety is
$$
V_{\Hopf}(2): \qquad  \lambda_1 =  \mu_2, \qquad  \lambda_2= \mu_1.
$$

We could have chosen a different contour, which would have led to the ordinary Fourier transform, and we would have gotten the saddle point corresponding to $P=1^2$ instead. The two distinct fillings of $\Lambda_K$ at infinity, associated with the partitions $1^2$ and $2$,
thus correspond to two distinct contours of integration that pick out different saddle point contributions to ${\tilde \Psi}_{\Hopf}(p_1, p_2).$

\section{Chern-Simons and torus link invariants}\label{app:toruslink}

Here we consider the case of a $(n,nf)$ torus link ${\tilde K}$, for $f\in {\mathbb Z}$. This link consists of $n$ parallel copies of an unknot $K$ in framing $f$, i.e., a $(1,f)$ torus knot.
The HOMFLY polynomial for the unknot in this framing colored by an arbitrary representation $R$ is
$$
H_R(K) = T^f_R S_{0R}/S_{00},
$$
where $S_{RR'}$ and $T_{RR'} = T_R \delta^R_{R'}$ are the $SU(N)$ WZW S and T matrices as before. The explicit expressions for these are well known.
Using the fact that $S_{0R}/S_{00}$ is the dimension of the corresponding quantum group representation,
we have
$$
S_{0R}/S_{00} =
\prod_{s \in R}(Q q)^{-1/2} q^{-a'(s) } \frac{1-Q q^{a'(s) - \ell'(s)}}{1-q^{-a(s)-\ell(s)-1}};
$$
here the product runs over all boxes at position $s=(i,j)$ in the Young diagram corresponding to $R$, and $a(s) = R_i -j$, $\ell(s) = R_j^T-i$,
$a'(s) = j-1$, $l'(s)=i-1$. Furthermore,
$T_R = q^{{1\over 2}(\lambda_R, \lambda_R + 2\rho)}$, where $\lambda_R$ is the highest weight vector of an $SU(N)$ representation $R$, and $\rho$ is the Weyl vector. This gives
$$
T_R = q^{{1\over 2} \sum_i R_i(R_i - 2i+1)} Q^{\sum_i{R_i}/2}.
$$

There is a nice way to rewrite this by considering the representations $R_n$ whose Young diagrams have at most $n$ of rows of length $R_i$ for $i=1,\ldots,n$.
Then
$$
H_{R_n}(K)  \sim \prod_{i=1}^n  B_f({m_i})  \prod_{1\leq i< j\leq n} \Delta(q^{{m}_1}, \ldots , q^{{m}_n})
$$
where ${m}_i = R_i + n-i$,
and
$$
B_f(m) = q^{(f-1) m^2/2} \; \prod_{i=1}^m\frac{1 - \hat{Q}q^{i-1}}{1-q^{i}} {\hat Q}^{(f-1) m/2}\; q^{-m (n (f-1)+1)/2}
$$
with ${\hat Q} = Q q^{-n+1}$; here we have dropped the overall normalization, and $\Delta(z_1, \ldots , z_n) = \prod_{i<j}(z_i-z_j)$ is the Vandermonde determinant. Note that the shift that replaces $R_i$ by $m_i$ simply ensures that no two $m_{i,j}$ are equal.

Next, consider
$$
Z({\tilde K})(x_1, \ldots, x_n) =\sum_{R_n} H_{R_n}(K) {\rm Tr}_{R_n}(e^x_1,\ldots,e^x_n) .
$$
Using the fact that
$$
{\rm Tr}_{R_n}(e^x_1,\ldots,e^x_n) = \frac{{\rm det}(e^{x_j ({R}_i+n-i)})}{\Delta(e^{x_1},\ldots e^{x_n})}
$$
and antisymmetry, we can replace
$$
Z({\tilde K})(x_1, \ldots, x_n) =\sum_{m_1>m_2\ldots >m_n\geq 0}  B_f(m_i) \; \Delta(q^{{m}_1}, \ldots , q^{{m}_n})
 \frac{{\rm det}(e^{x_j {m}_i})}{\Delta(e^{x_1},\ldots e^{x_n})}
$$
with
$$
Z({\tilde K})(x_1, \ldots, x_n) =\sum_{m_1,m_2\ldots ,m_n\geq 0}  B_f({m_i})  \;
e^{m_1 x_1+\ldots +m_n x_n} \frac{ \Delta(q^{{m}_1}, \ldots , q^{{m}_n})}{\Delta(e^{x_1},\ldots e^{x_n})}.
$$

\bibliographystyle{utcaps}	
\bibliography{myrefs-2}	

\providecommand{\href}[2]{#2}\begingroup\raggedright\begin{thebibliography}{10}

\bibitem{AV}
M.~Aganagic and C.~Vafa, ``{Large N Duality, Mirror Symmetry, and a Q-deformed
  A-polynomial for Knots},''
\href{http://arxiv.org/abs/1204.4709}{{\tt arXiv:1204.4709 [hep-th]}}.

\bibitem{Gukov:2003na}
S.~Gukov, ``{Three-dimensional quantum gravity, Chern-Simons theory, and the A
  polynomial},'' \href{http://dx.doi.org/10.1007/s00220-005-1312-y}{{\em
  Commun.Math.Phys.} {\bf 255} (2005)  577--627},
\href{http://arxiv.org/abs/hep-th/0306165}{{\tt arXiv:hep-th/0306165
  [hep-th]}}.

\bibitem{2009ForPh}
R.~{Dijkgraaf} and H.~{Fuji},
  \href{http://dx.doi.org/10.1002/prop.200900067}{``{The volume conjecture and
  topological strings},''{\em Fortschritte der Physik} {\bf 57} (Sept., 2009)
  825--856}, \href{http://arxiv.org/abs/0903.2084}{{\tt arXiv:0903.2084
  [hep-th]}}.

\bibitem{Dijkgraaf:2010ur}
R.~Dijkgraaf, H.~Fuji, and M.~Manabe, ``{The Volume Conjecture, Perturbative
  Knot Invariants, and Recursion Relations for Topological Strings},''
  \href{http://dx.doi.org/10.1016/j.nuclphysb.2011.03.014}{{\em Nucl.Phys.}
  {\bf B849} (2011)  166--211},
\href{http://arxiv.org/abs/1010.4542}{{\tt arXiv:1010.4542 [hep-th]}}.

\bibitem{Gukov:2011qp}
S.~Gukov and P.~Sulkowski, ``{A-polynomial, B-model, and Quantization},''
  \href{http://dx.doi.org/10.1007/JHEP02(2012)070}{{\em JHEP} {\bf 1202} (2012)
   070},
\href{http://arxiv.org/abs/1108.0002}{{\tt arXiv:1108.0002 [hep-th]}}.

\bibitem{Borot:2012cw}
G.~Borot and B.~Eynard, ``{All-order asymptotics of hyperbolic knot invariants
  from non-perturbative topological recursion of A-polynomials},''
\href{http://arxiv.org/abs/1205.2261}{{\tt arXiv:1205.2261 [math-ph]}}.

\bibitem{Fuji:2012nx}
H.~Fuji, S.~Gukov, and P.~Sulkowski, ``{Super-A-polynomial for knots and BPS
  states},'' \href{http://dx.doi.org/10.1016/j.nuclphysb.2012.10.005}{{\em
  Nucl.Phys.} {\bf B867} (2013)  506--546},
\href{http://arxiv.org/abs/1205.1515}{{\tt arXiv:1205.1515 [hep-th]}}.

\bibitem{marino}
A.~Brini, B.~Eynard, and M.~Marino, ``{Torus knots and mirror symmetry},''
\href{http://arxiv.org/abs/1105.2012}{{\tt arXiv:1105.2012 [hep-th]}}.

\bibitem{Jockers:2012pz}
H.~Jockers, A.~Klemm, and M.~Soroush, ``{Torus Knots and the Topological
  Vertex},''
\href{http://arxiv.org/abs/1212.0321}{{\tt arXiv:1212.0321 [hep-th]}}.

\bibitem{Witten:1988hf}
E.~Witten, ``{Quantum field theory and the Jones polynomial},''
\href{http://dx.doi.org/10.1007/BF01217730}{{\em Commun. Math. Phys.} {\bf 121}
  (1989)  351}.

\bibitem{HOMFLY}
P.~Freyd, D.~Yetter, J.~Hoste, W.~Lickorish, K.~Millet, and A.~Ocneanu, ``A new
  polynomial invariant of knots and links,''
\href{http://dx.doi.org/%10.1016/0550-3213(89)90232-0}{{\em Bull. AMS} {\bf 12}
  (1985)  239--246}.

\bibitem{Witten:1992fb}
E.~Witten, ``{Chern-Simons gauge theory as a string theory},'' {\em Prog.Math.}
  {\bf 133} (1995)  637--678,
\href{http://arxiv.org/abs/hep-th/9207094}{{\tt arXiv:hep-th/9207094
  [hep-th]}}.

\bibitem{OV}
H.~Ooguri and C.~Vafa, ``{Knot invariants and topological strings},''
  \href{http://dx.doi.org/10.1016/S0550-3213(00)00118-8}{{\em Nucl.Phys.} {\bf
  B577} (2000)  419--438},
\href{http://arxiv.org/abs/hep-th/9912123}{{\tt arXiv:hep-th/9912123
  [hep-th]}}.

\bibitem{GV}
R.~Gopakumar and C.~Vafa, ``{On the gauge theory/geometry correspondence},''
  {\em Adv. Theor. Math. Phys.} {\bf 3} (1999)  1415--1443,
\href{http://arxiv.org/abs/hep-th/9811131}{{\tt arXiv:hep-th/9811131}}.

\bibitem{HV}
K.~Hori and C.~Vafa, ``Mirror symmetry,''
\href{http://arxiv.org/abs/hep-th/0002222}{{\tt arXiv:hep-th/0002222
  [hep-th]}}.

\bibitem{Dimofte2}
T.~Dimofte, M.~Gabella, and A.~B. Goncharov, ``{K-Decompositions and 3d Gauge
  Theories},''
\href{http://arxiv.org/abs/1301.0192}{{\tt arXiv:1301.0192 [hep-th]}}.

\bibitem{LMV}
J.~Labastida, M.~Marino, and C.~Vafa, ``{Knots, links and branes at large N},''
  {\em JHEP} {\bf 0011} (2000)  007,
\href{http://arxiv.org/abs/hep-th/0010102}{{\tt arXiv:hep-th/0010102
  [hep-th]}}.

\bibitem{Eynard:2008he}
B.~Eynard and M.~Marino, ``{A Holomorphic and background independent partition
  function for matrix models and topological strings},''
  \href{http://dx.doi.org/10.1016/j.geomphys.2010.11.012}{{\em J.Geom.Phys.}
  {\bf 61} (2011)  1181--1202},
\href{http://arxiv.org/abs/0810.4273}{{\tt arXiv:0810.4273 [hep-th]}}.

\bibitem{Marino:2009dp}
M.~Marino, S.~Pasquetti, and P.~Putrov, ``{Large N duality beyond the genus
  expansion},'' \href{http://dx.doi.org/10.1007/JHEP07(2010)074}{{\em JHEP}
  {\bf 1007} (2010)  074},
\href{http://arxiv.org/abs/0911.4692}{{\tt arXiv:0911.4692 [hep-th]}}.

\bibitem{Beem:2012mb}
C.~Beem, T.~Dimofte, and S.~Pasquetti, ``{Holomorphic Blocks in Three
  Dimensions},''
\href{http://arxiv.org/abs/1211.1986}{{\tt arXiv:1211.1986 [hep-th]}}.

\bibitem{Garoufalidis}
S.~Garoufalidis, ``{The colored HOMFLY polynomial is q-holonomic},''
\href{http://arxiv.org/abs/1211.6388}{{\tt arXiv:1211.6388 [math.GT]}}.

\bibitem{kashiwara}
M.~Kashiwara, ``Micro-Local Calculus,'' {\em International Symposium on
  Mathematical Problems in Theoretical Physics, Lecture Notes in Physics Volume
  39} (1975)  30--37.

\bibitem{miwa}
M.~J. T.~Miwa, T.~Oshima, ``Introduction to Microlocal Analysis,'' {\em Publ.
  RIMS, Kyoto Univ. 12 Suppl., Lecture Notes in Physics Volume 39} (1974)
  267--300.

\bibitem{kashiwarad}
M.~Kashiwara, ``D-modules and Microlocal Calculus,'' {\em Volume 217 of
  Translations of mathematical monographs} (2003)  .

\bibitem{Witten:1982im}
E.~Witten, ``{Supersymmetry and Morse theory},''
{\em J.Diff.Geom.} {\bf 17} (1982)  661--692.

\bibitem{Ngsurvey}
L.~Ng, ``A topological introduction to knot contact homology,''
  \href{http://arxiv.org/abs/1210.4803}{{\tt arXiv:1210.4803}}.

\bibitem{EGH}
Y.~Eliashberg, A.~Givental, and H.~Hofer, ``Introduction to symplectic field
  theory,'' {\em Geom. Funct. Anal.} (2000) no.~Special Volume, Part II,
  560--673, \href{http://arxiv.org/abs/math.SG/0010059}{{\tt
  arXiv:math.SG/0010059}}. GAFA 2000 (Tel Aviv, 1999).

\bibitem{EES}
T.~Ekholm, J.~Etnyre, and M.~Sullivan, ``Legendrian contact homology in
  {$P\times\Bbb R$},''
  \href{http://dx.doi.org/10.1090/S0002-9947-07-04337-1}{{\em Trans. Amer.
  Math. Soc.} {\bf 359} (2007) no.~7, 3301--3335 (electronic)},
  \href{http://arxiv.org/abs/math/0505451}{{\tt arXiv:math/0505451}}.

\bibitem{Ekholm_rsft}
T.~Ekholm, ``Rational symplectic field theory over {$\Bbb Z\sb 2$} for exact
  {L}agrangian cobordisms,'' {\em J. Eur. Math. Soc. (JEMS)} {\bf 10} (2008)
  no.~3, 641--704, \href{http://arxiv.org/abs/math/0612029}{{\tt
  arXiv:math/0612029}}.

\bibitem{Ngframed}
L.~Ng, ``Framed knot contact homology,'' {\em Duke Math. J.} {\bf 141} (2008)
  no.~2, 365--406, \href{http://arxiv.org/abs/math/0407071}{{\tt
  arXiv:math/0407071}}.

\bibitem{EENStransverse}
T.~Ekholm, J.~Etnyre, L.~Ng, and M.~Sullivan, ``Filtrations on the knot contact
  homology of transverse knots,'' {\em Math. Annalen} (2012)  ,
  \href{http://arxiv.org/abs/1010.0450}{{\tt arXiv:1010.0450}}. To appear.

\bibitem{Ngtransverse}
L.~Ng, ``Combinatorial knot contact homology and transverse knots,'' {\em Adv.
  Math.} {\bf 227} (2011) no.~6, 2189--2219,
  \href{http://arxiv.org/abs/1010.0451}{{\tt arXiv:1010.0451}}.

\bibitem{EENS}
T.~Ekholm, J.~Etnyre, L.~Ng, and M.~Sullivan, ``Knot contact homology,''
  \href{http://arxiv.org/abs/1109.1542}{{\tt arXiv:1109.1542}}.

\bibitem{Tillmann}
S.~Tillmann, ``Boundary slopes and the logarithmic limit set,''
  \href{http://dx.doi.org/10.1016/j.top.2004.07.002}{{\em Topology} {\bf 44}
  (2005) no.~1, 203--216}. \url{http://dx.doi.org/10.1016/j.top.2004.07.002}.

\bibitem{Mishachev}
K.~Mishachev, ``The {$N$}-copy of a topologically trivial {L}egendrian knot,''
  {\em J. Symplectic Geom.} {\bf 1} (2003) no.~4, 659--682.

\bibitem{FO3}
K.~Fukaya, Y.-G. Oh, H.~Ohta, and K.~Ono, {\em Lagrangian intersection {F}loer
  theory: anomaly and obstruction. {P}art {I}}, vol.~46 of {\em AMS/IP Studies
  in Advanced Mathematics}.
\newblock American Mathematical Society, Providence, RI, 2009.

\bibitem{BEHWZ}
F.~Bourgeois, Y.~Eliashberg, H.~Hofer, K.~Wysocki, and E.~Zehnder,
  ``Compactness results in symplectic field theory,'' {\em Geom. Topol.} {\bf
  7} (2003)  799--888, \href{http://arxiv.org/abs/math.SG/0308183}{{\tt
  arXiv:math.SG/0308183}}.

\bibitem{Hofer}
H.~Hofer, ``A general {F}redholm theory and applications,'' in {\em Current
  developments in mathematics, 2004}, pp.~1--71.
\newblock Int. Press, Somerville, MA, 2006.
\newblock \href{http://arxiv.org/abs/math/0509366}{{\tt arXiv:math/0509366}}.

\bibitem{E}
T.~Ekholm, ``Morse flow trees and {L}egendrian contact homology in 1-jet
  spaces,'' {\em Geom. Topol.} {\bf 11} (2007)  1083--1224,
  \href{http://arxiv.org/abs/math.SG/0509386}{{\tt arXiv:math.SG/0509386}}.

\bibitem{CEL}
K.~Cieliebak, T.~Ekholm, and J.~Latschev, ``Compactness for holomorphic curves
  with switching {L}agrangian boundary conditions,'' {\em J. Symplectic Geom.}
  {\bf 8} (2010) no.~3, 267--298, \href{http://arxiv.org/abs/0903.2200}{{\tt
  arXiv:0903.2200}}.

\bibitem{CELN}
K.~Cieliebak, T.~Ekholm, J.~Latschev, and L.~Ng. In preparation.

\bibitem{Cornwell}
C.~Cornwell, ``Knot contact homology and representations of knot groups,''
  \href{http://arxiv.org/abs/1303.4943}{{\tt arXiv:1303.4943}}.

\bibitem{morozov}
H.~Itoyama, A.~Mironov, A.~Morozov, and A.~Morozov, ``{HOMFLY and
  superpolynomials for figure eight knot in all symmetric and antisymmetric
  representations},'' \href{http://dx.doi.org/10.1007/JHEP07(2012)131}{{\em
  JHEP} {\bf 1207} (2012)  131},
\href{http://arxiv.org/abs/1203.5978}{{\tt arXiv:1203.5978 [hep-th]}}.

\bibitem{kawagoe}
K.~Kawagoe, ``On the formulae for the colored HOMFLY polynomials,''
  \href{http://arxiv.org/abs/1210.7574}{{\tt arXiv:1210.7574 [math.GT]}}.

\bibitem{ramadevi1}
P.~Ramadevi and Zodinmawia, ``{Reformulated invariants for non-torus knots and
  links},''
\href{http://arxiv.org/abs/1209.1346}{{\tt arXiv:1209.1346 [hep-th]}}.

\bibitem{Zhu:2012tm}
S.~Zhu, ``{Colored HOMFLY polynomial via skein theory},''
\href{http://arxiv.org/abs/1206.5886}{{\tt arXiv:1206.5886 [math.GT]}}.

\bibitem{Morozov1}
A.~Mironov, A.~Morozov, and A.~Morozov, ``{Character expansion for HOMFLY
  polynomials. I. Integrability and difference equations},'' {\em Strings,
  Gauge Fields, and the Geometry Behind: The Legacy of Maximilian Kreuzer,
  edited by A. Rebhan, L. Katzarkov, J. Knapp, R. Rashkov, E.} {\bf
  Scheidegger} (World Scietific Publishins Co.Pte.Ltd. 2013)  pp.101--118,
\href{http://arxiv.org/abs/1112.5754}{{\tt arXiv:1112.5754 [hep-th]}}.

\bibitem{Morozov2}
A.~Anokhina, A.~Mironov, A.~Morozov, A.~Morozov, A.~Mironov, {\em et al.},
  ``{Racah coefficients and extended HOMFLY polynomials for all 5-, 6- and
  7-strand braids},''
  \href{http://dx.doi.org/10.1016/j.nuclphysb.2012.11.006}{{\em Nucl.Phys.}
  {\bf B868} (2013)  271--313},
\href{http://arxiv.org/abs/1207.0279}{{\tt arXiv:1207.0279 [hep-th]}}.

\bibitem{ramadevi2}
S.~Nawata, P.~Ramadevi, and Zodinmawia, ``{Multiplicity-free quantum 6j-symbols
  for $U_q(sl_N)$},''
\href{http://arxiv.org/abs/1302.5143}{{\tt arXiv:1302.5143 [hep-th]}}.

\bibitem{ramadevi3}
S.~Nawata, P.~Ramadevi, and Zodinmawia, ``{Colored HOMFLY polynomials from
  Chern-Simons theory},''
\href{http://arxiv.org/abs/1302.5144}{{\tt arXiv:1302.5144 [hep-th]}}.

\bibitem{Morozov3}
A.~Anokhina, A.~Mironov, A.~Morozov, and A.~Morozov, ``{Colored HOMFLY
  polynomials as multiple sums over paths or standard Young tableaux},''
\href{http://arxiv.org/abs/1304.1486}{{\tt arXiv:1304.1486 [hep-th]}}.

\bibitem{IH}
M.~Aganagic, R.~Dijkgraaf, A.~Klemm, M.~Marino, and C.~Vafa, ``Topological
  strings and integrable hierarchies,''
  \href{http://dx.doi.org/10.1007/s00220-005-1448-9}{{\em Commun.Math.Phys.}
  {\bf 261} (2006)  451--516},
\href{http://arxiv.org/abs/hep-th/0312085}{{\tt arXiv:hep-th/0312085
  [hep-th]}}.

\bibitem{DHSV}
R.~Dijkgraaf, L.~Hollands, P.~Sulkowski, and C.~Vafa, ``{Supersymmetric gauge
  theories, intersecting branes and free fermions},''
  \href{http://dx.doi.org/10.1088/1126-6708/2008/02/106}{{\em JHEP} {\bf 0802}
  (2008)  106},
\href{http://arxiv.org/abs/0709.4446}{{\tt arXiv:0709.4446 [hep-th]}}.

\bibitem{Dijkgraaf2}
R.~Dijkgraaf, L.~Hollands, and P.~Sulkowski, ``{Quantum Curves and
  D-Modules},'' \href{http://dx.doi.org/10.1088/1126-6708/2009/11/047}{{\em
  JHEP} {\bf 0911} (2009)  047},
\href{http://arxiv.org/abs/0810.4157}{{\tt arXiv:0810.4157 [hep-th]}}.

\bibitem{Dijkgraaf3}
R.~Dijkgraaf, ``{D-branes and D-modules},''
\href{http://dx.doi.org/10.1143/PTPS.177.12}{{\em Prog.Theor.Phys.Suppl.} {\bf
  177} (2009)  12--32}.

\bibitem{vstring}
C.~Vafa, ``{Supersymmetric Partition Functions and a String Theory in 4
  Dimensions},''
\href{http://arxiv.org/abs/1209.2425}{{\tt arXiv:1209.2425 [hep-th]}}.

\bibitem{Kapustin-Orlov}
A.~Kapustin and D.~Orlov, ``{Vertex algebras, mirror symmetry, and D-branes:
  The Case of complex tori},''
  \href{http://dx.doi.org/10.1007/s00220-002-0755-7}{{\em Commun.Math.Phys.}
  {\bf 233} (2003)  79--136},
\href{http://arxiv.org/abs/hep-th/0010293}{{\tt arXiv:hep-th/0010293
  [hep-th]}}.

\bibitem{Kapustin-Witten}
A.~Kapustin and E.~Witten, ``{Electric-Magnetic Duality And The Geometric
  Langlands Program},'' {\em Commun.Num.Theor.Phys.} {\bf 1} (2007)  1--236,
\href{http://arxiv.org/abs/hep-th/0604151}{{\tt arXiv:hep-th/0604151
  [hep-th]}}.

\bibitem{Gukov-Witten}
S.~Gukov and E.~Witten, ``{Branes and Quantization},'' {\em
  Adv.Theor.Math.Phys.} {\bf 13} (2009)  ,
\href{http://arxiv.org/abs/0809.0305}{{\tt arXiv:0809.0305 [hep-th]}}.

\bibitem{GL}
S.~Garoufalidis and T.~T.~Q. L{\^e}, ``The colored {J}ones function is
  {$q$}-holonomic,'' \href{http://dx.doi.org/10.2140/gt.2005.9.1253}{{\em Geom.
  Topol.} {\bf 9} (2005)  1253--1293 (electronic)}.
  \url{http://dx.doi.org/10.2140/gt.2005.9.1253}.

\end{thebibliography}\endgroup
\end{document}